\documentclass[11pt, DIV=15]{scrarticle}
\usepackage[utf8]{luainputenc}
\usepackage[T1]{fontenc}
\usepackage[UKenglish]{babel}
\usepackage{amsmath, amsfonts, amsthm, amssymb, mathtools, stmaryrd}

\usepackage[hyphens]{url}
\usepackage{breakurl}
\usepackage{tikz}
\usepackage{tikz-cd}
\usepackage{tkz-graph}
\usetikzlibrary{arrows,backgrounds,positioning,fit,cd,babel}

\usepackage[pdfusetitle,hidelinks]{hyperref}
\usepackage[capitalise,noabbrev]{cleveref}
\usepackage{csquotes}
\usepackage[osf,sc]{mathpazo}
\addtokomafont{disposition}{\normalfont\bfseries} % \scshape \rmfamily
\usepackage{relsize}
\usepackage{microtype}

\usepackage{thm-restate}
\usepackage{subcaption}
\usepackage{enumitem}
\usepackage{comment}
\usepackage{orcidlink}

\usepackage{todonotes}

\recalctypearea

\usepackage[style=numeric-comp,maxbibnames=99,backend=biber,bibencoding=utf8,sorting=nty]{biblatex} % ,sorting=none,backref=true

\usepackage{orcidlink}
\usepackage{booktabs}

\AtEveryBibitem{
	\clearfield{urldate}
	\clearfield{urlyear}	
	\clearfield{urlmonth}
	\clearfield{issn}	
	\clearfield{month}
	\clearfield{day}		
	\clearfield{language}
	\clearfield{Language}
	\clearfield{langid}
	\clearfield{note}
}
\DeclareSourcemap{
	\maps[datatype=bibtex]{
		\map[overwrite]{
			\step[fieldsource=doi, final]
			\step[fieldset=url, null]
			\step[fieldset=eprint, null]
			\step[fieldset=isbn, null]
		}  
	}
}

\theoremstyle{definition}
\newtheorem{definition}{Definition}[section]

\newtheorem{example}[definition]{Example}
\newtheorem{remark}[definition]{Remark}

\theoremstyle{plain}
\newtheorem{lemma}[definition]{Lemma}
\newtheorem{corollary}[definition]{Corollary}
\newtheorem{theorem}[definition]{Theorem}
\newtheorem{conjecture}[definition]{Conjecture}

\newtheorem{proposition}[definition]{Proposition}
\newtheorem{observation}[definition]{Observation}
\newtheorem{fact}[definition]{Fact}
\Crefname{fact}{Fact}{Facts}
\newtheorem{proviso}[definition]{Proviso}

\theoremstyle{remark}
\newtheorem{claim}{Claim}[definition]

\Crefname{claim}{Claim}{Claims}
\newenvironment{claimproof}[1][Proof of Claim]{\begin{proof}[#1] }{ \end{proof}}
\setlist[enumerate, 1]{font=\upshape, noitemsep, nolistsep}
\setlist[enumerate, 2]{font=\upshape, noitemsep, nolistsep}
\setlist[itemize, 1]{noitemsep, nolistsep,font=\upshape}
\setlist[itemize, 2]{noitemsep, nolistsep,font=\upshape}

\DeclarePairedDelimiter{\norm}{\lVert}{\rVert}
\DeclarePairedDelimiter{\abs}{\lvert}{\rvert}

\newcommand{\Cc}{{\cal C}}

\newcommand{\Ii}{{\cal I}}

\newcommand{\Xx}{{\cal X}}

\DeclareMathOperator{\tw}{tw}
\DeclareMathOperator{\hdtw}{hdtw}
\DeclareMathOperator{\mn}{mn}
\DeclareMathOperator{\sub}{sub}
\DeclareMathOperator{\emb}{emb}
\DeclareMathOperator{\perm}{perm}

\DeclareMathOperator{\vc}{vc}

\DeclareMathOperator{\cl}{cl}
\DeclareMathOperator{\id}{id}
\DeclareMathOperator{\cc}{lvc}

\tikzset{
	vertex/.style={draw,circle,fill=gray},
	every node/.style={anchor=center},
	lbl/.style={color=lightgray}
}

\title{Symmetric Algebraic Circuits and \\ Homomorphism Polynomials\footnote{This work was presented at the \emph{17th Innovations in Theoretical Computer Science Conference} (\textsmaller{ITCS} 2026) \cite{conference_version}.}}
\author{Anuj Dawar \orcidlink{0000-0003-4014-8248} \and Benedikt Pago \orcidlink{0000-0001-6377-1230} \and Tim Seppelt \orcidlink{0000-0002-6447-0568}}
\renewcommand{\phi}{\varphi}
\renewcommand{\epsilon}{\varepsilon}

\newcommand{\Aut}{\mathbf{Aut}}
\newcommand{\Sym}{\mathbf{Sym}}
\newcommand{\Alt}{\mathbf{Alt}}
\newcommand{\Stab}{\mathbf{Stab}}
\newcommand{\StabP}{\Stab^{\bullet}}
\newcommand{\Orb}{\mathbf{Orb}}

\DeclareMathOperator{\child}{children}

\newcommand{\bbN}{\mathbb{N}}
\newcommand{\bbQ}{\mathbb{Q}}
\newcommand{\bbF}{\mathbb{F}}

\definecolor{lightgray}{rgb}{0.60, 0.60, 0.61} % lipicsLightGray
\definecolor{gray}{rgb}{0.31, 0.31, 0.33} % lipicsBulletGray
\definecolor{yellow}{rgb}{0.99, 0.78, 0.07}

\usepackage{xargs}   
\newcommandx{\tim}[2][1=]{\todo[author=Tim,color=yellow,#1]{#2}}
\newcommandx{\benedikt}[2][1=]{\todo[author=Benedikt,color=orange,#1]{#2}}

\let\sup\relax
\DeclareMathOperator{\sup}{sup}

\DeclareMathOperator{\imm}{imm}
\DeclareMathOperator{\sgn}{sgn}
\DeclareMathOperator{\Match}{\# Match}
\DeclareMathOperator{\maxOrb}{maxOrb}
\DeclareMathOperator{\maxSup}{maxSup}

\newcommand{\VP}{\mathsf{VP}}
\newcommand{\VNP}{\mathsf{VNP}}
\newcommand{\VFPT}{\mathsf{VFPT}}
\newcommand{\VW}{\mathsf{VW}[1]}

\newcommand{\FPT}{\mathsf{FPT}}
\newcommand{\sharpW}{\#\mathsf{W}[1]}

\begin{document}
	\maketitle

\begin{abstract}
	The central open question of algebraic complexity is whether $\VP \neq \VNP$, which is saying that the permanent cannot be represented by families of polynomial-size algebraic circuits. 
	For symmetric algebraic circuits, this has been confirmed by Dawar and Wilsenach (2020, 2025), who showed exponential lower bounds on the size of symmetric circuits for the permanent. In this work, we set out to develop a more general symmetric algebraic complexity theory. Our main result is that a family of symmetric polynomials admits small symmetric circuits if and only if they can be written as a linear combination of homomorphism counting polynomials of graphs of bounded treewidth. We also establish a relationship between the symmetric complexity of subgraph counting polynomials and the vertex cover number of the pattern graph. As a concrete example, we examine the symmetric complexity of immanant families (a generalisation of the determinant and permanent) and show that a known conditional dichotomy due to Curticapean (2021) holds unconditionally in the symmetric setting. 
\end{abstract}

\section{Introduction}	
The study of \emph{algebraic circuit complexity} (also called \emph{arithmetic circuit complexity}) aims to understand the power of circuits to succinctly express (or compute) polynomials.  In short, we are interested in establishing how many operations of addition and multiplication are needed in a circuit that computes a polynomial $p \in \bbF[\Xx]$, for some field $\bbF$ and set of variables $\Xx$.  We are usually interested in how this complexity grows with $n$ for a family of polynomials $(p_n)_{n \in \bbN}$.  The central conjecture in the field (known as $\VP \neq \VNP$) is that there are no polynomial-size circuits for the \emph{permanent} in the way that there are for the \emph{determinant}.

Both the determinant and permanent are examples of polynomials over matrices.  That is to say that we can treat the set of variables $\Xx$ as the entries $x_{ij}$ of a square matrix.  Moreover, the permanent is \emph{symmetric} in the sense that permuting the rows and columns of the matrix does not change the polynomial.  The determinant has a smaller group of symmetries, being invariant under a permutation applied \emph{simultaneously} to the rows and columns.  This motivates the study of \emph{symmetric algebraic circuits} introduced in~\cite{dawar_symmetric_2020,dawar_symmetric_2025}, which are circuits where symmetries of the polynomial computed are reflected in the automorphisms of the circuits.  
In \cite{dawar_symmetric_2020,dawar_symmetric_2025}, an exponential gap between the symmetric circuit complexity of the determinant and the permanent was established:
There are polynomial-size symmetric circuits (in the sense of
invariance under simultaneous row and column permutations) that
compute the determinant but any family of symmetric circuits computing
the permanent is of exponential size.  It is the latter, exponential
lower bound, that is the main technical achievement of that
paper. Notably, this result also opens up a new approach to the
question of separating $\VNP$ from $\VP$ and challenges us to find the minimum group of symmetries for a polynomial for which we can prove unconditional lower bounds.
Thus, the study of symmetric circuits is motivated by the possibility of proving strong lower bounds which are presently out of reach without the assumption of symmetry.

Among the restrictions of circuit models that have been extensively studied in the literature, symmetry is arguably one of the newest and most interesting ones: Symmetric circuits not only admit provable lower bounds but are still surprisingly powerful, in that they can efficiently compute the determinant. 
By contrast, for example the restricted model of \emph{multilinear
  formulas} requires super-polynomial size for both determinant and
permanent \cite{Raz04}, and the same is true for \emph{bounded-depth
  multilinear circuits} \cite{RazY08}. Also for the well-studied class
of \emph{monotone} circuits, we have exponential lower bounds for the
permanent \cite{JerrumS82}.  This does not apply to the determiinant
as it is not a monotone function, but with a suitably adapted
definition of monotonicity the same methods also yield exponential
lower bounds for the determinant~\cite{KayalS14}.

In the present paper, we approach the study of symmetric algebraic circuits more systematically than has been done before. We develop a framework that links the expressive power of such circuits to the theory of graph \emph{homomorphism counting}. Using this, we are able to completely characterise -- within a large and interesting class of polynomials -- which polynomials admit efficient symmetric circuits and which ones do not. This characterisation is the main result of the paper.

To be more precise, the polynomials we consider are over (not necessarily square) matrices of variables, have rational coefficients and are invariant under arbitrary permutations of the rows and columns.  These are families $(p_{n,m})_{n,m \in \bbN}$ of polynomials where $p_{n,m} \in \bbQ[\Xx_{n,m}]$, and $\Xx_{n,m}$ is the set of variables $\{x_{ij} \mid i \in [n], j \in [m]\}$.  The invariance condition we require is that for any pair of permutations $\pi \in \Sym_n$, $\sigma \in \Sym_m$, if $p_{n,m}^{(\pi,\sigma)}$ denotes the polynomial obtained from $p_{n,m}$ by replacing every occurrence of $x_{ij}$ with $x_{\pi(i)\sigma(j)}$, then $p_{n,m} = p_{n,m}^{(\pi,\sigma)}$. 
Examples of such families are the permanent, for $m = n$, or more generally, for $m \geq n$, the \emph{rectangular} permanent. The latter has been used recently in a lower bound proof for bounded-depth algebraic circuits \cite{forbes2024low}.

The reason for focussing on this class of \emph{matrix-symmetric} polynomials is that they have a very natural semantic interpretation: Every $\Sym_n \times \Sym_m$-symmetric $p_{n,m} \in \bbQ[\Xx_{n,m}]$ describes a function from weighted bipartite undirected graphs $G$ with $(n,m)$ vertices (i.e.\ $G$ has a fixed bipartition $A \uplus B$ with $|A| = n$ and $|B| = m$) to $\bbQ$. The resulting value $p_{n,m}(G) \in \bbQ$ is the evaluation of $p_{n,m}$ in the bi-adjacency matrix of $G$, whose rows are indexed with the vertices in $A$, and columns with the vertices in $B$. The symmetry of $p_{n,m}$ ensures that the value $p_{n,m}(G)$ does not depend on the row and column ordering of the bi-adjacency matrix, and hence the $\bbQ$-valued function expressed by $p_{n,m}$ is in fact an isomorphism-invariant graph parameter. 
% For example, the graph parameter described by the permanent family is the number of perfect matchings. 
% A family $(p_{n,m})_{n,m \in \bbN}$ can then be seen as describing a graph parameter on \emph{all} bipartite graphs in a \emph{non-uniform} way, one for every size $(n,m)$ of the bipartition. 
% I think the uniformity vs. non-uniformity thing doesn't become clear. We might as well drop it.

It follows from a classical result of \textcite{lovasz_operations_1967} that the isomorphism-invariant $\mathbb{Q}$-valued functions on $n$-vertex graphs are precisely the linear combinations of homomorphism counting functions $\sum \alpha_F \hom(F, -)$ for finitely many graphs $F$ and rational coefficients $\alpha_F$. 
In other words,
the isomorphism-invariant $\mathbb{Q}$-valued functions on graphs are precisely the \emph{non-uniform graph motif parameters}, i.e.\ sequences $(\sum \alpha_{F, n} \hom(F, -))_{n \in \mathbb{N}}$ of linear combinations of homomorphism counts, one for every input graph size $n$.
In an influential paper \cite{curticapean_homomorphisms_2017}, Curticapean, Dell, and Marx 
studied functions which can be uniformly expressed in this way.
Their \emph{uniform graph motif parameters} are the functions of the form $\sum \alpha_F \hom(F, -)$ whose domain are all graphs regardless of their size.
They show that the fixed-parameter counting complexity of such uniform graph motif parameters is governed by the treewidth of the patterns $F$ with non-zero coefficients~$\alpha_F$.

\paragraph{Characterising matrix-symmetric polynomials admitting small symmetric circuits.}
We broaden the scope of \cite{curticapean_homomorphisms_2017} twofold by precisely characterising the symmetric circuit complexity of all, i.e.\ not necessarily uniform, graph motif parameters.
To that end, we observe that the matrix-symmetric polynomials $p \in \mathbb{Q}[\mathcal{X}_{n,m}]$ are precisely those that can be written as linear combinations of homomorphism polynomials, cf.\ \cref{lem:gnm-sym}:
For a bipartite multigraph   $F$  with bipartition $A \uplus B$, the \emph{homomorphism polynomial} $\hom_{F,n,m} \in \bbQ[\Xx_{n,m}]$ is defined as
\[
\hom_{F,n,m} \coloneqq \sum_{h \colon A \uplus B \to [n] \uplus [m]} \prod_{ab \in E(F)} x_{h(a)h(b)}.
\]
The name is justified by the fact that $\hom_{F,n,m}$ evaluates to the number of homomorphisms from $F$ to a $(n,m)$-vertex bipartite graph $G$ when substituting the bi-adjacency matrix  of $G$ for the matrix of variables $\Xx_{n,m}$.

Our main result asserts that the symmetric complexity of a family $(p_{n,m})$ of matrix-symmetric polynomials is completely governed by the treewidth of the graphs whose homomorphism polynomials arise in linear expansions\footnote{In stark contrast to uniform graph motif parameters studied by \textcite{curticapean_homomorphisms_2017}, non-uniform graph motif parameters generally do not admit a unique expansion as linear combinations of homomorphism counts.} of $(p_{n,m})$.
By symmetric complexity we mean the smallest possible \emph{orbit size} of a family $(C_{n,m})_{n,m \in \bbN}$ of symmetric circuits representing $(p_{n,m})_{n,m \in \bbN}$. This is the number of gates of $C_{n,m}$ in
the largest orbit of the action of $\Sym_n \times \Sym_m$ on $C_{n,m}$ (see \cref{sec:shortPreliminaries}). 
Note that lower bounds on orbit size imply lower bounds on the total number of gates in a circuit but not vice versa. 
We also show that orbit size, not total size, is the correct measure that allows us to build our theory, cf.\ \cref{thm:hardPolynomialsInT}.

To state our main result, let $\mathfrak{T}_{n,m}^k$ denote, for every $k,n,m \in \bbN$, the collection of all polynomials that can be obtained as linear combinations of homomorphism polynomials $\hom_{F,n,m}$ for graphs $F$ of treewidth less than $k$.
\begin{theorem}[restate=thmMain,label=thm:main1,name=]
	For every family of polynomials $p_{n,m} \in \mathbb{Q}[\mathcal{X}_{n,m}]$, the following are equivalent:
	\begin{enumerate}
		\item there exists a constant $k \in \mathbb{N}$ such that $p_{n,m} \in \mathfrak{T}_{n,m}^k$ for all $n,m \in \mathbb{N}$,\label{it:main1}
		%			\item $p$ has counting width at most $k-1$ on $(n,m)$-vertex edge-coloured graphs,\label{it:main2}
		\item the $p_{n,m}$ admit $\Sym_n \times \Sym_m$-symmetric circuits of orbit size polynomial in $n+m$.\label{it:main2}
	\end{enumerate}
\end{theorem}

It is instructive to consider again the permanent $\perm_n = \sum_{\pi \in \Sym_n} \prod_{i \in [n]} x_{i\pi(i)}$. 
This is not a homomorphism polynomial
and not even a uniform graph motif parameter
 but it can be seen as counting, given a bipartite graph $G$ instantiated as a bi-adjacency matrix $\Xx_{n,n}$, the number of subgraphs isomorphic to an $n$-matching. 
By M\"obius inversion, cf.\ \cref{thm:sub-hom}, subgraph counts can be written as a linear combination of homomorphism counts.
One consequence of \cref{thm:main1} together with the exponential lower bound on symmetric circuits for the permanent \cite{dawar_symmetric_2025} is that $\perm_n$ cannot be expressed as a linear combination of homomorphism polynomials of bounded treewidth.  
Alternatively, we can understand one direction of \cref{thm:main1} as a vast (in fact, as vast as possible) generalisation of the lower bound on the permanent. 
Indeed, the theorem captures \emph{all} super-polynomial lower bounds on the orbit size of symmetric circuits for matrix-symmetric polynomials.

\paragraph{Lower bounds via counting width.}
\cref{thm:main1} stipulates that, to derive a super-polynomial lower bound for some family $p_{n,m}$, one must show that there is no constant $k \in \bbN$ such that $p_{n,m} \in \mathfrak{T}^k_{n,m}$ for all $n,m \in \bbN$. 
This, however, seems to be difficult with the presently available techniques. 
Therefore, our next goal is to give more usable characterisations of $\mathfrak{T}^k_{n,m}$
via a descriptive complexity measure for graph parameters, namely in terms of counting width. 
The \emph{counting width} \cite{dawar_definability_2017,dawar_notions_2025} of a graph parameter $p$ is defined as the smallest $k \in \bbN$ such that whenever two graphs $G$, $H$ are indistinguishable by the $k$-dimensional Weisfeiler-Leman algorithm, then $p(G) = p(H)$, see also \cref{def:counting-width}.  
The Weisfeiler-Leman indistinguishability on graphs is a standard relaxation of graph
isomorphism (see~\cite{CFI,kiefer2020WL}), and coincides with equivalence in the $(k+1)$-variable fragment of first-order logic with counting quantifiers.

Counting width is well-studied in finite model theory and affords lower bounds, usually via the so-called Cai-F\"urer-Immerman construction \cite{CFI}.
For Boolean
circuits, a tight relationship
between the counting width and orbit size of symmetric
circuits is known~\cite{anderson_symmetric_2017}.

The aforementioned exponential lower bound for symmetric circuits for the permanent was in fact established via a counting width lower bound for the number of perfect matchings. In general, unbounded counting width implies super-polynomial symmetric circuit lower bounds (via techniques for example from \cite{anderson_symmetric_2017,  dawar_symmetric_2025}), so by Theorem \ref{thm:main1}, polynomials with unbounded counting width cannot be in $\mathfrak{T}_{n,m}^k$, for any fixed $k$.
The question is, does the converse hold: Is bounded counting width a sufficient criterion for the existence of orbit-small symmetric circuits?
We answer this question positively, at least in certain restricted, but interesting cases, cf.\ \cref{fig:overview}:
\begin{enumerate}
	\item\label{it:hom} when the $p_{n,m}$ are single homomorphism polynomials (rather than linear combinations of such), cf.\ \cref{thm:dichotomy-chromatic},
	\item\label{it:sub} when the $p_{n,m}$ are single subgraph polynomials for graphs of sublinear size, i.e.\ counting  the number of subgraphs isomorphic to the pattern, cf.\ \cref{thm:sublinear-intro}, and
	\item\label{it:lincomb} when the $p_{n,m}$ are linear combinations of homomorphism polynomials of sublinear size, cf. \ \cref{thm:lincomb-intro}.
\end{enumerate}

\begin{figure}
	\centering
	\begin{small}
		\begin{tikzpicture}
			\node [draw=black, inner sep=5, anchor=west, text width=6cm] (a) {$(p_{n,m})$ admits symmetric circuits of polynomial size};
			
			\node [draw=black, inner sep=5, anchor=west, text width=6cm, fill=gray!10] at (0,-2) (b1) {$(p_{n,m})$ admits symmetric circuits of polynomial orbit size};
			
			\node [draw=black, inner sep=5, anchor=west, text width=6cm, fill=gray!10] at (9,-2) (b2) {the $p_{n,m}$ are  linear combinations of bounded-treewidth homomorphism polynomials, i.e.\ $p_{n,m} \in \mathfrak{T}^k_{n,m}$.};
			
			\draw [<->] (b1) edge node [midway, above] {\cref{thm:main1}} (b2);
			\node [draw=black, inner sep=5, anchor=west, text width=6cm] at (0,-4) (c) {$(p_{n,m})$ has bounded counting width};
			
			\draw [->] (a) -- (b1);
			\draw [->] (b1) edge node [midway, right] {\cite{anderson_symmetric_2017}, cf.\ \cref{thm:counting-width-explicit}} (c);
			
			\draw [dashed] (c.west) -- (-.5, -4);
			\draw [dashed] (-.5,-4) edge node [midway, left, rotate=90, anchor=south] {\cref{thm:dichotomy-chromatic,thm:lincomb-intro,thm:sublinear-intro}} (-.5, 0);
			\draw [dashed, ->] (-.5, 0) -- (a.west);
			
			\draw [dashed,->] (c.east) -| (b2.south) node[midway, above left] {\Cref{thm:lincomb-intro,thm:multilinear}};
			
			%	\draw[densely dotted, ->] (c.east) -| (b2.south);
		\end{tikzpicture}
	\end{small}
	\caption{Implications between properties of families of matrix-symmetric polynomials $(p_{n,m})$. The dashed implications hold under additional assumptions.}
	\label{fig:overview}
\end{figure}

Moreover, in the case of \cref{it:hom,it:sub} and when \cref{it:lincomb} is restricted to linear combinations of polynomially many homomorphism counts, 
bounded counting width characterises the families of matrix-symmetric polynomials admitting symmetric circuits of polynomial size (rather than orbit size).
For \cref{it:hom,it:sub}, we give a combinatorial characterisation of the families of patterns whose homomorphism/subgraph polynomials have bounded counting width in terms of treewidth and vertex cover number, respectively.

Finally, to give another concrete example for the merits of the
symmetric circuit framework, in~\cref{sec:immanants} we show that a
complexity dichotomy for the \emph{immanant} families due to~\cite{curticapean2021full}, whose hard cases are conditional on certain complexity-theoretic assumptions, holds unconditionally for symmetric circuits. The immanants interpolate between the permanent and determinant polynomials, which are in fact the extreme cases. 
Thus, the symmetric dichotomy for the immanants completes the picture begun in~\cite{dawar_symmetric_2025}.

\paragraph{Related work.}
\textbf{Symmetric circuits} are an emerging area of research in recent years. 
We have already mentioned the lower bound for symmetric circuits for the permanent from \cite{dawar_symmetric_2025}. 
In a similar vein, it has been shown that the determinant admits no small $\Alt_n \times \Alt_n$-symmetric circuits, even though it does have polynomial size $\Sym_n$-symmetric ones \cite{dawar2021lower}.
Other known lower bounds are for example for the symmetric formula complexity of computing the product of permutation matrices \cite{HeR23}, and a lower bound for certain types of symmetric AC$^0$ circuits for the parity function \cite{Rossman19}. More recent work establishes lower bounds for symmetric bounded-depth circuits with modular counting gates \cite{KawalekW25}. 
In \cite{anderson_symmetric_2017}, a precise characterisation of the power of uniform polynomial-size symmetric Boolean circuits with threshold gates was given: Such circuit families are equally powerful as fixed-point logic with counting, a formalism that is also related to the Weisfeiler-Leman algorithm. In particular, the graph properties computable with these circuits have bounded counting width, too. Similarly, in \cite{dawar_rank_logic}, rank logic is characterised via a class of symmetric circuits with more involved types of gates. 

Another direction that has been explored are symmetric algebraic circuits in the context of algebraic proof complexity, more specifically in the circuit-based \emph{Ideal Proof System} \cite{symIPS}. We hope that this direction will benefit from the framework we develop here, possibly leading to new lower bounds in proof complexity.

Another aspect of our work are the \textbf{homomorphism polynomials}. Interestingly, their complexity has been studied before with respect to a different restricted model, namely \emph{monotone} circuits. In \cite{KomarathPR23, PrateekRadu25} it is shown that the monotone circuit complexity of homomorphism polynomials is controlled by the treewidth of the pattern graphs in a similar way as we show this for symmetric circuits (and moreover, the circuit depth relates to the depth of the tree decompositions). It should be mentioned, though, that our setting is more general in two ways: We consider \emph{linear combinations} of homomorphism polynomials, rather than single ones, and our set-up is non-uniform, while in \cite{KomarathPR23, PrateekRadu25}, the pattern graph is taken to be the same fixed graph for all sizes of host graphs. 
Certain types of homomorphism polynomials (defined differently from the ones we study) have also been shown to be $\VP$-complete \cite{durand2014homomorphism}.

Our results also compare nicely to what is known on the \emph{computational} complexity of \textbf{homomorphism} and  \textbf{subgraph counting}. 
As already mentioned, it was shown in \cite{curticapean_homomorphisms_2017} that the computational complexity of graph parameters expressible as linear combinations of homomorphism counts (so-called \emph{graph motif parameters}) depends only on the treewidth of the pattern graphs. The problem is in $\mathsf{FPT}$ if the treewidth is bounded, and $\mathsf{\# W}[1]$-hard, otherwise. 
Our Theorem \ref{thm:main1} yields the same tractability criterion for the symmetric circuit complexity of the polynomials expressing these parameters but our proof techniques are very different. 
Moreover, as already mentioned, \cite{curticapean_homomorphisms_2017} concerns only uniform graph motif parameters, whereas we also cover the non-uniform case where we are counting different patterns for each host graph size.

The relationship between the complexity of subgraph counting and the \emph{vertex cover number} of the pattern graphs that we obtain in \cref{thm:sublinear} also appears in \cite{curticapean_homomorphisms_2017}. In \cite[Theorem 1.3]{neuen_homomorphism-distinguishing_2024}, a tight connection between the counting width of subgraph counts (for a fixed constant-size pattern) and the vertex cover number of the pattern is established. 
Our  \cref{thm:sublinear} generalises this result to the case of sublinearly growing patterns.

\paragraph*{Acknowledgments.}We acknowledge fruitful discussions with Radu Curticapean, Filip Kučerák, Deepanshu Kush, and Benjamin Rossman.

The first and second author were funded by UK Research and Innovation (UKRI) under
the UK government’s Horizon Europe funding guarantee: grant number EP/X028259/1.
The third author was supported by the European Union (CountHom, 101077083). Views and opinions expressed are however those of the author(s) only and do not necessarily reflect those of the European Union or the European Research Council Executive Agency. Neither the European Union nor the granting
authority can be held responsible for them.

\section{Preliminaries}
\label{sec:shortPreliminaries}

We write $\mathbb{N} = \{0,1,\dots\}$ and $[n] = \{1, \dots, n\}$ for $n \geq 1$.
For a tuple $\boldsymbol{x} \in X^k$ and $i \in [k]$, $x \in X$, write $\boldsymbol{x}[i/x] = x_1 \dots x_{i-1} x x_{i+1} \dots x_k \in X^k$ and $\boldsymbol{x}[i/] = x_1 \dots x_{i-1} x_{i+1} \dots x_k \in X^{k-1}$.
For a set $A$,
write $\Pi(A)$ for the set of partitions of $A$.
For a partition $\pi \in \Pi(A)$,
write $A/\pi$ for its set of parts.
Abusing notation, we also write $\pi$ for the canonical map $A \to A/\pi$.

\subsection{Permutation groups and supports}
\label{sec:supports}
Let $\Gamma$ be a group acting on a set $X$.  
For $S \subseteq X$, let $\Stab_\Gamma(S) \coloneqq \{ \pi \in \Gamma \mid \pi(S) = S\}$ be the \emph{stabiliser} of $S$ and let $\StabP_\Gamma(S) \coloneqq \{ \pi \in \Gamma \mid \pi(s) = s \text{ for every } s\in S\}$ be the \emph{pointwise stabiliser} of $S$.
The \emph{orbit} of $S$ is denoted $\Orb_{\Gamma}(S) = \{ \pi(S) \mid \pi \in \Gamma  \}$. It holds $|\Orb_\Gamma(S)| = |\Gamma| / |\Stab_\Gamma(S)| = [\Gamma : \Stab_\Gamma(S)]$ by the Orbit-Stabiliser Theorem.
We drop the subscript $\Gamma$ if the group is clear from the context.

Let $\Delta \leq \Gamma$ be a subgroup. A set $S \subseteq X$
is a \emph{support} of $\Delta$ if $\StabP_\Gamma(S) \leq
\Delta$.  In what follows, we are mainly interested in the
case where $\Gamma$ is the group $\Sym_n$ acting on $[n]$ in
the natural way, or it is the group $\Sym_I \times \Sym_J$
acting on $I \uplus J$ where $I = [n], J = [m]$ denote the bipartition of a bipartite graph.  For
simplicity, we just refer to the group $\Gamma$ when we mean
its action on the relevant set.  In these special cases it is known that there always exists a unique minimal support.
\begin{lemma}
	\label{lem:minSupportExists}
	Every subgroup $\Delta \leq \Sym_n$ that has a support of size $< n/2$ has a unique minimal support  $\sup(\Delta)$.
	Every subgroup $\Delta \leq \Sym_I \times \Sym_J$ that has a support $S$ with $|S \cap I| < |I|/2$ and $|S \cap J| < |J|/2$ has a unique minimal support $\sup(\Delta)$.
	In both cases $\StabP(\sup(\Delta)) \leq \Delta \leq \Stab(\sup(\Delta))$.
\end{lemma}	
\begin{proof}
	The existence of a unique minimal support is shown by proving that the intersection of any two supports with the respective size bound is again a support. For $\Delta \leq \Sym_n$, this is proved in Lemma 26 in \cite{blass1999choiceless}. 
	For $\Delta \leq \Sym_I \times \Sym_J$, it is easy to check that the proof of that lemma also goes through.
	The assertion $\StabP(\sup(\Delta)) \leq \Delta \leq \Stab(\sup(\Delta))$ is shown in \cite{anderson_symmetric_2017}.
\end{proof}

We also use the following fact without explicitly referring to it later. 
\begin{lemma}[\cite{anderson_symmetric_2017}]
	\label{lem:groupActionOnSupports}
	Let $\Delta \leq \Gamma$ for an action of $\Gamma$ on $X$ such that the minimal support $\sup(\Delta)$ is defined. Let $\Delta' \coloneqq \pi \Delta \pi^{-1}$ for $\pi \in \Gamma$. Then $\sup(\Delta') = \pi(\sup(\Delta))$.
\end{lemma}	
When $\Gamma = \Sym_I \times \Sym_J$ acting on $I \uplus J$
we write $\sup_L(\Delta) \coloneqq \sup(\Delta) \cap I$ and
$\sup_R(\Delta) \coloneqq \sup(\Delta) \cap J$ to denote the
restrictions of the support to the left and right part of the
bipartition, respectively.  Generally, for any set $S \subseteq I \uplus J$, we write $S_L$ to denote $S \cap I$ and $S_R$ for $S \cap J$.

We also often need to think of supports as ordered tuples rather than unordered sets, in which case we write $\vec{\sup}(\Delta)$, or $\vec{\sup}_L(\Delta)$, $\vec{\sup}_R(\Delta)$. The ordering of the tuples is specified in the respective context.

\subsection{Symmetric algebraic circuits}

An \emph{algebraic circuit} over a variable set $\Xx$ and a field $\bbF$ is a connected \textsmaller{DAG} with multiedges
such that each vertex is labelled with an element of $\Xx \cup \bbF \cup \{+,\times\}$. 	The vertices of a circuit are called \emph{gates}, the edges \emph{wires}.
Gates labelled with elements from $\Xx \cup \bbF$ are called
\emph{input gates} and they are not allowed to have incoming
edges. Every element of $\Xx \cup \bbF$ is allowed to appear at most once as the label of a gate -- this is no restriction because one can simply identify the respective gates.  We assume that a circuit contains exactly one gate without outgoing wires, and this is called the \emph{output gate}.

An algebraic circuit represents (or computes) a polynomial in
$\bbF[\Xx]$ in the obvious way: the gates $+$ and $\times$
are simply interpreted as addition and multiplication on
polynomials, and we care about the polynomial at the output
gate. Arrows are directed from a computation result towards
the next gate where that result is used as an input. The set
of \emph{children} of a gate $g$, denoted $\child(g)$, is the
set of gates $h$ such that $(h,g)$ is a wire of the circuit.  
Multiedges are allowed so that for example the polynomial
$x^2$ can be represented with just one input and one
multiplication gate. The \emph{size} of a circuit $C$, is
defined as the number of gates plus the number of wires, counted with multiplicities. It is denoted $\norm{C}$.

Let $\Gamma$ be a group that acts on $\Xx$. Then a circuit $C$ is \emph{$\Gamma$-symmetric} if the action of every $\pi \in \Gamma$ on the input gates $\mathcal{X}$ extends to an automorphism of the circuit:
For every input gate $g$, let $\ell(g)$ denote its label. This is either a variable or a field element. 
Field elements are fixed by every $\pi \in \Gamma$.
We say that a permutation $\pi \in \Gamma$ \emph{extends to a circuit automorphism} of the  circuit $C$ 
if there exists a $\sigma \in \Sym(V(C))$ such that $\ell(\sigma(g)) = \pi(\ell(g))$ for every input gate, 
and such that $\sigma$ is an automorphism of the multigraph structure of $C$ (i.e.\ it maps edges to edges, non-edges to non-edges, and preserves the operation types of the gates). 
For more details, see \cite{dawar_symmetric_2025}. 

A $\Gamma$-symmetric circuit is called \emph{rigid} if it has no non-trivial circuit automorphism that fixes every input gate. This is equivalent to saying that for every $\pi \in \Gamma$, there is a \emph{unique} circuit automorphism $\sigma$ that extends $\pi$. In that case, we also write $\pi(g)$ for internal gates $g$, to denote the application of the unique circuit automorphism that extends $\pi$. Symmetric circuits can always be assumed to be rigid:
%We will actually need a stronger property than rigidity, let's call it \emph{strong rigidity}. A circuit is strongly rigid if there do not exist two distinct gates in it that have the same set of children. A strongly rigid circuit is always rigid because for there to be a non-trivial automorphism that fixes all input gates, there must exist two internal gates with the same set of children. It is also easy to construct examples of circuits that are rigid but not strongly rigid. 

\begin{lemma}
	\label{lem:rigidifyCircuits}
	Let $C$ be a $\Gamma$-symmetric circuit.
	Then there exists a $\Gamma$-symmetric \emph{rigid} circuit
	$C'$ that represents the same polynomial with $\norm{C'} \leq
	\norm{C}$.
\end{lemma}	
\begin{proof}
	For Boolean circuits, the transformation of a non-rigid into a rigid symmetric circuit is described in the proof of \cite[Lemma 7]{anderson_symmetric_2017}. 
	It is not hard to see that a similar transformation can be carried out for algebraic circuits, and that it satisfies $\norm{C'} \leq \norm{C}$.
\end{proof}	
Note that in a rigid $\Gamma$-symmetric circuit, we can identify $\Gamma$ with the automorphism group of the circuit. 
Another advantage of rigid circuits is
that we obtain a well-defined notion of \emph{support} of gates, and
that there is a strong link between the size of these supports and the
orbit sizes of the gates.  
In the following, we always assume that $\Gamma$ is $\Sym_n$ or $\Sym_n \times \Sym_m$.
In the former case, our variable set is $\Xx_n = \{ x_{ij} \mid i,j \in [n] \}$, and $\pi \in \Sym_n$ acts simultaneously on both indices, so $\pi(x_{ij}) = x_{\pi(i)\pi(j)}$.
In the latter case, the variable set is $\Xx_{n,m} = \{x_{ij} \mid i
\in [n], j \in [m]\}$. Then a pair of permutations $(\pi, \pi') \in
\Sym_n \times \Sym_m$ applied to $x_{ij}$ yields $x_{\pi(i)\pi'(j)}$.
However, the support of a subgroup $\Delta$ of $\Gamma$ is always defined
with respect to the action of $\Gamma$ on $[n]$ in the case of
$\Sym_n$ and on $[n]\uplus [m]$ in the case of $\Sym_n \times \Sym_m$.

\paragraph*{Complexity measures for symmetric circuits}
We study the complexity of symmetric polynomials in terms of properties of their symmetric circuits. 
An important measure for symmetric circuits is $\maxOrb(C)$, the \emph{maximum orbit size} of any gate in $C$, formally:
\[
\maxOrb(C) \coloneqq \max_{g \in V(C)} | \{ \sigma(g) \mid \sigma \in \Gamma \}  |.
\]

If $C$ is rigid, and $\Gamma$ and $C$ are such that every stabiliser
group of a gate admits a unique minimal support, then we can also
speak about the \emph{support size} of $C$. In particular, this is
possible if $\Gamma \in \{\Sym_n, \Sym_n \times
\Sym_m\}$ and all gates have a support of size $\leq n/2$, because then, Lemma \ref{lem:minSupportExists} applies.
We write $\sup(g) \subseteq [n] \uplus [m]$ to denote the \emph{minimal support} of the group $\Stab_{\Sym_n \times \Sym_m}(g)$ in $\Sym_n \times \Sym_m$, and as before, $\sup_L(g)$ and $\sup_R(g)$ are its restrictions to $[n]$ and $[m]$, respectively. Similarly, in a $\Sym_n$-symmetric circuit, $\sup(g) \subseteq [n]$ denotes the minimal support of the group $\Stab_{\Sym_n}(g)$.
The \emph{support size} of a circuit $C$ is then:	
\[
\maxSup(C) \coloneqq \max_{g \in V(C)} | \sup(g) |.
\]	
We have the following relations between $\maxOrb$ and $\maxSup$. 
\begin{lemma}[restate=constantSupportOfGates,label=lem:constantSupportOfGates,name=]
	\label{lem:constantSupportOfGates}
	Let $(\Gamma_{n,m})_{n,m \in \bbN}$ be either the sequence $(\Sym_n)_{n \in \bbN}$ or $(\Sym_{n} \times \Sym_{m})_{n,m \in \bbN}$.
	Let $(C_{n,m})_{n,m \in \bbN}$ be a sequence of $\Gamma_{n,m}$-symmetric rigid circuits.
	\begin{enumerate}
		\item If $\maxOrb(C_{n,m})$ is polynomially bounded in $n+m$, then there exists a constant $k \in \bbN$ such that for all large enough $n,m \in \bbN$, it holds $\maxSup(C_{n,m}) \leq k$.
		\item If $\maxOrb(C_{n,m})$ is at most $2^{o(n+m)}$, then $\maxSup(C_{n,m}) \leq o(n+m)$.
	\end{enumerate}	
\end{lemma}	
\begin{proof}
	In the case $\Gamma_n = \Sym_n$, this is stated in
	Theorems~6.2 and~6.3 in \cite{dawar_symmetric_2025}.
	One can check that for $\Gamma_n = \Sym_n \times \Sym_m$, the proofs of these theorems also go through. We refrain from spelling out all details but the core argument is this: 
	If the orbit size of a gate is polynomially bounded, then by the orbit-stabiliser theorem, the index of its stabiliser group in $\Gamma_n$ is also polynomial, i.e.\ can be bounded by $\binom{n}{k}$, for some constant $k$.
	Theorem~6.2 in \cite{dawar_symmetric_2025} relies on Theorems 4.6 and 4.10 in \cite{dawar_rank_logic}. The former of them states the following: If $\Delta \leq \Sym_n$ is a subgroup of index at most $\binom{n}{k}$, then $\Delta$ contains an alternating group $A$ on $(n-k)$ points. The proof of \cite[Theorem 4.10]{dawar_rank_logic} then shows that if $\Delta$ is the stabiliser group of a gate in a circuit, then the gate is stabilised not only by $A$, but by the full symmetric group on the $n-k$ many points. 
	This means that the remaining $k$ points constitute a support for the gate. Now if $\Delta \leq \Sym_n \times \Sym_m$ is the stabiliser group of a gate, then we can consider $\Delta_1, \Delta_2 \leq \Delta$ defined as the subgroups of $\Delta$ consisting of those members of the product that have the neutral element in the first or second component, respectively. 
	Because the orbit size of the gate is polynomially bounded,
	the index of $\Delta_1, \Delta_2$ in $\Sym_n, \Sym_m$,
	respectively, is at most $\binom{n}{k}$ (or $\binom{m}{k}$),
	for some constant $k$. Then the same reasoning as before
	applies (this is analogous to the argument made, in the context of
	the alternating group, in~\cite{dawar2021lower}).
	This finishes the first part of the theorem. The second part follows from the first part as shown in the proof of  \cite[Theorem~6.3]{dawar_symmetric_2025}, and this is independent of whether $\Gamma_n = \Sym_n$ or $\Gamma_n = \Sym_n \times \Sym_m$.
\end{proof}	

The converse of \cref{lem:constantSupportOfGates}, part 1, is also true:

\begin{lemma}[restate=constantSupportImpliesPolyOrbit,label=lem:constantSupportImpliesPolyOrbit,name=]
	\label{lem:constantSupportImpliesPolyOrbit}
	Let $(C_{n,m})_{n \in \bbN}$ be a sequence of rigid $\Gamma_{n,m}$-symmetric circuits, for $\Gamma \in \{\Sym_n, \Sym_n \times \Sym_m\}$. Assume that there is a constant $k \in \bbN$ such that $\maxSup(C_{n,m}) \leq k$, for all $n,m \in \bbN$.
	Then $\maxOrb(C_{n,m}) \leq (n+m)^k$.
\end{lemma}
\begin{proof}		
	We consider the case $\Gamma = \Sym_n \times \Sym_m$, the other case is similar.
	Since $C \coloneqq C_{n,m}$ is rigid, every $(\pi, \pi') \in \Sym_n \times \Sym_m$ has a unique extension to a circuit automorphism $\sigma$. 
	Let $g \in V(C)$ be a gate. Its stabiliser subgroup in $\Sym_n \times \Sym_m$ contains $\StabP
	_{\Sym_n \times \Sym_m}(\sup(g))$. 
	This means that for any two pairs $(\pi_1, \pi_1'), (\pi_2, \pi_2') \in \Sym_n \times \Sym_m$, it holds: If $(\pi_1, \pi_1')(s) = (\pi_2, \pi_2')(s)$ for every $s \in \sup(g)$, and $\sigma_1, \sigma_2 \in \Aut(C_n)$ are the unique circuit automorphisms that $(\pi_1, \pi_1'), (\pi_2, \pi_2')$ extend to, then $\sigma_1(g) = \sigma_2(g)$. Thus, if $\sigma_1(g) \neq \sigma_2(g)$, then also $(\pi_1, \pi_1')$ and $(\pi_2, \pi_2')$ must disagree on some $s \in \sup(g)$. 
	Therefore, the number of possible distinct images of $g$ under the circuit automorphisms that the action of $\Sym_n \times \Sym_m$ extends to is at most the number of different images of $\vec{\sup}(g)$, where this denotes $\sup(g)$ written as a tuple in an arbitrary fixed order:
	\begin{align*}
		&|\{ \sigma(g) \mid \sigma \text{ is the unique circuit automorphism extending a } (\pi, \pi') \in \Sym_n \times \Sym_m   \}|\\
		\leq
		&|\{ (\pi, \pi')(\vec{\sup}(g)) \mid (\pi, \pi') \in \Sym_n \times \Sym_m   \}| \leq (n+m)^k.
	\end{align*}
	The last inequality is due to the fact that $|\sup(g)| \leq k$, so the number of different $k$-tuples in $[n] \uplus [m]$ is at most $(n+m)^k$.
	Since $C$ does not have other circuit automorphisms due to its rigidity, this is indeed an upper bound on the orbit size of the gate.
\end{proof}

If we are considering a sequence $(C_{n,m})_{n,m \in \bbN}$ of $\Gamma_{n,m}$-symmetric circuits, then we say they have \emph{polynomially bounded orbit size} if there is a polynomial $p$ such that asymptotically, $\maxOrb(C_{n,m}) \leq p(n+m)$. Note that the overall size $\norm{C_{n,m}}$ may grow super-polynomially even if the orbit size is polynomially bounded. In that case, symmetry itself is not the reason why the circuits are large: They simply contain a super-polynomial number of distinct orbits.

\paragraph*{Relationship between supports of children and parent gates}

The children of a gate $g$ can be partitioned into orbits with respect to the action of $\Stab(g)$ because any permutation that induces a circuit automorphism fixing $g$ must also fix the children of $g$ setwise. We can also consider the orbit partition of children of $g$ with respect to $\StabP(\sup(g)) \leq \Stab(g)$, which may be finer. 
\begin{lemma}
\label{lem:intersectionOfChildSupports}
Let $C$ be a $\Sym_n \times \Sym_m$-symmetric circuit such that $n,m \geq 3 \cdot \maxSup(C)$. 
Let $g \in V(C)$ an internal gate, $h$ a child of $g$.
Let $O_h \coloneqq \Orb_{\StabP(\sup(g))}(h)$ and $S(h) \coloneqq \bigcap_{h' \in O_h} \sup(h')$.
Then $S(h) \subseteq \sup(g)$, and for every $h' \in O_h$, $(\sup(h') \setminus S(h)) \cap \sup(g) = \emptyset$. 
\end{lemma}	
\begin{proof}
For the first part, let $s \in S(h)$, so in particular, $s \in \sup(h)$. Assume for a contradiction that $s \notin \sup(g)$. Then let $(\pi, \pi') \in \Sym_n \times \Sym_m$ be such that $(\pi, \pi')(\sup(h) \setminus \sup(g))$ does not contain $s$, and $(\pi, \pi')$ fixes $\sup(g)$ pointwise. Such a permutation exists because we are assuming $n$ and $m$ to be large enough compared to the support sizes. Then $h' \coloneqq (\pi, \pi')(h) \in O_h$ but $s \notin \sup(h')$, so $s \notin S(h)$. Contradiction.
For the second part, assume for a contradiction that there exists $s \in (\sup(h') \setminus S(h)) \cap \sup(g)$, for some $h' \in O_h$. Since $s \in \sup(g)$, every permutation in $\StabP(\sup(g))$ fixes $s$. So $s$ must then be in $\sup(h')$ for every $h' \in O_h$. So it is in $S(h)$. Contradiction.	
\end{proof}

	\subsection{Graphs and Homomorphisms}

A \emph{simple graph} is a tuple $G = (V(G), E(G))$ of a finite set $V(G)$ and a set $E(G) \subseteq \binom{V(G)}{2}$.
A \emph{multigraph} is a tuple $G = (V(G), E(G))$ of a finite set $V(G)$ and a multiset $E(G)$ of elements in $\binom{V(G)}{2}$. 
When a distinction is not necessary, we call both simple graphs and multigraphs \emph{graphs}.
For a graph $G$, write $\norm{G} \coloneqq \abs{V(G)} + \abs{E(G)}$.

A \emph{tree decomposition} of a graph $F$ is a pair $(T, \beta)$ of a tree $T$ and a map $\beta \colon V(T) \to 2^{V(F)}$ such that the following conditions hold:
\begin{enumerate}
	\item $\bigcup_{t \in V(T)} \beta(t) = V(F)$,
	\item for every edge $e \in E(F)$, there exists a node $t \in V(T)$ such that $e \subseteq \beta(t)$,
	\item for evert vertex $v \in V(F)$, the subgraph of $T$ induced by the vertices $t \in V(T)$ such that $v\in \beta(t)$ is connected.
\end{enumerate}
The \emph{width} of a tree decomposition $(T, \beta)$ is $\max_{v \in V(T)} |\beta(v)| -1$. 
The \emph{treewidth} $\tw(F)$ of a graph $F$ is the minimum width of a tree decomposition of $F$.

Let $F$ and $G$ be graphs.
A \emph{homomorphism} $h \colon F\to G$ is a map $h \colon V(F) \to V(G)$ such that $h(uv) \in E(G)$ for all $uv \in E(F)$.
We write $\hom(F, G)$ for the number of homomorphisms from $F$ to $G$.

Homomorphism counts behave nicely with respect to graph operations.
For two graphs $F_1$, $F_2$, write $F_1 + F_2$ for their \emph{disjoint union}, i.e.\ $V(F_1 + F_2) \coloneqq V(F_1) \uplus V(F_2)$ and $E(F_1 + F_2) \coloneqq E(F_1) \uplus E(F_2)$.
For multiple graphs $F_1, \dots, F_n$, we write $\coprod_{i=1}^n F_i$ for $F_1 + \dots + F_n$.
For two graphs $G_1$, $G_2$, write $G_1 \times G_2$ for their \emph{categorical product}, i.e.\ $V(G_1 \times G_2) \coloneqq V(G_1) \times V(G_2)$ and $(v_1, v_2)$ and $(w_1, w_2)$ are adjacent in $G_1 \times G_2$ if, and only if, $v_1w_1 \in E(G_1)$ and $v_2w_2 \in E(G_2)$.
The following identities hold for all graphs $F_1$, $F_2$, $G_1$, $G_2$, and all connected graphs $K$, cf.\ e.g.\ \cite[(5.28)--(5.30)]{lovasz_large_2012}:
\begin{align}
	\hom(F_1 + F_2, G) &= \hom(F_1, G) \hom(F_2, G), \label{eq:coproduct} \\
	\hom(F, G_1 \times G_2) &= \hom(F, G_1) \hom(F, G_2), \text{ and }\label{eq:product} \\
	\hom(K, G_1 + G_2) &= \hom(K, G_1) + \hom(K, G_2). \label{eq:disjoint}
\end{align}

For graphs $G$ and $H$, write $G \boxplus H$ for the graph with vertex set $V(G \boxplus H) \coloneqq V(G) \uplus V(H)$ and $E(G \boxplus H) \coloneqq E(G) \uplus E(H) \uplus \{vw \mid v \in V(G), w\in V(H)\}$.
For a graph $F$ and $U \subseteq V(F)$,
write $F[U]$ for the \emph{subgraph induced by $U$} and $F -U$ for the subgraph obtained by deleting $U$, i.e.\ $V(F[U]) \coloneqq U$, $E(F[U]) \coloneqq E(F) \cap \binom{U}{2}$, and $F-U \coloneqq F[V(F) \setminus U]$.
For every graph $F$, it holds that
\begin{equation}\label{eq:boxplus}
	\hom(F, G \boxplus H) = \sum_{U \subseteq V(F)} \hom(F[U], G) \hom(F -U, H).
\end{equation}

For graphs $G_1$ and $G_2$, write $G_1 \cdot G_2$ for their \emph{lexicographic product}, i.e.\ $V(G_1 \cdot G_2) \coloneqq V(G_1) \times V(G_2)$ and $(v_1, v_2)$ and $(w_1, w_2)$ are adjacent in $G_1 \cdot G_2$ if, and only if, $v_1 = w_1$ and $v_2w_2 \in E(G_2)$, or $v_1w_1 \in E(G_1)$.
For a graph $F$ and a partition $\mathcal{R}$ of $V(F)$,
write $F/\mathcal{R}$ for the graph with vertex set $\mathcal{R}$ and $RS \in E(F/\mathcal{R})$ if, and only if, $R \neq S$ and there exist $r \in R$ and $s \in S$ such that $rs \in E(F)$.
By \cite[Theorem~16]{seppelt_logical_2024},
for every graph $F$,
\begin{equation}\label{eq:lexprod}
	\hom(F, G_1 \cdot G_2) = \sum_{\mathcal{R} \in \Gamma(F)} \hom(F/\mathcal{R}, G_1) \hom(\coprod_{R \in \mathcal{R}} F[R], G_2).
\end{equation}	
where $\Gamma(F)$ denotes the set of all partitions $\mathcal{R}$ of $V(F)$ such that for all $R \in \mathcal{R}$ the graph $F[R]$ is connected.

\subsection{Counting Logic and Bijective Games}

Counting logic is first-order logic
augmented by counting quantifiers of the form $\exists^{\geq i} x$,
for every $i \in \bbN$. A graph (or another finite structure) $G$
satisfies a formula $\exists^{\geq i} x \phi(x)$, written as $G
\models \exists^{\geq i} x \phi(x)$, if there exist at least $i$ many
distinct $v \in V(G)$ such that $G \models \phi(v)$. The $k$-variable
fragment of counting logic is denoted $\Cc^k$.
More details on counting logic and references may be found
in~\cite{dawar_siglog15}.  For our purposes, the main interest is the
equivalence relation that this induces on graphs.  Two structures $G,H$
are called \emph{$\Cc^k$-equivalent}, denoted $G
\equiv_{\mathcal{C}^k} H$, if $G$ and $H$ satisfy exactly the same
sentences of $\Cc^k$.  This family of equivalence relations has many
characterisations, in combinatorics, logic, algebra and linear
optimisation.  In particular, it is known that $G
\equiv_{\mathcal{C}^{k+1}} H$ if, and only if, $G$ and $H$ are not
distinguished by the $k$-dimensional Weisfeiler-Leman test for graph
isomorphism (see~\cite{CFI}).

For our purposes, a particularly useful characterisation is in terms
of \emph{$k$-pebble bijective games}, introduced by
Hella~\cite{hella96}.  This is a game played by two players called
\emph{Spoiler} and \emph{Duplicator} on the graphs $G$ and $H$ using~$k$
pairs of pebbles~$(a_1,b_1),\dots,(a_k,b_k)$.  In a game position,
some (or all) of the pebbles~$a_1,\ldots,a_k$ are placed on vertices 
of~$G$ while the matching pebbles among~$b_1,\ldots,b_k$ are
placed on vertices of~$H$.  Where it causes no confusion, we
do not distinguish notationally between the pebble~$a_i$ (or~$b_i$)
and the vertex on which it is placed.  In each move of the game, Spoiler
chooses a pair of pebbles~$(a_i,b_i)$ and Duplicator has to respond by
giving a bijection~$f: V(G) \rightarrow V(H)$ which agrees with
the map~$a_j \mapsto b_j$ for all~$j \neq i$.    Spoiler chooses $a
\in V(G)$ and the pebbles $a_i$ and $b_i$ are placed on $a$ and $f(a)$
respectively.  If the resulting partial map from~$G$ to~$H$ given
by~$a_i \mapsto b_i$ is not a partial isomorphism, then Spoiler has
won the game.   We say that Duplicator has a \emph{winning
	strategy} if, no matter how Spoiler plays, she can play forever
without losing.  It is known that $G  \equiv_{\mathcal{C}^k} H$ if, and only if, Duplicator has a winning
strategy in the~$k$-pebble bijective game played on $G$ and $H$.

\textcite{CFI} showed that for each $k$ there
are graphs $G$ and $H$ with $O(k)$ vertices which are not
isomorphic but such that $G  \equiv_{\mathcal{C}^k} H$.  The graphs
$G$ and $H$ are constructed from a base graph $F$ which is
sufficiently richly connected (in particular it has treewidth larger
than $k$) by replacing each vertex $v$ of $F$ by a small gadget whose
size depends on the degree of $v$.  The construction is known as the
\textsmaller{CFI} construction, and the gadgets as the
\textsmaller{CFI} gadgets.  There are many variations of the
construction in the literature, and we rely on the construction
from~\cite{roberson_oddomorphisms_2022}, known as the
\textsmaller{CFI} construction without internal
vertices (see also~\cite{rossman_csl25}). Details can be found in Appendix \ref{sec:cfi}.

\section{Overview of results}

We start with a summary of our results and proof techniques. Full technical details of all proofs are presented in Sections~\ref{sec:circuitsAndHomPolynomials}--\ref{sec:immanants}.  
\subsection{Characterisation of symmetric circuit complexity by homomorphism polynomials }
We start by observing that any symmetric polynomial is expressible as a linear combination of \emph{homomorphism} or \emph{subgraph polynomials}.
Let $F$ be a bipartite multigraph with bipartition $A \uplus B$.
For integers $n,m \in \mathbb{N}$,
the \emph{homomorphism polynomial $\hom_{F, n,m} \in \mathbb{Q}[\mathcal{X}_{n,m}]$} and the \emph{subgraph polynomial $\sub_{F, n,m} \in \mathbb{Q}[\mathcal{X}_{n,m}]$} of $F$ are defined as
\begin{align*}
	\hom_{F, n, m} &\coloneqq \sum_{h \colon A \uplus B \to [n] \uplus [m]} \prod_{ab \in E(F)} x_{h(a)h(b)}, 
	\quad \quad
	\sub_{F, n, m} \coloneqq \frac{1}{|\Aut(F)|}\sum_{h \colon A \uplus B \hookrightarrow [n] \uplus [m]} \prod_{ab \in E(F)} x_{h(a)h(b)}.
\end{align*}
Here, we think of $h$ as mapping $A$ to $[n]$ and $B$ to $[m]$.
For example, $\hom_{K_2, n, m} = \sum_{v \in [n]} \sum_{w \in [m]} x_{vw}$ and $\sub_{K_{n,m}, n,m} = \prod_{v\in [n]} \prod_{w \in [m]} x_{vw}$. 
If $A_G \in \{0,1\}^{[n] \times [m]}$ is the bi-adjacency matrix of an undirected bipartite graph $G$ with bipartition $V(G) = [n] \uplus [m]$, then $\hom_{F,n,m}(A_G)$ evaluates to the number of bipartition-respecting homomorphisms from $F$ to $G$, and $\sub_{F, n,m}(A_G)$ is the number of occurrences of $F$ as a subgraph in $G$. We also write $\hom_{F,n,m}(G)$ and $\sub_{F, n,m}(G)$ for the polynomials evaluated in the bi-adjacency matrix of $G$.
Let
$
\mathfrak{G}_{n,m} \coloneqq \left\{\sum \alpha_i \hom_{F_i, n,m} \mid \text{bipartite multigraphs } F_i, \alpha_i \in \mathbb{Q} \right\} \subseteq \mathbb{Q}[\mathcal{X}_{n,m}].
$
Then we have:

\begin{restatable}{lemma}{symmetricPolynomials}
	\label{lem:gnm-sym}
	For $n,m \in \mathbb{N}$ and $p \in \mathbb{Q}[\mathcal{X}_{n,m}]$,
	the following are equivalent:
	\begin{enumerate}
		\item the polynomial $p$ is $\Sym_n \times \Sym_m$-symmetric,\label{it:gnm1}
		\item $p \in \mathfrak{G}_{n,m}$,\label{it:gnm2}
		\item $p = \sum \alpha_i \sub_{F_i, n,m}$ for some bipartite multigraphs $F_i$ and $\alpha_i \in \mathbb{Q}$.\label{it:gnm3}
	\end{enumerate}
\end{restatable}
	
For the proof,  \cref{it:gnm3} can be concluded from \cref{it:gnm1} by partitioning the monomials in $p$ into their orbits, and showing that each of these orbits is a subgraph polynomial for some pattern $F$. Via Möbius inversion (\cref{thm:sub-hom}), subgraph polynomials can be expressed as linear combinations of homomorphism polynomials.

Thus, both subgraph and homomorphism polynomials can be seen as spanning sets for the vector space of matrix-symmetric polynomials. 
We are now interested in the ones that admit polynomial size symmetric circuits.
For $k,n,m \in \mathbb{N}$, write
\[
\mathfrak{T}_{n,m}^k \coloneqq \left\{\sum \alpha_i \hom_{F_i, n,m} \ \middle|\  \text{bipartite multigraphs } F_i \text{ such that } \tw(F_i) < k, \alpha_i \in \mathbb{Q} \right\} \subseteq \mathfrak{G}_{n,m}.
\]
for the set of finite $\mathbb{Q}$-linear combinations of homomorphism polynomials of bipartite multigraphs of treewidth less than~$k$.

Our main result, \cref{thm:main1}, states that the families of polynomials in $\mathfrak{T}_{n,m}^k$, for constant $k$, are precisely the ones that admit symmetric circuits of polynomial orbit size. The full proof can be found in Section \ref{sec:circuitsAndHomPolynomials} but the idea is the following: 

\begin{proof}[Proof sketch of \cref{thm:main1}]
The easier direction is to show that if $p_{n,m} \in \mathfrak{T}_{n,m}^k$, then it has circuits of polynomial orbit size.
It turns out that each $\hom_{F, n, m}$ can be represented by a symmetric circuit with support size at most $\tw(F) < k$. 
The circuit computes $\hom_{F, n, m}$ inductively along a tree decomposition of $F$, using ideas as in e.g.\ \cite{dvorak_recognizing_2010}. 
Then any polynomial in $\mathfrak{T}_{n,m}^k$ is just a linear combination of such circuits.
By \cref{lem:constantSupportImpliesPolyOrbit}, a circuit with support size $\leq k$ has orbit size at most $(n+m)^k$, as desired. Note that a polynomial in $\mathfrak{T}_{n,m}^k$ may be a linear combination involving a super-polynomial number of terms. Then only the orbit size of the circuit will be polynomial, but the total number of gates will not be.

The other direction of Theorem \ref{thm:main1} requires substantial technical effort. It is proved using so-called \emph{labelled} homomorphism polynomials: Imagine that a pattern $F$ has $k$ distinctly labelled vertices $a_1, \dots, a_k$. Then for any $k$-tuple of vertices $b_1, \dots, b_k$ in the host graph $G$, we can ask how many homomorphisms $h$ exist from $F$ to $G$ such that $h(a_i) = b_i$, for each $i \in [k]$. A labelled homomorphism polynomial is one that counts only these label-preserving homomorphisms. 

We want to prove that every polynomial represented by a symmetric circuit of polynomial orbit size is in $\mathfrak{T}_{n,m}^k$, for some $k \in \bbN$.
\cref{lem:constantSupportOfGates} tells us that the support size of every gate in such a circuit must be bounded by some constant $k$. We prove by induction on the circuit structure that at every gate $g$ in the circuit, the polynomial computed at $g$ is a linear combination of \emph{labelled} homomorphism polynomials of patterns of treewidth $\leq k$.
The crucial insight is that the supports of the gates always correspond exactly to the images of the labels. That is, semantically, the gate $g$ counts homomorphisms that map the labels $a_1, \dots, a_k$ of the pattern graphs to the vertices
$b_1, \dots, b_k$ enumerating $\sup(g)$.
The root gate of a symmetric circuit always has empty support (as it is fixed by every permutation in $\Sym_n \times \Sym_m$), and so, the polynomial computed at the root is a linear combination of homomorphism polynomials for patterns $F$ with $\tw(F) \leq k$ and without labels. Thus, the circuit computes a polynomial in $\mathfrak{T}_{n,m}^k$.
\end{proof}

We also prove in Section \ref{sec:circuitsAndHomPolynomials} that under the common assumption $\VP \neq \VNP$,
\cref{thm:main1} is best-possible in the sense that the
characterisation cannot be improved to capture total circuit size
rather than orbit size: There exist polynomials in
$\mathfrak{T}_{n,m}^k$, for constant $k$, that do not admit polynomial
size circuits (neither symmetric nor asymmetric), unless $\VP =
\VNP$. 
\begin{theorem}[restate=hardPolynomialsInT,label=thm:hardPolynomialsInT,name=]
	There is a $\VNP$-hard family of polynomials $(p_{n})_{n \in \bbN}$ such that $p_{n} \in \mathfrak{T}_{n,n}^2$ for all $n \in \mathbb{N}$.
\end{theorem}	

\subsection{Symmetric complexity of homomorphism polynomials}

In order to make the characterisation in \cref{thm:main1} more usable, we study $\mathfrak{T}^k_{n,m}$ in relation with the class of polynomials of bounded counting width.

\begin{definition}[\cite{anderson_symmetric_2017}]
	Let $p$ be a graph parameter on $(n,m)$-vertex bipartite graphs.
	The \emph{counting width of $p$ on simple graphs} is the least integer $k$ such that, for all $(n,m)$-vertex simple bipartite graphs $G$ and $H$, if $G$ and $H$ are $\mathcal{C}^k$-equivalent, then $p(G) = p(H)$.
\end{definition}

%
%
%Here, $(n,m)$-vertex simple graphs are viewed as elements of $\{0,1\}^{n \times m}$.
%More generally, an \emph{edge-weighted} $(n,m)$-vertex bipartite graph is an element of $\mathbb{Q}^{n \times m}$. 
%For edge-weighted graphs $G$ and $H$, $\mathcal{C}^k$-equivalence is meant with respect to the following structural representation of $G$ and $H$: 
%Let $\tau \coloneqq \{A, B\} \cup \{R_q \mid q \in \mathbb{Q}\}$ denote the infinite relational signature with two unary symbols $A$ and $B$ and one binary symbol $R_q$ for every $q \in \mathbb{Q}$.
%We view elements $G \in \mathbb{Q}^{n \times m}$ as $\tau$-structures with universe $n \times m$ by interpreting $A^G \coloneqq [n]$, $B^G \coloneqq [m]$, and $R_q^G = \{ij \in n \times m \mid G(ij) = q\}$ for $q \in \mathbb{Q}$.

It follows with known arguments that whenever a family $(p_{n,m})_{n,m\in \mathbb{N}}$ admits circuits of polynomial orbit size, then its counting width is bounded (see \cref{thm:counting-width}). We ask in which cases the converse is true. As a start, we identify where the difficulty lies in this question. Namely, if we only care about the semantics of the polynomials on simple graphs, that is $\{0,1\}$-valued rather than $\bbQ$-valued matrices, then we do have a conclusive answer. On $\{0,1\}$-assignments, any polynomial $p$ is equivalent to its \emph{multilinearisation}. This is the polynomial obtained from $p$ by replacing every exponent greater than $1$ with $1$. Up to multilinearisation, bounded counting width is indeed sufficient for the existence of polynomial orbit size symmetric circuits:
	
	\begin{restatable}{theorem}{multilinear}
		\label{thm:multilinear}
		Let $k,n,m \in \mathbb{N}$.
		For $p \in \mathbb{Q}[\mathcal{X}_{n,m}]$, the following are equivalent:
		\begin{enumerate}
			\item there exists $q \in \mathfrak{T}^k_{n,m}$ such that $p$ and $q$ are equal up to multilinearisation,
			\item the counting width of $p$ on $(n,m)$-vertex simple bipartite graphs is at most $k$.
		\end{enumerate}
	\end{restatable}
	
Thus, if we only care about evaluating the polynomial $p$ on
$\{0,1\}$ (equivalently, we regard it as a parameter on simple,
unweighted, bipartite graphs) the symmetric tractability question is
completely settled.  But this does not settle the general case , where we wish to
consider it on all $\bbQ$-valued assignments (see,
e.g.~\cref{ex:matching}).  However, we can characterise it in more
specific cases.  In particular, 
we obtain a precise characterisation of the bounded counting width case for linear combinations of homomorphism polynomials with sublinear-size patterns:

\begin{theorem}[see \cref{thm:lincomb}]\label{thm:lincomb-intro} 
	For $n,m \in \mathbb{N}$, 
	let $F_{n,m,i}$ be bipartite multigraphs and $\alpha_{n,m,i}
        \in \mathbb{Q} \setminus \{0\}$ for $i \in [s_{n,m}]$ for some
        $s_{n,m} \in \mathbb{N}$. 
	Let $p_{n,m} \coloneqq \sum_i \alpha_{n,m,i} \hom_{F_{n,m,i}, n,m}$.
	Suppose that, for all $\nu, \mu \in \mathbb{N}$,
	there are $n', m' \in \mathbb{N}$ 
	such that $\max_i |F_{\nu n, \mu m,i}| \leq \min \{n,m\}$ for all $n > n'$ and $m > m'$.		
	Then the following are equivalent:
	\begin{enumerate}
		\item $\max_i \tw(F_{n,m,i})$ is bounded,
		\item the counting width of $p_{n,m}$ on $(n,m)$-vertex simple graphs is bounded,
		\item the $p_{n,m}$ admit $\Sym_n\times \Sym_m$-symmetric circuits of orbit size polynomial in $n+m$.
	\end{enumerate}
\end{theorem}

If the linear expansions of the $p_{n,m}$ are comprised of polynomially many homomorphism polynomials,
then the above conditions are equivalent to   the $p_{n,m}$ admitting $\Sym_n\times \Sym_m$-symmetric circuits of total size polynomial in $n+m$.
In the case where $p_{n,m} = \hom_{F_{n,m}, n,m}$ for all $n,m \in \bbN$, i.e.\ we just have a single homomorphism polynomial for each $n,m$, we obtain the above result also if $|F_{n,m}|$ is not in $o(n+m)$, see \cref{thm:dichotomy-chromatic}.

We prove \cref{thm:lincomb} by generalising notions introduced by \textcite{roberson_oddomorphisms_2022} for describing the distinguishing power of graph isomorphism relaxations via homomorphism counts. 
By \cite{dvorak_recognizing_2010,dell-grohe-rattan},
two graphs $G$ and $H$ are $\mathcal{C}^k$-equivalent if, and only if, they are \emph{homomorphism indistinguishable} over the graphs of treewidth less than $k$, i.e.\ $\hom(F,G) = \hom(F, H)$ for all $F$ such that $\tw(F) < k$.
Hence, a graph parameter $p$ has bounded counting width if, and only if, it is determined by homomorphism counts from bounded-treewidth graphs.
For uniform graph motif parameters $\sum \alpha_F \hom(F, -)$, it was known \cite{seppelt_logical_2024,neuen_homomorphism-distinguishing_2024,lanzinger2024on} that this happens if and only if $\max_F \tw(F)$ is bounded.
We generalise this result to non-uniform graph motif parameters, whose patterns and coefficients may depend on the size of the host graph.

\subsection{Symmetric complexity of subgraph polynomials}	
By \cref{lem:gnm-sym}, linear combinations of subgraph and homomorphism polynomials can be transformed into one another. Nevertheless, we also study subgraph polynomials on their own because it seems likely that in this case we find another natural parameter of the pattern graphs that fully controls the counting width and symmetric complexity. Indeed, we show that for patterns of sublinear size, this parameter is $\vc(F)$, the \emph{vertex cover number}. 

\begin{theorem}[see \cref{thm:sublinear}]\label{thm:sublinear-intro}
	Let $(F_{n,m})_{n,m \in \mathbb{N}}$ be a family of simple bipartite graphs such that, 
	for all $\nu, \mu \in \mathbb{N}$,
	there exist $n',m' \in \mathbb{N}$,
	such that $|F_{\nu n, \mu m}| \leq \min \{n,m\}$ for all $n > n'$ and $m > m'$.
	The following are equivalent:
	\begin{enumerate}
		\item $\vc(F_{n,m})$ is bounded,
		\item the counting width of $\sub_{F_{n,m},n,m}$ on $(n,m)$-vertex simple graphs is bounded,
		\item the $\sub_{F_{n,m}, n, m}$ admit $\Sym_n \times \Sym_m$-symmetric circuits of orbit size polynomial in $n+m$,
		\item the $\sub_{F_{n,m}, n, m}$ admit $\Sym_n \times \Sym_m$-symmetric circuits of size polynomial in $n+m$.
	\end{enumerate}
\end{theorem}

The size bound in this theorem is necessary as it is easy to find examples of subgraph polynomials (of size $\Theta(n+m)$) whose complexity is not governed by the pattern's vertex cover number.
Consider e.g.\ $p_{n,m} = \sub_{K_{n,m},n,m} = \prod_{v \in [n]} \prod_{w \in [m]} x_{vw}$.  The natural symmetric circuit for this linear-size pattern has just one product gate that ranges over all variables, so this family $(p_{n,m})_{n,m \in \bbN}$ has low symmetric circuit complexity (see also \cref{ex:knm}). But, $\vc(K_{n,m})$ is clearly unbounded. 

More generally, for any function $f \colon \bbN \to \bbN$ with $f(n) \leq n$, and the complete bipartite graphs $K_{f(n),f(n)}$ as patterns, we  show that the only tractable cases for $\sub_{K_{f(n), f(n)},n,n}$ are when $\min \{ f(n), n-f(n) \}$ is constant (\cref{thm:knn}). This suggests that it might be the minimum vertex cover number of the pattern graph and its complement that determines the symmetric complexity of $\sub_{F_{n,m},n,m}$. We call this parameter the \emph{logical vertex cover number} $\cc_{n,m}(F_{n,m})$.
In particular, we can show for all $F$ that the counting width of $\sub_{F,n,m}$ and $\sub_{\overline{F},n,m}$ is the same (\cref{thm:complements}), giving further evidence for this hypothesis,
which is formally stated as \cref{conj:sub}.

\subsection{Symmetric complexity of the immanants}
So far, we have studied $\Sym_n \times \Sym_m$-symmetric families. There are also polynomials in the variables $\Xx_{n,n}$ which are symmetric under the simultaneous action of $\Sym_n$ on rows and columns but not under $\Sym_n \times \Sym_n$.  One such example is the \emph{determinant}, and more generally, other \emph{immanants}. 
We have chosen to develop our theory for $\Sym_n \times
\Sym_m$-symmetric families but we believe that most of it can be
adapted to the $\Sym_n$-symmetric case. Then, the relevant
homomorphism or subgraph polynomials involve \emph{directed} and not
necessarily bipartite graphs. Rather than doing that here, we consider
the immanants as a concrete $\Sym_n$-symmetric example and show that
their symmetric circuit complexity is subject to a dichotomy similar
to what is known for their computational complexity~\cite{curticapean2021full}.  
 
\emph{Immanants} are families of $\Sym_n$-symmetric polynomials
$(p_n)_{n \in \bbN}$ that are defined via so-called \emph{irreducible
  characters} of the symmetric group. An irreducible character of
$\Sym_n$ is a homomorphism $f$ from $\Sym_n$ to the multiplicative group
$\mathbb{C}\setminus\{0\}$ which is constant on the conjugacy classes
of $\Sym_n$.  Given such an $f$, the corresponding immanant is defined as
$
\imm_{f} = \sum_{\pi \in \Sym_n} f(\pi) \prod_{i \in [n]} x_{i,\pi(i)}.
$
If $f$ is constantly $1$, then $\imm_f$ is the permanent. The
determinant is obtained by letting $f = \sgn$.  While all immanents are
$\Sym_n$-symmetric, the symmetry groups of particular members of this
family may be larger.  In particular, the permanent is, in fact,
$\Sym_n \times \Sym_n$ while the determinant is symmetric under a
subgroup of this of index $2$ consisting of those pairs of
permutations $(\pi,\sigma)$ for which $\sgn(\pi) = \sgn(\sigma)$.  The
$\Sym_n$ symmetric circuit complexity of the permanent and the determinant was
studied in~\cite{dawar_symmetric_2025} and for the larger symmetry
groups in~\cite{dawar2021lower}.

In~\cite{curticapean2021full}, Curticapean shows a complexity dichotomy for the immanants: In the tractable case, $\imm_{f}$ is in $\VP$ (i.e.\ admits polynomial size algebraic circuits) and is computable in polynomial time. In the intractable case, it is not in $\VP$, unless $\VFPT = \VW$, and also, it is not computable in polynomial time, unless $\FPT = \sharpW$. The complexity is controlled by a certain parameter of $f$. 

We show that the same parameter of $f$ produces the analogous
dichotomy with regards to the complexity of $\Sym_n$-symmetric
algebraic circuits for the immanants, but without the need for
complexity-theoretic assumptions for the hardness part.  This
generalizes the resuls of~\cite{dawar_symmetric_2025}.
The irreducible characters $f$ correspond naturally to partitions
of $[n]$ (see, for example,~\cite[Prop.~1.10.1]{Sagan}).  If $\lambda$ is a partition of $[n]$, we denote by $\imm_{\lambda}$ the immanant defined by the irreducible character corresponding to $\lambda$.
%Partitions of $[n]$ will be denoted as tuples $(k_1,\dots ,k_s)$, where each entry $k_i$ denotes a part of size $k_i$.
%For example, the permanent is $\imm_{(n)}$, and the determinant is $\imm_{(1,\dots ,1)}$. For other partitions, the rule for computing their corresponding irreducible characters is more involved and not needed here (for details, see \cite{curticapean2021full}).
Let $b(\lambda) \coloneqq n-s$, where $s$ denotes the number of parts and $n$ the size of the set that is partitioned by $\lambda$. This is the parameter that controls the complexity of $\imm_{\lambda}$.

A family $\Lambda$ of partitions is said to \emph{support growth} $g \colon \bbN \to \bbN$ if for every $n \in \bbN$ there is a partition $\lambda^{(n)}$ in $\Lambda$ with $b(\lambda^{(n)}) \geq g(n)$ and size $\Theta(n)$. According to the dichotomy from \cite{curticapean2021full}, $(\imm_{\lambda})_{\lambda \in \Lambda}$ is tractable if there is a constant that bounds $b(\lambda)$ for all $\lambda \in \Lambda$. Otherwise, if $\Lambda$ supports growth $\omega(1)$, then the immanant family is (conditionally) intractable. If $\Lambda$ even supports growth $\omega(n^k)$, then $(\imm_{\lambda})_{\lambda \in \Lambda}$ is $\VNP$-complete.
Our analogous result for symmetric circuits is the following.
\begin{theorem}[restate=immanantDichotomy, label=thm:immanantDichotomy,name=]
	Let $\Lambda$ be a family of partitions.
	\begin{enumerate}
		\item If there exists a $k \in \bbN$ such that $b(\lambda) \leq k$ for all $\lambda \in \Lambda$, then the $(\imm_{\lambda})_{\lambda \in \Lambda}$ admit $\Sym_n$-symmetric circuits of polynomial size in $n$.
		\item If $\Lambda$ supports growth $g \in \omega(1)$, then the counting width of $(\imm_{\lambda})_{\lambda \in \Lambda}$ on directed $\bbQ$-edge-weighted $n$-vertex graphs is unbounded, where $n$ is the integer that is partitioned by the respective $\lambda$.
		Thus, all $\Sym_n$-symmetric circuits for the immanant family have super-polynomial size in $n$. 
		\item If $\Lambda$ supports growth $g \in \omega(n^k)$ for a constant $k$, then the counting width of $(\imm_{\lambda})_{\lambda \in \Lambda}$ on directed $\bbQ$-edge-weighted $n$-vertex graphs is $\Omega(n)$, so all $\Sym_n$-symmetric circuits have size $2^{\Omega(n)}$.
	\end{enumerate}	
\end{theorem}

	\section{Circuits with Bounded Support and Homomorphism Polynomials}
	\label{sec:circuitsAndHomPolynomials}
	In this section, we prove \cref{thm:main1} which characterises the polynomials which admit $\Sym_n \times \Sym_m$-symmetric circuits of polynomial orbit size.
	To that end, we first observe that  the polynomials which admit $\Sym_n \times \Sym_m$-symmetric circuits of arbitrary orbit size are precisely the linear combinations of subgraph or homomorphism polynomials.
	
	Let $F$ be a bipartite multigraph with bipartition $A \uplus B$.
	For integers $n,m \in \mathbb{N}$,
	the \emph{homomorphism polynomial $\hom_{F, n,m} \in \mathbb{Q}[\mathcal{X}_{n,m}]$} and the \emph{subgraph polynomial $\sub_{F, n,m} \in \mathbb{Q}[\mathcal{X}_{n,m}]$} of $F$ are defined as
	\begin{align*}
		\hom_{F, n, m} &\coloneqq \sum_{h \colon A \uplus B \to [n] \uplus [m]} \prod_{ab \in E(F)} x_{h(a)h(b)}, \\
		\sub_{F, n, m} &\coloneqq \frac{1}{|\Aut(F)|}\sum_{h \colon A \uplus B \hookrightarrow [n] \uplus [m]} \prod_{ab \in E(F)} x_{h(a)h(b)}.
	\end{align*}
	Here, $h$ ranges over maps taking $A$ to $[n]$ and $B$ to $[m]$.
	For example, $\hom_{K_2, n, m} = \sum_{v \in [n]} \sum_{w \in [m]} x_{vw}$ and $\sub_{K_{n,m}, n,m} = \prod_{v\in [n]} \prod_{w \in [m]} x_{vw}$. 
	If $A_G \in \{0,1\}^{[n] \times [m]}$ is the bi-adjacency matrix of an undirected bipartite graph $G$ with bipartition $V(G) = [n] \uplus [m]$, then $\hom_{F,n,m}(A_G)$ evaluates to the number of bipartition-respecting homomorphisms from $F$ to $G$, and $\sub_{F, n,m}(A_G)$ is the number of occurrences of $F$ as a subgraph in $G$. We also write $\hom_{F,n,m}(G)$ and $\sub_{F, n,m}(G)$ for the polynomials evaluated in the bi-adjacency matrix of $G$.
	
	For every multigraph $F$, the coefficients of all monomials in $\sub_{F,n,m}$ are $1$. If $F$ is a simple graph, then $\sub_{F,n,m}$ is multilinear. 
	This is generally not the case for $\hom_{F, n,m}$, not even when $F$ is simple.
%	Note that every monomial in $\sub_{F, n, m}$ has coefficient $1$ and not $\frac{1}{|\Aut(F)|}$.
	
	As a first step, we show that the set
	\[
		\mathfrak{G}_{n,m} \coloneqq \left\{\sum \alpha_i \hom_{F_i, n,m} \mid \text{bipartite multigraphs } F_i, \alpha_i \in \mathbb{Q} \right\} \subseteq \mathbb{Q}[\mathcal{X}_{n,m}].
	\]
	contains precisely the $\Sym_n \times \Sym_m$-symmetric
        polynomials, and we also get the same collection if we take linear
        combinations of subgraph polynomials instead.

	\symmetricPolynomials*
	\begin{proof}
		For \ref{it:gnm1} $\Rightarrow$ \ref{it:gnm3},
		partition the monomial set of $p$ into its $\Sym_n \times \Sym_m$-orbits. 
		We aim to show that each orbit is the polynomial $\sub_{F,n,m}$, for some $F$, multiplied with some non-zero constant. 
		Let $\Omega$ be such an orbit, i.e.\ a set of monomials. 
		Let $q \in \Omega$ denote an arbitrary representative. 
		It describes a bipartite graph $F$ with bipartition $V(F) = \{ i \in [n] \mid \text{ there exists } j \text{ s.t. } x_{ij} \text{ divides } q \}  \allowbreak  \uplus \allowbreak \{ j \in [m] \mid \text{ there exists } i \text{ s.t. } x_{ij} \text{ divides } q \}$.
		An edge $ij$ appears in $F$ with the multiplicity given by the degree of $x_{ij}$ in $q$.
		
		Now $q$ encodes the identity embedding $\iota_{\id} \colon V \hookrightarrow [n] \uplus [m]$. 
		For every $(\pi,\sigma) \in \Sym_n \times \Sym_m$, 
		we have $(\pi, \sigma)(q) = \prod_{vw \in E(F)} x_{\pi(\iota_{\id}(v))\sigma(\iota_{\id}(w))}$. 
		Each embedding $\iota \colon V \hookrightarrow [n] \uplus [m]$ can be written as $(\pi,\sigma) \circ \iota_{\id}$ for some $(\pi,\sigma) \in \Sym_n \times \Sym_m$. 
		Therefore, the orbit $\Omega$ is indeed $\sub_{F,n,m}$, multiplied with some constant.
		
		The implication \ref{it:gnm3} $\Rightarrow$ \ref{it:gnm2} follows from \cref{thm:sub-hom}.
		For \ref{it:gnm2} $\Rightarrow$ \ref{it:gnm1}, let $F$ be a bipartite multigraph and $(\pi, \sigma) \in \Sym_n \times \Sym_m$.
		Then
		\[
			(\pi,\sigma)(\hom_{F, n, m}) = \sum_{h \colon A \uplus B \to [n] \uplus [m]} \prod_{ab \in E(F)} x_{\pi(h(a)) \sigma(h(b))}
			= \sum_{h \colon A \uplus B \to [n] \uplus [m]} \prod_{ab \in E(F)} x_{h(a)) h(b)} = \hom_{F, n,m}
		\]
		observing that the set of maps $h \colon A \uplus B \to [n] \uplus [m]$ is in bijection with the set of maps $(\pi, \sigma) \circ h$ for $h \colon A \uplus B \to [n] \uplus [m]$.
	\end{proof}
	\begin{remark} 
		\label{rem:squareSymmetricPolynomials}
	One can similarly characterise the $\Sym_n$-symmetric polynomials in terms of homomorphism/subgraph polynomials for \emph{directed} and not necessarily bipartite graphs. A $\Sym_n$-symmetric polynomial is one that is fixed by every $\pi \in \Sym_n$ acting on $\Xx_{n,n}$ by mapping $x_{ij}$ to $x_{\pi(i)\pi(j)}$. In fact, we believe that all or most of the theory we develop in this paper for $\Sym_n \times \Sym_m$-symmetric polynomials applies to $\Sym_n$-symmetric ones as well, but we prove our results only for the $\Sym_n \times \Sym_m$-symmetric setting.
	\end{remark}	
	Even though the definition of $\mathfrak{G}_{n,m}$ involves homomorphism polynomials of infinitely many graphs,
	it follows from \cref{lem:gnm-sym} that, for every $n,m\in \mathbb{N}$, finitely many such polynomials already generate $\mathfrak{G}_{n,m}$ as a $\mathbb{Q}$-algebra:
	
	\begin{corollary}
		For every $n,m \in \mathbb{N}$,
		there exist finitely many bipartite multigraphs $F_1, \dots, F_r$ such that, for every $p \in \mathfrak{G}_{n,m}$,
		there exists a polynomial $q \in \mathbb{Q}[y_1, \dots, y_r]$ such that 
		\(
		p = q(\hom_{F_1, n,m}, \dots, \hom_{F_r, n,m}).
		\)
	\end{corollary}
	\begin{proof}
		By the invariant theory of finite groups, cf.\ e.g.\ \cite[Proposition~3.0.1]{derksen_computational_2015}, 
		the set of $\Sym_n \times \Sym_m$-symmetric polynomials in $\mathbb{Q}[\mathcal{X}_{n,m}]$ is finitely generated as a $\mathbb{Q}$-algebra.
		By \cref{lem:gnm-sym}, each of its generators is a finite linear combination of homomorphism polynomials from some bipartite multigraphs.
		Hence, the set of $\Sym_n \times \Sym_m$-symmetric polynomials in $\mathbb{Q}[\mathcal{X}_{n,m}]$ is equal to the $\mathbb{Q}$-algebra generated by some finite set of homomorphism polynomials of bipartite multigraphs.
	\end{proof}
	
	We now turn to the characterisation of symmetric polynomials admitting symmetric circuits of polynomial orbit size.
	To that end, for $k,n,m \in \mathbb{N}$, write
	\[
		\mathfrak{T}_{n,m}^k \coloneqq \left\{\sum \alpha_i \hom_{F_i, n,m} \ \middle|\  \text{bipartite multigraphs } F_i \text{ such that } \tw(F_i) < k, \alpha_i \in \mathbb{Q} \right\} \subseteq \mathfrak{G}_{n,m}.
	\]
	for the set of finite $\mathbb{Q}$-linear combinations of homomorphism polynomials of bipartite multigraphs of treewidth less than~$k$.
	Our main theorem is the following:

	\thmMain*
	
	As a first step towards \cref{thm:main1}, we remark that homomorphism polynomials of graphs of bounded treewidth admit polynomial-size symmetric circuits. 
	
	\begin{theorem}[restate=thmHomCircuit, label=thm:hom-circuit-main, name=]
		Let $k, n, m \in \mathbb{N}$.
		Let $F$ be a bipartite multigraph of treewidth less than~$k$.
		Then there exists a $\Sym_n \times \Sym_{m}$-symmetric rigid circuit $C_F$ for $\hom_{F,n,m}$ of size at most $5k!^2k \cdot (n+m)^{k+1} \cdot \norm{F}^2$ satisfying $\maxSup(C_F) \leq k$.
	\end{theorem}
	
	\Cref{thm:hom-circuit-main} yields the forward implication in \cref{thm:main1}: If $p \in \mathfrak{T}_{n,m}^k$, then it is a linear combination of polynomials $\hom_{F,n,m}$, where $\tw(F) \leq k$ for each $F$ in the linear combination. With \cref{thm:hom-circuit-main}, we get a rigid symmetric circuit $C_F$ with $\maxSup(C_F) \leq k$ for each of these $\hom_{F,n,m}$. To obtain $p$, we just take a linear combination of these circuits and rigidify it again using \cref{lem:rigidifyCircuits}. This yields a rigid circuit $C$ representing $p$ that still satisfies $\maxSup(C) \leq k$. By \cref{lem:constantSupportImpliesPolyOrbit}, it has orbit size at most $(n+m)^k$.  
	
	The remainder of this section is concerned with the backward implication in \cref{thm:main1} and the proof of \cref{thm:hom-circuit-main}.

	\subsection{Homomorphism Polynomials of Labelled Graphs}

	The main tool for proving \cref{thm:main1} are homomorphism polynomials of labelled graphs.
		
	\begin{definition}\label{def:labelled-hom-poly}
		Let $\ell, r \in \mathbb{N}$.
		An \emph{$(\ell, r)$-labelled bipartite graph} $\boldsymbol{F} = (F, \boldsymbol{a}, \boldsymbol{b})$ is a tuple of a bipartite multigraph $F$ with bipartition $A \uplus B = V(F)$,
		an $\ell$-tuple $\boldsymbol{a} \in A^\ell$, and an $r$-tuple $\boldsymbol{b} \in B^r$.
		Write $\mathcal{G}(\ell, r)$ for the class of all $(\ell, r)$-labelled bipartite graphs.
		
		For $n, m \in \mathbb{N}$ and tuples $\boldsymbol{v} \in [n]^\ell$ and $\boldsymbol{w} \in [m]^r$, write
		\[
		\boldsymbol{F}_{n,m}(\boldsymbol{v}, \boldsymbol{w}) \coloneqq \sum_{\substack{h \colon A \uplus B \to [n] \uplus [m]\\ h(\boldsymbol{a}) = \boldsymbol{v} \\ h(\boldsymbol{b}) = \boldsymbol{w}}} \prod_{ab \in E(F)} x_{h(a)h(b)} \in \mathbb{Q}[\mathcal{X}_{n,m}]
		\]
		for the \emph{homomorphism polynomial of $\boldsymbol{F}$ at $\boldsymbol{v},  \boldsymbol{w}$}.
		The map $\boldsymbol{F}_{n,m} \colon [n]^\ell \times [m]^r \to \mathbb{Q}[\mathcal{X}_{n,m}]$ given by $(\boldsymbol{v}, \boldsymbol{w}) \mapsto \boldsymbol{F}_{n,m}(\boldsymbol{v}, \boldsymbol{w})$ is the \emph{homomorphism polynomial map} of $\boldsymbol{F}$. Write
		\[
		\mathfrak{G}_{n,m}(\ell, r) \coloneqq \left\{\sum \alpha_i \boldsymbol{F}^i_{n,m} \ \middle|\  \alpha_i \in \mathbb{Q}, \boldsymbol{F}^i \in \mathcal{G}(\ell, r)\right\} \subseteq \mathbb{Q}[\mathcal{X}_{n,m}]^{[n]^\ell \times [m]^r}
		\]
		for the set of finite $\mathbb{Q}$-linear combinations of homomorphism polynomial maps of $(\ell,r)$-labelled bipartite graphs.
	\end{definition}

	Note that for $\ell = 0 = r$, the homomorphism polynomial map $\boldsymbol{F}_{n,m}$ simplifies to the homomorphism polynomial $\hom_{F, n, m}$, i.e.\ essentially $\mathfrak{G}_{n,m}(0,0) = \mathfrak{G}_{n,m}$. 
	
	\begin{example}\label{ex:edge}
		Let $\boldsymbol{F} = (F, a, b)$ denote the $(1,1)$-labelled bipartite graph with $V(F) = \{a,b\}$ and $E(F) = \{ab\}$.
		For $n,m \in \mathbb{N}$, $v \in [n]$, and $w \in [m]$,
		we have $\boldsymbol{F}_{n,m}(v, w) = x_{vw} \in \mathbb{Q}[\mathcal{X}_{n,m}]$.
	\end{example}
	
	Next, we define what it means for a labelled bipartite graph to have bounded treewidth.
	Furthermore, we generalise $\mathfrak{T}^k_{n,m}$.

	\begin{definition}\label{def:treewidth-labelled}
		Let $\ell, r, k \in \mathbb{N}$.
		An \emph{$(\ell, r)$-labelled bipartite graph} $\boldsymbol{F} = (F, \boldsymbol{a}, \boldsymbol{b})$ has \emph{treewidth} less than $k$ if there exists a tree decomposition $(T, \beta)$ of $F$ of width less than $k$ with a vertex $s \in V(T)$ such that $a_1, \dots, a_\ell, b_1, \dots, b_r \in \beta(s)$.
		Write $\mathcal{T}^k(\ell, r) \subseteq \mathcal{G}(\ell, r)$ for the class of all $(\ell, r)$-labelled bipartite graphs of treewidth less than~$k$.
		For $n,m \in \mathbb{N}$, write
		\[
			\mathfrak{T}_{n,m}^k(\ell, r) \coloneqq \left\{\sum \alpha_i \boldsymbol{F}^i_{n,m} \ \middle|\  \alpha_i \in \mathbb{Q}, \boldsymbol{F}^i \in \mathcal{T}^k(\ell, r)\right\}
			\subseteq \mathfrak{G}_{n,m}(\ell, r)
		\]
		for the set of $\mathbb{Q}$-linear combinations of homomorphism polynomial maps of $(\ell,r)$-labelled bipartite graphs of treewidth less than~$k$.
	\end{definition}
	
	To reiterate, the $k$ in $\mathfrak{T}_{n,m}^k(\ell, r)$ indicates the maximal admissible size of a bag in a tree decomposition rather than the treewidth, which is defined as the maximal bag size minus one.
	This convention simplifies subsequent arguments.

	Note that essentially $\mathfrak{T}^k_{n,m} = \mathfrak{T}^k_{n,m}(0,0)$. For $\boldsymbol{v} \in [n]^\ell, \boldsymbol{w} \in [m]^r$, we also define the set of \emph{instantiated homomorphism polynomial maps} as  
	\(
	\mathfrak{T}_{n,m}^k(\boldsymbol{v},\boldsymbol{w}) \coloneqq \left\{ \phi(\boldsymbol{v}, \boldsymbol{w}) \ \middle|\  \phi \in \mathfrak{T}^k_{n,m}(\ell,r) \right\} \subseteq \mathbb{Q}[\mathcal{X}_{n,m}].
	\)

	\begin{example}\label{ex:one}
		For $\ell, r \in \mathbb{N}$,
		let $\boldsymbol{J} = (J, \boldsymbol{a}, \boldsymbol{b})$ be the $(\ell, r)$-labelled edge-less $(\ell, r)$-vertex bipartite graphs with labels residing on distinct vertices. 
		Clearly, $\boldsymbol{J} \in \mathcal{T}^{\ell+r}(\ell, r)$.
		For all $n, m \in \mathbb{N}$, $\boldsymbol{v} \in [n]^\ell$, and $\boldsymbol{w} \in [m]^r$,
		\[
			\boldsymbol{J}_{n,m}(\boldsymbol{v}, \boldsymbol{w}) \coloneqq \sum_{\substack{h \colon A \uplus B \to [n] \uplus [m]\\ h(\boldsymbol{a}) = \boldsymbol{v} \\ h(\boldsymbol{b}) = \boldsymbol{w}}} 1 = 1
		\]
		That is, $\boldsymbol{J}_{n,m} \in \mathfrak{T}^{\ell+r}_{n,m}(\ell, r)$ is the homomorphism polynomial map which maps all tuples to $1$.
	\end{example}

	\subsection{Operations on Labelled Graphs and Homomorphism Polynomial Maps}
	\label{sec:ops}
	
	Fix  $n,m, \ell, r, k \in \mathbb{N}$ throughout this section.
	In this section, we prove various closure properties of $\mathfrak{T}_{n,m}^k(\ell, r)$.
	\subsubsection{Swapping Left and Right Labels}
	
	We first note that the left and right labels can be interchanged.
	In order to simplify notation, we subsequently focus on operations involving only the left labels.
	The following purely syntactic lemma justifies this approach.
	To that end,
	define for $\boldsymbol{F} = (F, \boldsymbol{a}, \boldsymbol{b}) \in \mathcal{G}(\ell, r)$
	the $(r,\ell)$-labelled bipartite graph
	 $\boldsymbol{F}^* \coloneqq (F^*, \boldsymbol{b}, \boldsymbol{a})  \in \mathcal{G}(r, \ell)$ where $F^*$ is the bipartite multigraph obtained from $F$ by interchanging the parts of the bipartition.
	Analogously, for $\phi = \sum \alpha_i \boldsymbol{F}^i_{n,m} \in \mathfrak{G}_{n,m}(\ell, r)$,
	define $\phi^* \coloneqq  \sum \alpha_i (\boldsymbol{F}^i)^*_{m,n} \in \mathfrak{G}_{m,n}(r, \ell)$.
        The following
        is then immediate.
	%TODO: Think about what this is really needed for?
	\begin{lemma}\label{lem:reverse}
		Let $\boldsymbol{F} \in \mathcal{G}(\ell, r)$
		and $\phi \in \mathfrak{G}_{n,m}(\ell, r)$.
		\begin{enumerate}
			\item If $\boldsymbol{F} \in \mathcal{T}^k(\ell, r)$,
			then $\boldsymbol{F}^* \in \mathcal{T}^k(r, \ell)$.
			\item For $\boldsymbol{v} \in [n]^\ell$ and
                          $\boldsymbol{w} \in [m]^r$, it holds that $\boldsymbol{F}_{n,m}(\boldsymbol{v}, \boldsymbol{w}) = \boldsymbol{F}^*_{m,n}(\boldsymbol{w}, \boldsymbol{v})$.
			\item If $\phi \in \mathfrak{T}^k_{n,m}(\ell, r)$, then $\phi^* \in \mathfrak{T}^k_{m,n}(r, \ell)$.
		\end{enumerate}
		
	\end{lemma}
	
	\subsubsection{Unlabelling and Sums}
	Next, we consider the unlabelling operation.
	For $\ell \geq 1$, $\boldsymbol{F} = (F, \boldsymbol{a}, \boldsymbol{b}) \in \mathcal{G}(\ell, r)$, and $i \in [\ell]$,
	define $\Sigma_i \boldsymbol{F} \coloneqq (F, \boldsymbol{a}[i/], \boldsymbol{b}) \in \mathcal{G}(\ell-1, r)$
	as the graph obtained by dropping the $i$-th left label.
	Analogously, 
	for $\phi \in \mathfrak{G}_{n,m}(\ell, r)$,
	define $\Sigma_i\phi \in \mathfrak{G}_{n,m}(\ell-1, r)$ via $(\Sigma_i \phi)(\boldsymbol{v}[i/], \boldsymbol{w}) \coloneqq \sum_{v \in [n]}\phi(\boldsymbol{v}[i/v], \boldsymbol{w})$ for all $\boldsymbol{v} \in [n]^\ell$ and $\boldsymbol{w} \in [m]^r$. 

	\begin{lemma}\label{lem:unlabelling}
		For $\ell \geq 1$,
		let $\boldsymbol{F} \in \mathcal{G}(\ell, r)$
		and $\phi \in \mathfrak{G}_{n,m}(\ell, r)$.
		\begin{enumerate}
			\item If $\boldsymbol{F} \in \mathcal{T}^k(\ell, r)$,
			then $\Sigma_i\boldsymbol{F} \in \mathcal{T}^k(\ell-1, r)$.
			\item For $\boldsymbol{v} \in [n]^\ell$ and $\boldsymbol{w} \in [m]^r$,  it is
					\[
			(\Sigma_i \boldsymbol{F})_{n,m}(\boldsymbol{v}[i/],\boldsymbol{w}) = \sum_{v \in [n]} \boldsymbol{F}_{n,m}(\boldsymbol{v}[i/v], \boldsymbol{w}).
			\]
			\item If $\phi \in \mathfrak{T}^k_{n,m}(\ell, r)$, then $\Sigma_i \phi \in \mathfrak{T}^k_{n,m}(\ell-1, r)$. 
		\end{enumerate}
	\end{lemma}

	\subsubsection{Disjoint Union and Tensor Products}

	The \emph{disjoint union} of labelled graphs is defined as follows:
	For elements $\boldsymbol{F} = (F, \boldsymbol{a}, \boldsymbol{b}) \in \mathcal{G}(\ell, r)$ and $\boldsymbol{F}' = (F', \boldsymbol{a}', \boldsymbol{b}') \in \mathcal{G}(\ell', r')$,
	define $\boldsymbol{F} \otimes \boldsymbol{F}' \in \mathcal{G}(\ell+\ell', r+r')$
	as $(F \uplus F', \boldsymbol{a}\boldsymbol{a}', \boldsymbol{b}\boldsymbol{b}')$.
	Analogously, for $\phi \in \mathfrak{G}(\ell ,r )$ and $\psi \in \mathfrak{G}(\ell', r')$,
	define $\phi \otimes \psi \in \mathfrak{G}(\ell+\ell', r+r')$ by $(\phi \otimes \psi)(\boldsymbol{v}, \boldsymbol{w}) \coloneqq \phi(v_1 \dots v_\ell, w_1 \dots w_\ell) \cdot \psi(v_{\ell+1} \dots v_{\ell+\ell'}, w_{r+1} \dots w_{r+r'})$ for all $\boldsymbol{v} \in [n]^{\ell+\ell'}$ and $\boldsymbol{w} \in [m]^{r+r'}$.
	
	We use $\otimes$ to denote disjoint unions in order to be aligned with the notation of \cite{mancinska_quantum_2020,seppelt_homomorphism_2024} for homomorphism tensors where the disjoint union of labelled graphs corresponds to the Kronecker product of their homomorphism tensors.
	\begin{lemma}\label{lem:disjoint-union}
		For $\ell, \ell', r,r', k, k', n, m \in \mathbb{N}$,
		let $\boldsymbol{F} \in \mathcal{G}(\ell, r)$,  $\boldsymbol{F}' \in \mathcal{G}(\ell', r')$,  $\phi \in \mathfrak{G}(\ell, r)$, and $\psi \in \mathfrak{G}(\ell', r')$.
		\begin{enumerate}
			\item If $\boldsymbol{F} \in \mathcal{T}^k(\ell, r)$ and $\boldsymbol{F}' \in \mathcal{T}^{k'}(\ell', r')$, then $\boldsymbol{F} \otimes \boldsymbol{F}' \in \mathcal{T}^{k''}(\ell+\ell', r+r')$ for $k'' \coloneqq \max\{k,k', \ell+\ell'+r+r'\}$.
			\item For $\boldsymbol{v} \in [n]^{\ell+\ell'}$ and $\boldsymbol{w} \in [m]^{r+r'}$, 
			\[
				(\boldsymbol{F} \otimes \boldsymbol{F}')_{n,m}(\boldsymbol{v}, \boldsymbol{w}) = \boldsymbol{F}_{n,m}(v_1 \dots v_\ell, w_1 \dots w_r) \cdot \boldsymbol{F}'_{n,m}(v_{\ell+1} \dots v_{\ell+ \ell'}, w_{r+1} \dots w_{r+r'}).
			\]
			\item If $\phi \in \mathfrak{T}^k(\ell, r)$ and $\phi' \in \mathfrak{T}^{k'}(\ell', r')$,
			then $\phi \otimes \phi' \in \mathfrak{T}^{k''}(\ell+\ell', r+r')$ for $k'' \coloneqq \max\{k,k', \ell+\ell'+r+r'\}$.
		\end{enumerate}
	\end{lemma}
	\begin{proof}
		For the first assertion, note that a tree decomposition for $\boldsymbol{F} \otimes \boldsymbol{F}'$ can be constructed by making the root nodes of the tree decompositions of $\boldsymbol{F}$ and $\boldsymbol{F}'$ adjacent to a fresh node containing the vertices labelled in $\boldsymbol{F} \otimes \boldsymbol{F}'$.
		This yields a tree decomposition with bags of size at most $k'' \coloneqq \max\{k,k', \ell+\ell'+r+r'\}$.
	\end{proof}

	\subsubsection{Gluing and Point-wise Products}
	The final elementary operation is gluing.
	For elements $\boldsymbol{F} = (F, \boldsymbol{a}, \boldsymbol{b})$ and $\boldsymbol{F}' = (F', \boldsymbol{a}', \boldsymbol{b}')$ of $\mathcal{G}(\ell, r)$,
	define $\boldsymbol{F} \odot \boldsymbol{F}' \in \mathcal{G}(\ell, r)$
	by taking the disjoint union of $F$ and $F'$ and identifying for $i \in [\ell]$ and $j \in [r]$ the vertices $a_i$ with $a'_i$ and $b_i$ with $b'_i$.
	Since left/right labels come from the left/right side of the bipartitions the resulting graph is indeed bipartite.
	If $\ell = 0 = r$, then gluing amounts to taking a disjoint union.
	
	For $\phi, \phi' \in \mathfrak{G}_{n,m}(\ell, r)$, define  the point-wise product $\phi \odot \phi' \in \mathfrak{G}_{n,m}(\ell, r)$ via $(\phi \odot \phi')(\boldsymbol{v}, \boldsymbol{w}) \coloneqq \phi(\boldsymbol{v}, \boldsymbol{w}) \phi'(\boldsymbol{v}, \boldsymbol{w})$.
	When convenient, we write $\phi \phi'$ for the point-wise product $\phi \odot \phi'$.

	\begin{lemma}\label{lem:gluing}
		Let $\boldsymbol{F}, \boldsymbol{F}' \in \mathcal{G}(\ell, r)$.
		Let $\phi, \phi' \in \mathfrak{T}^k_{n,m}(\ell, r)$.
		\begin{enumerate}
			\item If $\boldsymbol{F}, \boldsymbol{F}' \in \mathcal{T}^k(\ell, r)$,
			then $\boldsymbol{F} \odot \boldsymbol{F}' \in \mathcal{T}^k(\ell, r)$.
			\item For $\boldsymbol{v} \in [n]^\ell$ and $\boldsymbol{w} \in [m]^r$,  			\[
				(\boldsymbol{F} \odot \boldsymbol{F}')_{n,m}(\boldsymbol{v}, \boldsymbol{w})  = \boldsymbol{F}_{n,m}(\boldsymbol{v}, \boldsymbol{w}) \cdot \boldsymbol{F}'_{n,m}(\boldsymbol{v}, \boldsymbol{w}).
			\]
			\item The point-wise product $\phi \odot \phi'$ of $\phi$ and $\phi'$ is in $\mathfrak{T}^k_{n,m}(\ell, r)$.
		\end{enumerate}
	\end{lemma}
	\begin{proof}
		For a detailed proof, see \cite[Lemma 6.13]{grohe_homomorphism_2025} and \cite[Lemma 3.2.13]{seppelt_homomorphism_2024}.
		Intuitively, for the first claim, if $(T, \beta)$ and $(T', \beta')$ are tree decompositions of $\boldsymbol{F}$ and $\boldsymbol{F}'$ of width less than $k$ with all labelled vertices occurring in the bags of vertices $r \in V(T)$ and $r' \in V(T')$, 
		then are tree decomposition for $\boldsymbol{F} \odot \boldsymbol{F}'$ can be constructed by taking the disjoint union of $T$, $T'$, and a fresh vertex $x$ and connecting $x$ to $r$ and $r'$.
		The bags of the new decomposition are induced by $\beta$ on $T$ and by $\beta'$ on $T'$.
		The bag associated with $x$ contains the labelled vertices in $\boldsymbol{F} \odot \boldsymbol{F}'$.
	\end{proof}

	\Cref{lem:gluing} shows that $\mathfrak{T}^k_{n,m}(\ell,r)$ forms a $\mathbb{Q}$-algebra for all $k,n,m, \ell, r \in \mathbb{N}$.
	If $\ell +r \leq k$, then this algebra is unital by \cref{ex:one}.

 	\subsubsection{Restricted Sums}
	As a modification of \cref{lem:unlabelling},
	we consider the following restricted sum operation.
	It is needed when we analyse the semantics of summation gates in dependence of their support.
	Let $i \in [\ell]$ and $J \subseteq [\ell] \setminus \{i\}$.
	For $\phi \in \mathfrak{G}_{n,m}(\ell, r)$,
	define $\Sigma_{i, J} \phi \in \mathfrak{G}_{n,m}(\ell-1, r)$
	via \[ (\Sigma_{i, J} \phi)(\boldsymbol{v}[i/], \boldsymbol{w}) \coloneqq \sum_{v \in [n] \setminus \{v_j \mid j \in J\}} \phi(\boldsymbol{v}[i/v], \boldsymbol{w}). \]
	for all $\boldsymbol{v} \in [n]^\ell$ and $\boldsymbol{w} \in [m]^r$.
	Note that $\Sigma_i \phi = \Sigma_{i, \emptyset} \phi$.
	
	\begin{lemma}\label{lem:sum-exclude-lincomb}
		Suppose that $\ell+r \leq k$.
		For $\phi \in \mathfrak{T}^k_{n,m}(\ell,r )$,  $i \in [\ell]$, and $J \subseteq [\ell] \setminus \{i\}$,
		we have that $\Sigma_{i, J} (\phi) \in \mathfrak{T}^k_{n,m}(\ell-1, r)$.
	\end{lemma}
	\begin{proof}
		For distinct $i, j \in [\ell]$, 
		let $\boldsymbol{D}^{i,j} \in \mathcal{T}^{k}(\ell, r)$ be the $(\ell, r)$-labelled edge-less $(\ell-1,r)$-vertex bipartite graph whose labels are placed such that the $i$-th and the $j$-th left label
		reside on the same vertex while all other labels are carried by distinct vertices.
		Note that, for all $\boldsymbol{v} \in [n]^{\ell}$ and $\boldsymbol{w} \in [m]^r$,
		\begin{equation}\label{eq:diell}
			\boldsymbol{D}^{i,j}_{n,m}(\boldsymbol{v}, \boldsymbol{w}) = \begin{cases}
				1,& v_i = v_j, \\
				0,& \text{otherwise}.
			\end{cases}
		\end{equation}
		Let $p \in \mathbb{Q}[x]$ be a polynomial such that $p(0) = 1$ and $p(1) = \dots = p(\ell) = 0$.
		By applying \cref{lem:gluing,ex:one},
		define $\delta \coloneqq p(\sum_{j \in J} \boldsymbol{D}^{i,j}_{n,m}) \in \mathfrak{T}^k_{n,m}(\ell, r)$.
		Note that, for all $\boldsymbol{v} \in [n]^{\ell}$ and $\boldsymbol{w} \in [m]^r$,
		\[
			\delta(\boldsymbol{v}, \boldsymbol{w}) = \begin{cases}
				1, & v_i  \not\in \{v_j \mid j \in J\},\\
				0, & \text{otherwise}.
			\end{cases}
		\]
		Hence, $\sum_{i, J} \phi = \sum_{i} (\delta \odot \phi)$.
		The claim follows from \cref{lem:gluing,lem:unlabelling}.
	\end{proof}

	\subsubsection{Products}
	
	Suppose that $\ell \geq 1$.
	In this section, we consider the following product operator:
	For $\phi \in \mathfrak{G}_{n,m}(\ell, r)$ and $i \in [\ell]$,
	define $\Pi_i \phi \in \mathfrak{G}_{n,m}(\ell-1, r)$ via 
	\[ (\Pi_i \phi)(\boldsymbol{v}[i/], \boldsymbol{w}) \coloneqq \prod_{v \in [n]} \phi(\boldsymbol{v}[i/v], \boldsymbol{w}) \] 
	for all $\boldsymbol{v} \in [n]^\ell$ and $\boldsymbol{w} \in [m]^k$.
	Note that, in contrast to $\Sigma_i$, and $\odot$,
	which are linear and bilinear, respectively,
	the operator $\Pi_i$ is not linear.
	Despite that, the following \cref{thm:product-lincomb} shows that applying $\Pi_i$ to a linear combination of homomorphism polynomial maps yields a linear combination of homomorphism polynomial maps.
	This theorem is the most important technical novelty in the proof of \cref{thm:main1}.
	
	\begin{theorem}\label{thm:product-lincomb}
		For $\phi \in \mathfrak{T}^k_{n,m}(\ell,r )$ and $i \in [\ell]$,
		we have that $\Pi_i\phi \in \mathfrak{T}^k_{n,m}(\ell-1, r)$.
	\end{theorem}

	\Cref{thm:product-lincomb} is proved in several steps.
	We first show in \cref{lem:product-one-graph} that $\Pi_i \boldsymbol{F}_{n,m} \in \mathfrak{T}^k_{n,m}(\ell-1, r)$ for all $\boldsymbol{F} \in \mathcal{T}^k(\ell, r)$.
	To that end, consider the following lemma.
	
	\begin{lemma}\label{lem:pi-i-pi}
		Let $\boldsymbol{F}^1, \dots, \boldsymbol{F}^n \in \mathcal{T}^k(\ell, r)$ and $i \in [\ell]$.
		For a partition $\pi$ of $[n]$,
		define the $(\ell-1, r)$-labelled bipartite graph
		\(
			\Pi_i^\pi(\boldsymbol{F}^1, \dots, \boldsymbol{F}^n) \coloneqq \bigodot_{P \in [n]/\pi} \Sigma_i(\bigodot_{v \in P} \boldsymbol{F}^v) \in \mathcal{G}(\ell-1, r).
		\)
		Then
		\[
			(\Pi_i^\pi(\boldsymbol{F}^1, \dots, \boldsymbol{F}^n))_{n,m}(\boldsymbol{v}[i/], \boldsymbol{w}) = \sum_{h \colon [n]/\pi \to [n]}  \prod_{v \in [n]} \boldsymbol{F}^v_{n, m}(\boldsymbol{v}[i/(h\circ \pi)(v)], \boldsymbol{w})
		\]
		for all $\boldsymbol{v} \in [n]^\ell$ and $\boldsymbol{w} \in [m]^r$.
		Furthermore, $\Pi_i^\pi(\boldsymbol{F}^1, \dots, \boldsymbol{F}^n) \in \mathcal{T}^k(\ell-1, r) $. 
	\end{lemma}
	\begin{proof}
		The final assertion follows from \cref{lem:gluing,lem:unlabelling}.
		For the first assertion, observe that, by \cref{lem:gluing,lem:unlabelling},
		\begin{align*}
			(\Pi_i^\pi(\boldsymbol{F}^1, \dots, \boldsymbol{F}^n))_{n,m}(\boldsymbol{v}[i/], \boldsymbol{w})
		%	&= \prod_{P \in [n]/\pi} \Sigma_i\left(\prod_{v \in P} \boldsymbol{F}^v_{n,m}(\boldsymbol{v}, \boldsymbol{w}) \right) \\
			&= \prod_{P \in [n]/\pi} \sum_{v' \in [n]} \prod_{v \in P} \boldsymbol{F}^v_{n,m}(\boldsymbol{v}[i/v'], \boldsymbol{w}) \\
			&= \sum_{h \colon [n]/\pi \to [n]}  \prod_{P \in [n]/\pi} \prod_{v \in P} \boldsymbol{F}^v_{n,m}(\boldsymbol{v}[i/h(P)], \boldsymbol{w}) \\
			&= \sum_{h \colon [n]/\pi \to [n]}  \prod_{v \in [n]} \boldsymbol{F}_{n, m}^v(\boldsymbol{v}[i/(h\circ \pi)(v)], \boldsymbol{w}). \qedhere
		\end{align*}
	\end{proof}
	
	\Cref{lem:pi-i-pi} is used to prove the following \cref{lem:product-one-graph}.
	Here,  the $\boldsymbol{F}^1, \dots, \boldsymbol{F}^n$ are taken to be all the same $(\ell, r)$-labelled bipartite graph $\boldsymbol{F}$.
	We abbreviate $\Pi_i^\pi \boldsymbol{F} \coloneqq \Pi_i^\pi(\boldsymbol{F}, \dots, \boldsymbol{F})$.
	\Cref{lem:pi-i-pi} then yields
	\begin{equation}\label{eq:pi-i-pi-simplified}
		(\Pi_i^\pi \boldsymbol{F})_{n,m}(\boldsymbol{v}[i/], \boldsymbol{w}) = \sum_{h \colon [n]/\pi \to [n]}  \prod_{v \in [n]} \boldsymbol{F}_{n, m}(\boldsymbol{v}[i/(h\circ \pi)(v)], \boldsymbol{w})
	\end{equation}
	
	\begin{lemma}\label{lem:product-one-graph}
		For $\boldsymbol{F} \in \mathcal{T}^k(\ell, r)$ and $i \in [\ell]$,
		 $\Pi_i (\boldsymbol{F}_{n,m}) \in \mathfrak{T}^k_{n, m}(\ell-1, r)$.
	\end{lemma}
	\begin{proof}
		Write $\mu$ for the Möbius function of the partition lattice of $[n]$, cf.\ \cref{eq:frs}.
		More concretely, 
		we show that $\Pi_i \boldsymbol{F}_{n,m} = \frac{1}{n!} \sum_{\pi \vdash n} \mu_\pi (\Pi_i^\pi \boldsymbol{F})_{n,m}$.
		By \cref{lem:pi-i-pi},  $\Pi^\pi_i\boldsymbol{F} \in \mathcal{T}^k(\ell-1, r)$.
		Hence, it remains to show that
		\[
		\prod_{v \in [n]} \boldsymbol{F}_{n,m}(\boldsymbol{v}[i/v], \boldsymbol{w})
		=  \frac{1}{n!} \sum_{\pi \vdash n} \mu_{\pi} (\Pi^\pi_i\boldsymbol{F})_{n, m}(\boldsymbol{v}[i/], \boldsymbol{w})
		\]
		for all $\boldsymbol{v} \in [n]^\ell$ and $\boldsymbol{w} \in [m]^r$.
		Note that, crucially, the coefficients in the above expression do not depend on $\boldsymbol{v} \in [n]^\ell$ and $\boldsymbol{w} \in [m]^r$.

		The left hand-side expression can be rewritten as follows noting that, for every bijection $h \colon [n] \hookrightarrow [n]$, it is 
		$\prod_{v \in [n]} \boldsymbol{F}_{n,m}(\boldsymbol{v}[i/v], \boldsymbol{w}) = \prod_{v \in [n]} \boldsymbol{F}_{n,m}(\boldsymbol{v}[i/h(v)], \boldsymbol{w})$
		by commutativity of the product.
		The second equality follows from \cref{lem:moebius-polynomial},
		the third from \cref{eq:pi-i-pi-simplified}.
		\begin{align*}
			\prod_{v \in [n]} \boldsymbol{F}_{n,m}(\boldsymbol{v}[i/v], \boldsymbol{w})
			&= \frac{1}{n!} \sum_{h \colon [n] \hookrightarrow [n]} \prod_{v \in [n]} \boldsymbol{F}_{n,m}(\boldsymbol{v}[i/h(v)], \boldsymbol{w}) \\
			&= \frac{1}{n!} \sum_{\pi \vdash n} \mu_\pi \sum_{h \colon [n]/\pi \to [n]} \prod_{v \in [n]} \boldsymbol{F}_{n,m}(\boldsymbol{v}[i/(h \circ \pi)(v)], \boldsymbol{w}) \\
			&= \frac{1}{n!} \sum_{\pi \vdash n} \mu_\pi (\Pi^\pi_i \boldsymbol{F})_{n,m}(\boldsymbol{v}[i/], \boldsymbol{w}). \qedhere
		\end{align*}
	\end{proof}

	Now we consider the general case: applying the product operator to a linear combination of labelled homomorphism polynomials.
	In preparation, we note the following \cref{fact:orbits}.

	\begin{fact}\label{fact:orbits}
		Let $\Gamma$ denote a finite group acting on a finite set $X$.
		For an element $x \in X$, write $\beta_x$ for the number of pairs $(z, \gamma) \in X \times \Gamma$ such that $\gamma(z) = x$.
		If $x, y \in X$ are in the same orbit under $\Gamma$,
		then $\beta_x = \beta_y$.
	\end{fact}
	\begin{proof}
		Suppose that $\iota(x) = y$ for some $\iota \in \Gamma$.
		Then the map $(z, \gamma) \mapsto (z, \iota \gamma )$ sends a pair $(z, \gamma)$ such that $\gamma(z) = x$ to the pair $(z, \iota \gamma)$ satisfying $\iota \gamma (z) = \iota(x) = y$.
		This map is bijective.
	\end{proof}

	\begin{proof}[Proof of \cref{thm:product-lincomb}]
		Write $\phi = \sum_{t \in T} \alpha_t \boldsymbol{F}^t_{n,m}$ for some finite set $T$,
		coefficients $\alpha_t \in \mathbb{Q}$ and $\boldsymbol{F}^t \in \mathcal{T}^k(\ell, r)$.
		The symmetric group $\Sym_n$ acts on the set of functions $f \colon [n] \to T$ by composition.
		The orbits of this action are in bijection with maps $\lambda \colon T \to \{0, \dots, n\}$ such that $\sum_{t\in T} \lambda(t) = n$.
		Write $\Lambda$ for the set of all such functions.
		A function $f \colon [n] \to T$ belongs to the orbit $O_\lambda$ for $\lambda \in \Lambda$ if $|f^{-1}(t)| = \lambda(t)$ for all $t \in T$.
		
		Let $\boldsymbol{v} \in [n]^\ell$ and $\boldsymbol{w} \in [m]^r$.
		First, product and sum are interchanged. The resulting sum over all functions $f \colon [n] \to T$ is then grouped by orbits:
		\begin{align*}
			(\Pi_i \phi)(\boldsymbol{v}[i/], \boldsymbol{w})
			&= \prod_{v \in [n]} \sum_{t \in T} \alpha_t \boldsymbol{F}_{n,m}^t(\boldsymbol{v}[i/v], \boldsymbol{w}) \\
			&= \sum_{f \colon [n] \to T} \prod_{v \in [n]} \alpha_{f(v)} \boldsymbol{F}_{n,m}^{f(v)}(\boldsymbol{v}[i/v], \boldsymbol{w}) \\
			&= \sum_{\lambda \in \Lambda} \sum_{f \in O_\lambda} \prod_{v \in [n]} \alpha_{f(v)} \boldsymbol{F}_{n,m}^{f(v)}(\boldsymbol{v}[i/v], \boldsymbol{w}) \\
			&= \sum_{\lambda \in \Lambda} \alpha_\lambda  \sum_{f \in O_\lambda} \prod_{v \in [n]} \boldsymbol{F}_{n,m}^{f(v)}(\boldsymbol{v}[i/v], \boldsymbol{w}).
		\end{align*}
		where $\alpha_\lambda \coloneqq \prod_{t \in T} \alpha_t^{\lambda(t)} \in \mathbb{Q}$.
		
		From now on, we consider each orbit $\lambda \in \Lambda$ separately.
		By \cref{fact:orbits}, there exists a positive number $\beta_\lambda \in \mathbb{N}$ such that, for every $f \in O_\lambda$, the number of pairs $g \colon [n] \to T$ and $h \in \Sym_n$ such that $g \circ h = f$ is $\beta_\lambda$.	
		In any such pair, we necessarily have $g \in O_\lambda$.
		Towards applying Möbius inversion, we denote permutations $h \in \Sym_n$ as injective maps $[n] \hookrightarrow [n]$.
		\begin{align*}
			\beta_\lambda \sum_{f \in O_\lambda} \prod_{v \in [n]} \boldsymbol{F}_{n,m}^{f(v)}(\boldsymbol{v}[i/v], \boldsymbol{w})
			&= \sum_{g \in O_\lambda} \sum_{h \colon [n] \hookrightarrow [n]} \prod_{v \in [n]} \boldsymbol{F}_{n,m}^{(g \circ h)(v)}(\boldsymbol{v}[i/v], \boldsymbol{w}) \\
			&= \sum_{g \in O_\lambda} \sum_{h \colon [n] \hookrightarrow [n]} \prod_{v \in [n]} \boldsymbol{F}_{n,m}^{g(v)}(\boldsymbol{v}[i/h(v)], \boldsymbol{w}) \\
			&= \sum_{g \in O_\lambda} \sum_{\pi \vdash n} \mu_\pi \sum_{h \colon [n]/\pi \to [n]} \prod_{v \in [n]} \boldsymbol{F}_{n,m}^{g(v)}(\boldsymbol{v}[i/(h \circ \pi)(v)], \boldsymbol{w}) \\
			&= \sum_{g \in O_\lambda} \sum_{\pi \vdash n} \mu_\pi (\Pi^\pi_i(\boldsymbol{F}^{g(1)}, \dots, \boldsymbol{F}^{g(n)}))_{n,m}(\boldsymbol{v}[i/], \boldsymbol{w})
		\end{align*}
		Here, the first equality follows from \cref{fact:orbits}.
		The second equality is obtained by first replacing the product over $v \in [n]$ by a product over $h(v)$ for $v \in [n]$ and then replacing the sum over $h \colon [n] \hookrightarrow [n]$ by a sum over $h^{-1}$ for $h \colon [n] \hookrightarrow [n]$.
		The third equality follows from \cref{lem:moebius-polynomial}.
		The final equality is implied by \cref{lem:pi-i-pi}.
		
		Combining the identities above, we obtain
		\begin{equation}\label{eq:product-lincomb}
		(\Pi_i \phi)(\boldsymbol{v}[i/], \boldsymbol{w})
		= \sum_{\lambda \in \Lambda} \frac{\alpha_\lambda}{\beta_\lambda} \sum_{g \in O_\lambda} \sum_{\pi \vdash n} \mu_\pi \Pi_i^\pi(\boldsymbol{F}^{g(1)}, \dots, \boldsymbol{F}^{g(n)})_{n,m}(\boldsymbol{v}[i/], \boldsymbol{w}).
		\end{equation}
		Note that neither the $\alpha_\lambda$ nor $\beta_\lambda$ depend on $\boldsymbol{v}$.
		Hence, $\Pi_i \phi \in \mathfrak{T}^k_{n,m}(\ell-1, r)$.
	\end{proof}

	\begin{remark}
		\cref{thm:product-lincomb} indeed generalises \cref{lem:product-one-graph}.
		If $T = \{\boldsymbol{F}\}$ is a singleton and there are no coefficients,
		then $\Lambda$ consists of a single map $\lambda$.
		Then, $\alpha_\lambda = 1$ and  $\beta_\lambda = n!$.
		The orbit $O_\lambda$ contains only one map $g$ and $\Pi_i^\pi(\boldsymbol{F}^{g(1)}, \dots, \boldsymbol{F}^{g(n)}) = \Pi_i^\pi \boldsymbol{F}$.
	\end{remark}

	\subsubsection{Restricted Products}
	
	Finally, we derive a corollary of \cref{thm:product-lincomb} concerned with the restricted product operator $\Pi_{i, J}$ which is defined as follows:
	Let $i \in [\ell]$ and $J \subseteq [\ell] \setminus \{i\}$.
	For $\phi \in \mathfrak{G}_{n,m}(\ell, r)$,
	define $\Pi_{i, J} \phi \in \mathfrak{G}_{n,m}(\ell-1, r)$
	via \[ (\Pi_{i, J} \phi)(\boldsymbol{v}[i/], \boldsymbol{w}) \coloneqq \prod_{v \in [n] \setminus \{v_j \mid j \in J\}} \phi(\boldsymbol{v}[i/v], \boldsymbol{w}). \]
	for all $\boldsymbol{v} \in [n]^\ell$ and $\boldsymbol{w} \in [m]^r$.
	Note that $\Pi_i \phi = \Pi_{i, \emptyset} \phi$.

	\begin{corollary}\label{cor:product-exclude-lincomb}
		Suppose that $\ell+r \leq k$.
		For $\phi \in \mathfrak{T}^k_{n,m}(\ell,r )$,  $i \in [\ell]$ and $J \subseteq [\ell] \setminus \{i\}$,
		 $\Pi_{i, J} (\phi) \in \mathfrak{T}^k_{n,m}(\ell-1, r)$.
	\end{corollary}

	\begin{proof}
		For $J \subseteq [\ell]$, write $V_J \coloneqq [n] \setminus \{v_j \mid j \in J\}$.
		The proof is by induction on $|J|$.
		The base case $|J| = 0$ is established in \cref{thm:product-lincomb}.
		
		For the inductive step, let $j \in J$ and write $J' \coloneqq J \setminus \{j\}$.
		We first distinguish cases based on the form of the tuple $\boldsymbol{v}$.
		This is a priori problematic since the desired expression is required not to depend on $\boldsymbol{v}$.
		We overcome this problem by implementing the case distinction itself using homomorphism polynomial maps.
		
		Write $\phi = \sum_{t \in T} \alpha_t \boldsymbol{F}^t$ and
		suppose first that $v_j \not\in \{v_{j'} \mid j' \in J'\}$, i.e.\ $V_{J'} = V_{J} \uplus \{v_j\}$.
		Recall the definition of $\boldsymbol{D}^{i,j} \in \mathcal{T}^k(\ell, r)$ from the proof of \cref{lem:sum-exclude-lincomb}.
		We rewrite $\Pi_{i, J} (\phi)$ in terms of expressions to which the inductive hypothesis applies by considering the following polynomial:
		\begin{align*}
			&\prod_{v \in V_{J'}}
			\left( \phi(\boldsymbol{v}[i/v], \boldsymbol{w}) + \boldsymbol{D}^{i, j}_{n, m}(\boldsymbol{v}[i/v], \boldsymbol{w}) \right)  \\
			&= \left( \phi(\boldsymbol{v}[i/v_j], \boldsymbol{w}) + \boldsymbol{D}^{i, j}_{n, m}(\boldsymbol{v}[i/v_j], \boldsymbol{w}) \right)  \prod_{v \in V_{J}}
			\left( \phi(\boldsymbol{v}[i/v], \boldsymbol{w}) + \boldsymbol{D}^{i, j}_{n, m}(\boldsymbol{v}[i/v], \boldsymbol{w}) \right) \\
			&\overset{\eqref{eq:diell}}{=} \left( \phi(\boldsymbol{v}[i/v_j], \boldsymbol{w}) + 1 \right) \prod_{v \in V_{J}} \phi(\boldsymbol{v}[i/v], \boldsymbol{w}) \\
			&= \prod_{v \in V_{J'}}  \phi(\boldsymbol{v}[i/v], \boldsymbol{w}) + \prod_{v \in V_{J}} \phi(\boldsymbol{v}[i/v], \boldsymbol{w}).
		\end{align*}
		If otherwise $v_j \in \{v_{j'} \mid j' \in J'\}$,
		then the expression for $J$ equals the expression for $J'$.
		Formally,
		\[
			\prod_{v \in V_{J'}}\phi(\boldsymbol{v}[i/v], \boldsymbol{w}) 
			= \prod_{v \in V_{J}}\phi(\boldsymbol{v}[i/v], \boldsymbol{w}). 
		\]
		In order to combine the two cases considered above,
		observe that, for every  $\boldsymbol{v} \in [n]^{\ell}$ and $\boldsymbol{w} \in [m]^r$, by \cref{lem:unlabelling},
		\[
			\frac1n \sum_{j' \in J'} (\Sigma_i\boldsymbol{D}^{j', j})_{n,m}(\boldsymbol{v}[i/], \boldsymbol{w}) 
			= \frac1n \sum_{v \in [n]} \sum_{j' \in J'}\boldsymbol{D}^{j', j}_{n,m}(\boldsymbol{v}[i/v], \boldsymbol{w})
			= \sum_{j' \in J'} \boldsymbol{D}^{j', j}_{n,m}(\boldsymbol{v}, \boldsymbol{w})
		\]
		is equal to the number of indices $j' \in J'$ such that $v_j = v_{j'}$.
		This is a value from $\{0, \dots, \ell\}$.
		Here, it is crucial that $i \not\in J$, as assumed.
		Let $p \in \mathbb{Q}[x]$ denote a polynomial such that $p(0) = 0$ and $p(1) = p(2) = \dots = p(\ell) = 1$.
		By \cref{lem:unlabelling,lem:gluing},
		$\psi \coloneqq p(\frac1n \sum_{j' \in J'} (\Sigma_i\boldsymbol{D}^{j', j})_{n,m}) \in \mathfrak{T}_{n,m}^k(\ell-1, r)$.
		Then
		\begin{equation*}
			\psi(\boldsymbol{v}[i/], \boldsymbol{w}) = \begin{cases}
				1, & v_j \in \{v_{j'} \mid j' \in J'\}, \\
				0, & \text{otherwise}.
			\end{cases}
		\end{equation*}
		Combining the above,
		it follows that the expression $\prod_{v \in V_{J}}
		\phi(\boldsymbol{v}[i/v], \boldsymbol{w})$ of interest is equal to
		\begin{align*}
			&\psi(\boldsymbol{v}[i/], \boldsymbol{w})
			\prod_{v \in V_{J'}} \phi(\boldsymbol{v}[i/v], \boldsymbol{w}) \\
			&+
			(1-\psi(\boldsymbol{v}[i/], \boldsymbol{w}))
			\left(\prod_{v \in V_{J'}}
			\left( \phi(\boldsymbol{v}[i/v], \boldsymbol{w}) + \boldsymbol{D}^{i, j}_{n, m}(\boldsymbol{v}[i/v], \boldsymbol{w}) \right) 
			- \prod_{v \in V_{J'}}  \phi(\boldsymbol{v}[i/v], \boldsymbol{w}) \right).
		\end{align*}
		This expression depends uniformly on $\boldsymbol{v}$ and $\boldsymbol{w}$.
		In other words,
		\begin{align*}
			\Pi_{i, J} (\phi) &= \psi \ \Pi_{i, J'}( \phi) + (1- \psi) \left( \Pi_{i,J'}(\phi + \boldsymbol{D}^{i,j}_{n,m}) - \Pi_{i, J'} (\phi) \right).
%			&= \Pi_{i, J'}(\phi) + (1- \psi) \left( \Pi_{i,J'}(\phi + \boldsymbol{D}^{i,j}_{n,m}) \right).\label{eq:lincomb-product-restricted}
		\end{align*}
		Here, we regard $1$ as the element of $\mathfrak{T}^k_{n,m}(\ell-1, r)$ which is uniformly $1$, cf.\ \cref{ex:one}.
		The induction hypothesis applies to $\Pi_{i, J'}$.
		Hence, $\Pi_{i, J'}( \phi),\Pi_{i,J'}(\phi + \boldsymbol{D}^{i,j}_{n,m}) \in \mathfrak{T}^{k}_{n,m}(\ell-1,r)$.
		By \cref{lem:gluing}, $\mathfrak{T}^k_{n,m}(\ell-1, r)$ is closed under pointwise products.
		Hence, $\Pi_{i, J} (\phi) \in \mathfrak{T}^k_{n,m}(\ell-1, r)$.
	\end{proof}

 	\subsection{Homomorphism Polynomials from Symmetric Circuits}
 
 	In this section we prove the direction \ref{it:main2} $\Rightarrow$ \ref{it:main1} of \cref{thm:main1} using the results we have established so far. For the proof, we start with a rigid symmetric circuit $C$ with $\maxSup(C) \leq k$, and show via induction from the inputs to the output that at each gate $g$, the computed polynomial is in $\mathfrak{T}^{2k}_{n,m}(\vec{\sup}_L(g), \vec{\sup}_R(g))$.  
 	Recall that $\vec{\sup}_L(g)$ and $\vec{\sup}_R(g)$ denote ordered tuples whose entries are the elements of $\sup_L(g)$ and $\sup_R(g)$ (see \cref{sec:supports}). The ordering of the tuples represents the assignment of the labels to the support elements:
 	The label $i \in[\ell]$ is mapped to the $i$-th entry of $\vec{\sup}_L(g)$, and analogously for $\vec{\sup}_R(g)$. 
 	
 	For the inductive step of the proof, we assume that $g$ is a multiplication or summation gate such that every child $h$ of $g$ computes a polynomial in $\mathfrak{T}_{n,m}^{2k}(\vec{\sup}_L(h), \vec{\sup}_R(h))$. Then we show that the polynomial computed at $g$ is in $\mathfrak{T}_{n,m}^{2k}(\vec{\sup}_L(g), \vec{\sup}_R(g))$. In the following, if $g$ is a gate, then $p_g$ denotes the polynomial that this gate outputs.
  	
  	\begin{observation}\label{lem:labelsMovedByPermutations}
   		Let $n,m, \ell, r \in \mathbb{N}$.
  		Let $\boldsymbol{F} \in \mathcal{G}(\ell, r)$ and let $(\pi,\sigma) \in \Sym_n \times \Sym_m$. 
  		Then $(\pi, \sigma)(\boldsymbol{F}_{n,m}(\boldsymbol{v}, \boldsymbol{w})) = \boldsymbol{F}_{n,m}(\pi(\boldsymbol{v}), \sigma(\boldsymbol{w}))$ for all $\boldsymbol{v} \in [n]^\ell$ and $\boldsymbol{w} \in [m]^r$.
  	\end{observation}	
  	\begin{proof}
 		For $\boldsymbol{F} = (F, \boldsymbol{a}, \boldsymbol{b})$,
  		we have $(\pi,\sigma)( \{ h\colon A \uplus B \to [n] \uplus [m] \mid h(\boldsymbol{a}) = \boldsymbol{v}, h(\boldsymbol{b}) = \boldsymbol{w} \} ) =  \{ h\colon A \uplus B \to [n] \uplus [m] \mid h(\boldsymbol{a}) = \pi(\boldsymbol{v}) , h(\boldsymbol{b}) = \sigma(\boldsymbol{v}) \}$.
  	\end{proof}	
  	
 	\begin{lemma}
 		\label{lem:supportOfGateDeterminesLabels}
 		Let $C$ be a rigid $\Sym_n \times \Sym_m$-symmetric circuit and $g \in V(C)$ a gate such that the polynomial computed at $g$ is of the form $p_g = \phi(\vec{\sup}_L(g),\vec{\sup}_R(g))$ for some $\phi \in \mathfrak{T}^{2k}_{n,m}(\ell, r)$.
 		Let $g' \in \Orb_{\Sym_{n} \times \Sym_{m}}(g)$. Then the polynomial computed by $g'$ is $p_{g'} = \phi(\vec{\sup}_L(g'),\vec{\sup}_R(g'))$.
 	\end{lemma}	
 	\begin{proof}
 		Let $(\pi, \sigma) \in \Sym_{n} \times \Sym_{m}$ such that $(\pi, \sigma)(g)  = g'$. So the polynomial computed by $g'$ is $(\pi, \sigma)(p)$. 
 		Then the lemma follows from \cref{lem:labelsMovedByPermutations}.
 	\end{proof}

 	The next lemma shows the inductive step in case that the children of the gate $g$ form a single orbit.
 	\begin{lemma}
 		\label{lem:productOfOneChildOrbit}
 		Let $k \in \bbN$.
 		Let $C$ be a rigid $\Sym_{n} \times \Sym_{m}$-symmetric circuit, and $g \in V(C)$ a multiplication or a summation gate with $|\sup(g)| \leq k$.
 		Assume that for every $h \in \child(g)$, it also holds
                that $|\sup(h)| \leq k$ and the polynomial computed at $h$ is $p_h = \phi(\vec{\sup}_L(h), \vec{\sup}_R(h))$ for some $\phi \in \mathfrak{T}^{2k}_{n,m}(|\vec{\sup}_L(h)|, |\vec{\sup}_R(h)|)$.
 		Let $h \in \child(g)$ and let $O_h \coloneqq \Orb_{\StabP(\sup(g))}(h)$.\\
 		Then $\prod_{h' \in O_h} p_{h'}$ and $\sum_{h' \in O_h} p_{h'}$ are in $\mathfrak{T}_{n,m}^{2k}(\vec{\sup}_L(g), \vec{\sup}_R(g))$.
 	\end{lemma}	
 	\begin{proof}
 		As in \cref{lem:intersectionOfChildSupports}, let $S(h) = \bigcap_{h' \in O_h} \sup(h')$, and $S_L(h) = S(h) \cap [n]$, $S_R(h) = S(h) \cap [m]$.
  		Note that for every $h' \in O_h$, the tuples $\vec{\sup}_L(h')$ and $\vec{\sup}_R(h')$ have the elements of $S(h)$ in the same positions because $S(h) \subseteq \sup(g)$ by Lemma \ref{lem:intersectionOfChildSupports}. Since all these gates are in the same $\StabP(\sup(g))$-orbit, the positions of $S(h)$ in $\vec{\sup}_L(h')$ and $\vec{\sup}_R(h')$ are the same. In the following, when we write $\vec{S}_L(h)$ and $\vec{S}_R(h)$, we refer to the ordering of $S_L(h)$ and $S_R(h')$ as in $\vec{\sup}_L(h')$ and $\vec{\sup}_R(h')$, respectively. 
 		
		For simplicity, assume in the following that the first $|S_L(h)|$ positions in each tuple $\vec{\sup}_L(h')$ are occupied by the elements of $S_L(h)$ and the first $|S_R(h)|$ positions in $\vec{\sup}_R(h')$ are occupied by $S_R(h)$. 
 		Since $\Orb_{\StabP(\sup(g))}(h)$ is a subset of $\Orb_{\Sym_n \times \Sym_m}(h)$, by \cref{lem:supportOfGateDeterminesLabels}, all these gates compute the same instantiated homomorphism polynomial map, just with different instantiation of the labels.
 		Let $\phi \in \mathfrak{T}_{n,m}^{2k}(\ell , r)$ denote this homomorphism polynomial map, where $\ell = |\sup_L(h)|, r = |\sup_R(h)|$. 
 		Recall the definition of $\boldsymbol{J}$ from \cref{ex:one}.
 		First consider $\prod_{h' \in O_h} p_{h'}$. 
 		\begin{align*}
 		 \prod_{h' \in O_h} p_{h'} 
 		 &= \prod_{h' \in O_h}  \phi(\vec{\sup}_L(h'), \vec{\sup}_R(h')) \\
 		 &= \prod_{\boldsymbol{v} \in ([n] \setminus \sup_L(g))^{\ell - |S_L(h)|}} \prod_{\boldsymbol{w} \in ([m] \setminus \sup_R(g))^{r - |S_R(h)|}} \phi(\vec{S}_L(h) \boldsymbol{v}, \vec{S}_R(h) \boldsymbol{w})  \\
 		 &= \prod_{\boldsymbol{v} \in ([n] \setminus \sup_L(g))^{\ell - |S_L(h)|}} \prod_{\boldsymbol{w} \in ([m] \setminus \sup_R(g))^{r - |S_R(h)|}} \phi(\vec{S}_L(h) \boldsymbol{v}, \vec{S}_R(h) \boldsymbol{w})
 		  \cdot \boldsymbol{J}_{n,m}(\vec{\sup}_L(g) \setminus \vec{S}_L(h), \vec{\sup}_R(g) \setminus \vec{S}_R(h))\\
 		  &= \prod_{\boldsymbol{v} \in ([n] \setminus \sup_L(g))^{\ell - |S_L(h)|}} \prod_{\boldsymbol{w} \in ([m] \setminus \sup_R(g))^{r - |S_R(h)|}} (\phi \otimes \boldsymbol{J}_{n,m})(\vec{\sup}_L(g) \boldsymbol{v}, \vec{\sup}_R(g) \boldsymbol{w}). 
 		\end{align*}
 		The second equality is true because by Lemma \ref{lem:intersectionOfChildSupports}, for each $h' \in O_h$, 
 		the set $\sup(h') \setminus S(h)$ does not intersect $\sup(g)$, and $O_h$ is the orbit of $h$ under the pointwise stabiliser of $\sup(g)$ in $\Sym_n \times \Sym_m$. Hence, there is a bijection between the gates in $O_h$ and the set $\{  \boldsymbol{v}\boldsymbol{w} \mid  \boldsymbol{v} \in ([n] \setminus \sup_L(g))^{\ell - |S_L(h)|},  \boldsymbol{w} \in ([m] \setminus \sup_R(g))^{r - |S_R(h)|}  \}$.
 		The third equality holds because by \cref{ex:one}, $\boldsymbol{J}_{n,m}(\vec{\sup}_L(g) \setminus \vec{S}_L(h), \vec{\sup}_R(g) \setminus \vec{S}_R(h))= 1$. The last equality is due to the definition of $\otimes$ on homomorphism polynomial maps, cf.\ \cref{lem:disjoint-union}.
 		Analogously, 
 		it holds that
 		\[
 			\sum_{h' \in O_h} p_{h'}
 			=
 			\sum_{\boldsymbol{v} \in ([n] \setminus \sup_L(g))^{\ell - |S_L(h)|}} 
 			\sum_{\boldsymbol{w} \in ([m] \setminus \sup_R(g))^{r - |S_R(h)|}} (\phi \otimes \boldsymbol{J}_{n,m})(\vec{\sup}_L(g) \boldsymbol{v}, \vec{\sup}_R(g) \boldsymbol{w}).
 		\]
 		
 		By \cref{lem:disjoint-union}, $(\phi \otimes \boldsymbol{J}) \in \mathfrak{T}_{n,m}^{2k}(|\sup_L(g)|+|\sup_L(h) \setminus \sup_L(g)|, |\sup_R(g)|+|\sup_R(h) \setminus \sup_R(g)|)$. 
 		Here it is important that $| \sup(g)| + |\sup(h)| \leq 2k$.
 		By repeatedly applying \cref{cor:product-exclude-lincomb} to the last expression above,
 		it follows that $\prod_{h' \in O_h} p_{h'} \in \mathfrak{T}^{2k}_{n,m}(\vec{\sup}_L(g), \vec{\sup}_R(g))$.
 		If $g$ is a summation gate, then similarly, with \cref{lem:sum-exclude-lincomb}, $\sum_{h' \in O_h} p_{h'} \in \mathfrak{T}^{2k}_{n,m}(\vec{\sup}_L(g), \vec{\sup}_R(g))$.
 	\end{proof}
 	
 	We now generalise the above lemma so that it covers all orbits of children of~$g$.
 	
 	 \begin{lemma}
 		\label{lem:productOfAllChildOrbits}
 		Let $k \in \bbN$. 
 		Let $C$ be a rigid $\Sym_{n} \times \Sym_{m}$-symmetric circuit. 
 		Let $g \in V(C)$ be a multiplication or summation gate with $|\sup(g)| \leq k$.
 		Assume that for every $h \in \child(g)$, the polynomial $p_h$ is in $\mathfrak{T}_{n,m}^{2k}(\vec{\sup}_L(h), \vec{\sup}_R(h))$.
 		Then the polynomial computed at $g$ is in $\mathfrak{T}_{n,m}^{2k}(\vec{\sup}_L(g), \vec{\sup}_R(g))$.
 	\end{lemma}	
 	\begin{proof}
 		We treat the case when $g$ is a multiplication gate in detail.
 		The other case is analogous.
 		Partition the set $\child(g)$ into orbits under $\StabP(\sup(g))$. 
		To each such orbit $O_h$, we apply \cref{lem:productOfOneChildOrbit}, 
		which shows that $p_{O_h} \coloneqq \prod_{h' \in O_h} p_{h'} \in \mathfrak{T}^{2k}_{n,m}(\vec{\sup}_L(g), \vec{\sup}_R(g))$.
		Let $\Omega \coloneqq \{O_h \mid h \in \child(g)\}$ and let $\phi_{O_h}(\vec{\sup}_L(g), \vec{\sup}_R(g))$ be the instantiated homomorphism polynomial map in $\mathfrak{T}^{2k}_{n,m}(\vec{\sup}_L(g), \vec{\sup}_R(g))$ that is equal to $p_{O_h}$. For each orbit $O_h \in \Omega$, let $m(O_h) \in \bbN$ denote the multiplicity of the wire between $g$ and each gate $h' \in O_h$.
 		The polynomial computed at $g$ is 
 		%$\prod_{O_h \in \Omega} p_{O_h}$.
 		%We show that $\prod_{O_h \in \Omega} p_{O_h} \in \mathfrak{T}^{2k}_{n,m}(\vec{S}_L, \vec{S}_R)$, for $S \coloneqq \bigcup_{h \in \child(g)} S(h)$, where $S(h)$ is like in Lemma \ref{lem:productOfOneChildOrbit}. For each $O_h \in \Omega$, let $\phi_{O_h}(\vec{\sup}_L(g), \vec{\sup}_R(g)) \in \mathfrak{T}_{n,m}^{2k}(\vec{\sup}_L(g), \vec{\sup}_R(g))$ be the polynomial that we get from \cref{lem:productOfOneChildOrbit}.
 		%Let $T_{O_h}$ be an index set such that $\phi_{O_h}= \sum_{i \in T_{O_h}} \alpha_i\boldsymbol{F}^i_{n,m}$. Let $T = \biguplus_{O_h \in \Omega} T_{O_h}$.
 		\begin{align*}
 			p_g = \prod_{O_h \in \Omega} p_{O_h}^{m(O_h)} &= \prod_{O_h \in \Omega} \phi_{O_h}(\vec{\sup}_L(g), \vec{\sup}_R(g))^{m(O_h)}.
 			 %\sum_{\stackrel{f\colon \Omega \to T}{f(O_h) \in T_{O_h} \text{ for each } O_h \in \Omega}} \prod_{O_h \in \Omega} \boldsymbol{F}^{f(O_h)}_{n,m}(\vec{S}(h)_L, \vec{S}(h)_R)\\
 			%&= \sum_{\stackrel{f\colon \Omega \to T}{f(O_h) \in T_{O_h} \text{ for each } O_h \in \Omega}} \boldsymbol{G}^{f}_{n,m}(\vec{S}_L, \vec{S}_R).
 		\end{align*}	
 		%In the last line, we applied \cref{lem:semiGluing} to the product: $G^f$ denotes the graph that is obtained by gluing the graphs in $\{ F^{f(O_h)} \mid O_h \in \Omega \}$ together, where vertices of distinct graphs $F^{f(O_{h_1})}, F^{f(O_{h_2})}$ are identified whenever their labels are assigned to the same element of $[n] \uplus [m]$ in the tuples $\vec{S}(h_1)$ and $\vec{S}(h_2)$ (see \cref{lem:semiGluing}). 
 		This polynomial is in $\mathfrak{T}_{n,m}^{2k}(\vec{\sup}_L(g), \vec{\sup}_R(g))$ by \cref{lem:gluing}.
 	\end{proof}

 	\subsection{Proof of the backward direction of \cref{thm:main1}}
 	\begin{proof}
 		Let $(C_{n,m})_{n,m \in \bbN}$ be a family of circuits such that each $C_{n,m}$ is $\Sym_n \times \Sym_m$-symmetric, and $\maxOrb(C_{n,m})$ is polynomially bounded in $(n+m)$. By \cref{lem:rigidifyCircuits}, we can assume that the circuits are rigid. Then by \cref{lem:constantSupportOfGates}, there exists a constant $k \in \bbN$ such that $\maxSup(C_{n,m}) \leq k$ for all $n,m \in \bbN$. Consider now fixed $n, m \in \bbN$.
 		By induction on the structure of $C_{n,m}$, we show
                that for every gate $g \in V(C_{n,m})$, we have $p_g \in \mathfrak{T}^{2k}_{n,m}(\vec{\sup}_L(g), \vec{\sup}_R(g))$. If $g$ is an input gate labelled with a constant, then $g$ is fixed by every permutation, so $\sup(g) = \emptyset$. Indeed, any constant is a polynomial in $ \mathfrak{T}^{2k}_{n,m}(0, 0)$ by \cref{ex:one}. If $g$ is an input gate labelled with a variable $x_{ij}$, then $\sup_L(g) = \{i\}, \sup_R(g) = \{j\}$. It is clear that the polynomial $x_{ij}$ is in $\mathfrak{T}^{2k}_{n,m}(i, j)$ because $x_{ij} = \boldsymbol{F}_{n,m}(i,j)$ where $\boldsymbol{F}$ is the graph that consists of a single edge. This finishes the base cases. 
 		In the inductive step, $g$ is either a summation or multiplication gate. Both these cases are handled by \cref{lem:productOfAllChildOrbits}. Finally, the output gate of $C_{n,m}$ is stabilised by $\Sym_n \times \Sym_m$, so its minimal support is empty, which means that the polynomial computed by $C_{n,m}$ is in $\mathfrak{T}^{2k}_{n,m}$. 
 	\end{proof}

 	\subsection{Symmetric Circuits for Homomorphism Polynomials}
 	In this section, we prove \cref{thm:hom-circuit-main}. 	
 	Let $F$ be a graph with a tree decomposition $(T, \beta)$ and $r \in V(T)$.
 	For $s \in V(T)$, write $T^s$ for the subtree of $T$ induced by $s$ and all its descendents in the rooted tree $(T, r)$.
 	Write $F^s$ for the subgraph of $F$ induced by $\bigcup_{t \in V(T^s)} \beta(t)$.
 	
 	\begin{lemma} \label{lem:hom-circuit-core-bipartite}
 		Let $k, n, m \in \mathbb{N}$.
 		Let $F$ be a bipartite multigraph with bipartition $A \uplus B$, tree decomposition $(T, \beta)$, and $r \in V(T)$ such that 
 		\begin{enumerate}
 			\item $|\beta(t)| = k$ for all $t \in V(T)$,
 			\item $|\beta(s) \cap \beta(t)| = k-1$ for all $st \in E(T)$, 
% 			\item if $s, t \in V(T)$ are siblings in $(T,r)$, then $\beta(s) \neq \beta(t)$, and
% Removed as this follows from the previous two conditions
			\item every vertex in the rooted tree $(T,r)$ has out-degree at most $k$.
 		\end{enumerate}
 		There exists a rigid $\Sym_n \times \Sym_m$-symmetric circuit of size at most $4  k!^2 k \cdot (n+m)^{k+1} \cdot \abs{V(T)} \cdot \norm{F} $
 		and support size at most $k$
 		which contains, for every 
 		$s \in V(T)$,
 		$\boldsymbol{a} \in A^\ell$,
 		$\boldsymbol{b} \in B^r$,
 		$\ell +r = k$ such that
 		$\beta(s) = \{a_1, \dots, a_\ell, b_1, \dots, b_r\}$, and $\boldsymbol{v} \in [n]^\ell$,
 		$\boldsymbol{w} \in [m]^r$,
 		a gate which computes $\boldsymbol{F}^s_{n,m}(\boldsymbol{v}, \boldsymbol{w})$
 		where $\boldsymbol{F}^s \coloneqq (F^s, \boldsymbol{a}, \boldsymbol{b}) \in \mathcal{T}^{k}(\ell, r)$.
 	\end{lemma}
 	\begin{proof}
 		We build the circuit by induction on the depth of $s$ in $(T, r)$.
 		For the base case, consider all leaves of $T$.
 		For every leaf $s \in V(T)$,
 		all vertices in $ \boldsymbol{F}^s$ are labelled.
 		Hence,
 		$ \boldsymbol{F}^s_{n,m}(\boldsymbol{v},\boldsymbol{w})$ can be represented by a single multiplication gate whose fan-in equals the number of edges in $F^s$.
 		For every leaf, this yields a subcircuit of size at most $\norm{F} n^\ell m^r k! \leq \norm{F} (n+m)^k k!$.
 		
 		For the inductive step, suppose that a circuit has been constructed which contains the required gates for all $s \in V(T)$ of depth $> d$.
 		Let $s \in V(T)$ be a vertex at depth $d$.
 		Fix $\boldsymbol{v} \in [n]^\ell$,
 		$\boldsymbol{w} \in [m]^r$, $\boldsymbol{a} \in A^\ell$,
 		$\boldsymbol{b} \in B^r$,
 		$\ell +r = k$ such that
 		$\beta(s) = \{a_1, \dots, a_\ell, b_1, \dots, b_r\}$.
 		
 		Let $t \in V(T)$ be a child of $s$.
 		Write $u$ for the unique vertex in $\beta(s) \setminus \beta(t)$.
 		Suppose wlog that $u = a_\ell \in A$.
 		Write $K^t$ for the subgraph of $F$ induced by $\beta(s) \cup \bigcup_{t' \in V(T^t)} \beta(t')$.
 		Let $\boldsymbol{K}^t \coloneqq (K^t, \boldsymbol{a}, \boldsymbol{b})$.
 		We first construct a circuit for $\boldsymbol{K}^t_{n,m}(\boldsymbol{v}, \boldsymbol{w})$.
 		Let $u'$ denote the unique vertex in $\beta(t) \setminus \beta(s)$.
 		Distinguish cases:
 		\begin{enumerate}
 			\item If $u' \in A$, let $\boldsymbol{a}' \coloneqq \boldsymbol{a}[\ell/u']$ and define $\boldsymbol{F}^t \coloneqq (F^t, \boldsymbol{a}', \boldsymbol{b})$.
 				  Then
 				  \begin{equation}\label{eq:circuit1}
 				  \boldsymbol{K}^t_{n,m}(\boldsymbol{v}, \boldsymbol{w}) = \sum_{v \in [n]} \boldsymbol{F}^t_{n,m}(\boldsymbol{v}[\ell/v], \boldsymbol{w}) \prod_{j \in [r]} \prod_{\substack{a b \in E(F) \\ a = a_\ell \\ b = b_j}} x_{v_\ell w_j}.
 				  \end{equation}
 			\item If $u' \in B$, let $\boldsymbol{a}' \coloneqq \boldsymbol{a}[\ell/]$ and define $\boldsymbol{F}^t \coloneqq (F^t, \boldsymbol{a}', \boldsymbol{b}u')$.
 				  Then
 				  \begin{equation}\label{eq:circuit2}
 				  \boldsymbol{K}^t_{n,m}(\boldsymbol{v}, \boldsymbol{w}) = \sum_{w \in [m]} \boldsymbol{F}^t_{n,m}(\boldsymbol{v}[\ell/], \boldsymbol{w}w) \prod_{j \in [r]} \prod_{\substack{a b \in E(F) \\ a = a_\ell \\ b = b_j}} x_{v_\ell w_j}.
 				  \end{equation}
 		\end{enumerate}
 		In both cases, the inductive hypothesis applies to $\boldsymbol{F}^t$.
 		In order to construct $\boldsymbol{K}^t_{n,m}(\boldsymbol{v}, \boldsymbol{w})$, we require one summation gate of fan-in $\max\{n,m\}$
 		and a product gate of fan-in at most $\norm{F} + 1$.
 		
 		Let now $t_1, \dots, t_\ell$ for $\ell \leq k$ denote the children of $s$.
 		By \cref{lem:gluing},
 		\begin{equation}\label{eq:circuit3}
 			\boldsymbol{F}^s_{n,m}(\boldsymbol{v},\boldsymbol{w}) = \prod_{i=1}^\ell \boldsymbol{K}^{t_i}_{n,m}(\boldsymbol{v}, \boldsymbol{w}).
 		\end{equation}
 		This polynomial can be realised by one additional product gate of fan-in at most $k$.
 		
 		Thus, in order to construct the additional gates for $s \in V(T)$, 
 		we require at most \[ n^\ell m^r \ell! r! (4 + k  + \max\{n,m\} + \norm{F})  \leq 4(n+m)^{k+1} k!^2 k \cdot \norm{F} \] gates and edges.
 		It follows that the overall circuit has the desired size. In the end, we invoke \cref{lem:rigidifyCircuits} to make the circuit rigid, which allows us to speak about the supports of gates.
 		
 		Let $\boldsymbol{v} \in [n]^\ell$ and $\boldsymbol{w} \in [m]^r$.
 		At any stage, the subcircuit computing $\boldsymbol{F}^s_{n,m}(\boldsymbol{v}, \boldsymbol{w})$ is by construction $\StabP_{\Sym_n \times \Sym_m}(\boldsymbol{v}, \boldsymbol{w})$-symmetric, so $\boldsymbol{v}\boldsymbol{w}$ is a support of the gate computing $\boldsymbol{F}^s_{n,m}(\boldsymbol{v}, \boldsymbol{w})$. Note that $|\boldsymbol{v}\boldsymbol{w}| \leq k$.
 		
 		The overall circuit is $\Sym_n \times \Sym_m$-symmetric,
 		since, by \cref{lem:labelsMovedByPermutations}, $(\pi, \sigma)(\boldsymbol{F}^s_{n,m}(\boldsymbol{v}, \boldsymbol{w})) = \boldsymbol{F}^s_{n,m}(\pi(\boldsymbol{v}), \sigma(\boldsymbol{w}))$ is computed as well for $(\pi, \sigma) \in \Sym_n \times \Sym_m$.
 	\end{proof}
 	
 	The lemma implies the theorem:
 	
 	\thmHomCircuit*
 	\begin{proof}
 		By \cite[Lemma~6]{seppelt_algorithmic_2024} and its proof, 
 		there exists a tree decomposition $(T, \beta)$ of $F$ with a vertex $r \in V(T)$ possessing the properties required by \cref{lem:hom-circuit-core-bipartite}.
 		Note that $|V(T)| \leq |V(F)|$.
 		
 		Write $A \uplus B$ for the bipartition of $F$.
 		Let $\ell, r \in \mathbb{N}$ be such that there exists $\boldsymbol{a} \in A^\ell$ and $\boldsymbol{b} \in B^r$ such that $\beta(r) = \{a_1, \dots, a_\ell, b_1, \dots, b_r\}$.
 		Let $\boldsymbol{F} \coloneqq (F, \boldsymbol{a}, \boldsymbol{b}) \in \mathcal{T}^k(\ell, r)$.
 		By \cref{lem:hom-circuit-core-bipartite}, 
 		for $n,m \in \mathbb{N}$,
 		there exists a $\Sym_n \times \Sym_m$-symmetric circuit computing $\boldsymbol{F}_{n,m}(\boldsymbol{v}, \boldsymbol{w})$
 		for every $\boldsymbol{v} \in [n]^\ell$ and $\boldsymbol{w} \in [m]^r$.
 		The circuit is of size at most $4 k!^2 k \cdot (n+m)^{k+1} \cdot \norm{F}^2$ and support size at most $k$.
 		By \cref{lem:unlabelling}, the desired circuit can be constructed from this circuit by adding one additional summation gate of fan-in $(n+m)^k$.
 		That is,
 		\[
 			\hom_{F, n,m} = \sum_{\boldsymbol{v} \in [n]^\ell} \sum_{\boldsymbol{w} \in [m]^r} \boldsymbol{F}_{n,m}(\boldsymbol{v}, \boldsymbol{w}). \qedhere
 		\]
 	\end{proof}

	\subsection{Polynomial Orbit Size Versus Polynomial Size}
	
	\cref{thm:main1} gives a characterisation for when symmetric circuits with polynomial \emph{orbit size}, rather than polynomial total size, exist. A super-polynomial lower bound on the orbit size of course implies a super-polynomial lower bound on the circuit size itself, but for upper bounds, this is not true: There can exist symmetric circuits with polynomial orbit size but whose total size is super-polynomial. The following theorem shows that under the common assumption $\VP \neq \VNP$, this situation indeed arises: There exist polynomials in $\mathfrak{T}_{n,m}^k$ that do not admit polynomial size circuits (neither symmetric nor asymmetric), unless $\VP = \VNP$. Thus, \cref{thm:main1} is best-possible, and obtaining a characterisation of \emph{total} circuit size rather than orbit size can be expected to be as difficult for symmetric circuits as for general ones.
	
	\hardPolynomialsInT*

	\begin{proof}
		We consider a family of polynomials that was shown to be $\VNP$-hard in \cite[Theorem~1]{curticapean_et_al2022}, and modify it. 
		For a partition $\lambda \vdash d$, write $R$ for the sets of parts of $\lambda$,
		and, for $i \in [d]$, write $s_i$ for the number of size-$i$ parts in $\lambda$.
		Let $n \geq d$ be an integer. 
		Consider the \emph{monomial symmetric polynomial}:
		\[
		m_\lambda(y_1, \dots, y_n) \coloneqq \alpha_\lambda \sum_{f \colon R \hookrightarrow [n]} \prod_{r \in R} y_{f(r)}^{|r|}.
		\]	
		Here, $\alpha_\lambda \coloneqq \prod_{i=1}^d \frac{1}{s_i!}$ is such that all coefficients in $m_\lambda$ are $1$.
		
		In \cite[Theorem~1]{curticapean_et_al2022}, the existence of polynomial functions $r,s \colon \mathbb{N} \to \mathbb{N}$ and of partitions $\lambda_n \vdash r(n)$, $n \in \mathbb{N}$,
		is established such that the family of polynomials $m_{\lambda_n}(y_1, \dots, y_{s(n)})$ is $\VNP$-complete.
		
		Given $m_\lambda(y_1, \dots, y_n)$, we define a polynomial $p_\lambda \in \mathbb{Q}[\mathcal{X}_{n,n}]$ by substituting, for every $v \in [n]$, the expression $\sum_{w \in [n]} x_{vw}$ for $y_v$.
		Clearly, $m_\lambda$ can be recovered from $p_\lambda$ by substituting $\frac1n y_v$ for $x_{vw}$ for all $v, w \in [n]$.
		Hence, the $\VNP$-complete family in \cite[Theorem~1]{curticapean_et_al2022} gives rise to a $\VNP$-complete family of polynomials of the form $p_\lambda \in \mathbb{Q}[\mathcal{X}_{n,n}]$.
		It remains to prove membership in $\mathfrak{T}^2_{n,n}$.
		
		To that end, let $\lambda \vdash d$ and $n \geq d$ be arbitrary.
		Recall that $R$ denotes the set of parts of $\lambda$
		and apply Möbius inversion to the outer sum.
		By \cref{lem:moebius-polynomial},
		\begin{align*}
			m_\lambda(y_1, \dots, y_n) &= \alpha_\lambda \sum_{\pi \in \Pi(R)} \mu_\pi \sum_{f \colon R/\pi \to [n]} \prod_{r \in R} y_{f\pi(r)}^{|r|} \\
			&= \alpha_\lambda \sum_{\pi \in \Pi(R)} \mu_\pi \sum_{f \colon R/\pi \to [n]} \prod_{s \in R/\pi} y_{f(s)}^{\sum_{r \in s}|r|} \\
			&= \alpha_\lambda  \sum_{\pi \in \Pi(R)} \mu_\pi \prod_{s \in R/\pi} \sum_{v \in [n]} y_v^{\sum_{r \in s}|r|}.
		\end{align*}
		For $\phi(v) \coloneqq \sum_{w \in [n]} x_{vw}$ where $v \in [n]$,
		we have $\phi \in \mathfrak{T}^2_{n,n}(1,0)$ by \cref{lem:unlabelling,ex:edge}.
		Similarly, $\sum_{v \in [n]} \phi(v)^{\ell} \in \mathfrak{T}^2_{n,n}$ for every  integer $\ell \geq 1$ by  \cref{lem:gluing,lem:unlabelling}.
		Hence, the desired polynomial $p_{\lambda} = m_{\lambda}(\phi(1), \dots, \phi(n))$ is a linear combination of polynomials in $\mathfrak{T}^2_{n,m}$ and thus in $\mathfrak{T}^2_{n,m}$ itself.
	\end{proof}	
	
	\section{Counting Width and Orbit Size Lower Bounds}
	
	\Cref{thm:main1}, the main result of the preceding section, characterises the polynomials which are computed by symmetric circuits of polynomial orbit size. Ultimately, the goal of the symmetric circuit programme is to show super-polynomial lower bounds on the orbit and hence circuit size, for polynomials for which this is not currently possible in the non-symmetric circuit setting. However, \cref{thm:main1} itself does not give us an immediate tool to prove such lower bounds in general. 
	
	In this section and in the following ones, we consider more restricted families of polynomials with the goal of determining their orbit size.
	To the best of our knowledge, the only available tool for proving orbit size lower bounds is the \emph{counting width} as introduced by \textcite{dawar_definability_2017}.
	It describes the descriptive complexity of a graph parameter.
	
	In order to define counting width, we view $(n,m)$-vertex bipartite graphs as elements of $\{0,1\}^{n \times m}$.
	More generally, an \emph{edge-weighted} $(n,m)$-vertex bipartite graph is an element of $\mathbb{Q}^{n \times m}$. 
	
	Let $\tau \coloneqq \{A, B\} \cup \{R_q \mid q \in \mathbb{Q}\}$ denote the infinite relational signature with two unary symbols $A$ and $B$ and one binary symbol $R_q$ for every $q \in \mathbb{Q}$.
	We view elements $G \in \mathbb{Q}^{n \times m}$ as $\tau$-structures with universe $n \times m$ by interpreting $A^G \coloneqq [n]$, $B^G \coloneqq [m]$, and $R_q^G = \{ij \in n \times m \mid G(ij) = q\}$ for $q \in \mathbb{Q}$.
	% By default, whenever we speak of bipartite graphs, we mean that they are equipped with two unary relations marking the bipartition.
	An \emph{isomorphism} between $G, H \in \mathbb{Q}^{n \times m}$ is a bijection $h \colon [n] \uplus [m] \to [n] \uplus [m]$ such that $G \circ h = H$.
	
	\begin{definition}[\cite{dawar_definability_2017}]
		\label{def:counting-width}
		Let $n,m\in \mathbb{N}$.
		Let $p$ be an $\Sym_n \times \Sym_m$-invariant function on domain $\mathbb{Q}^{n \times m}$. 
		\begin{enumerate}
			\item The \emph{counting width of $p$ on $(n,m)$-vertex simple graphs} is the least integer $k$ such that, for all $G, H \in \{0,1\}^{n \times m}$, if $G$ and $H$ are $\mathcal{C}^k$-equivalent, then $p(G) = p(H)$.
			\item The \emph{counting width of $p$ on $(n,m)$-vertex edge-weighted graphs} is the least integer $k$ such that, for all $G, H \in \mathbb{Q}^{n \times m}$, if $G$ and $H$ are $\mathcal{C}^k$-equivalent, then $p(G) = p(H)$.
		\end{enumerate}
	\end{definition} 
	%\begin{definition}\label{def:counting-width}
	%	Let $n,m\in \mathbb{N}$.
	%	Let $p$ be an isomorphism-invariant function with domain $\mathbb{Q}^{n \times m}$. 
	%	\begin{enumerate}
		%	\item The \emph{counting width of $p$ on $(n,m)$-vertex simple graphs} is the least integer $k$ such that, for all $G, H \in \{0,1\}^{n \times m}$, if $G$ and $H$ are $\mathcal{C}^k$-equivalent, then $p(G) = p(H)$.
			%\item The \emph{counting width of $p$ on $(n,m)$-vertex edge-weighted graphs} is the least integer $k$ such that, for all $G, H \in \mathbb{Q}^{n \times m}$, if $G$ and $H$ are $\mathcal{C}^k$-equivalent, then $p(G) = p(H)$.
	%	\end{enumerate}
	%\end{definition} 
	
	Since $p$ is assumed to be isomorphism-invariant and every edge-weighted graph can be defined in $\mathcal{C}^{n+m}$ up to isomorphism, the counting width on $(n,m)$-vertex graphs is at most $n+m$ and well-defined.
	Clearly, the counting width on edge-weighted graphs is an upper bound on the counting width on simple graphs.
	
	The following theorem establishes the desired connection between counting width and orbit size.	
	It is essentially due to \textcite{anderson_symmetric_2017}.
	
	\begin{theorem}\label{thm:counting-width}
		Let $(p_{n,m})_{n,m\in \mathbb{N}}$ be a family of polynomials.
		\begin{enumerate}
			\item If the $p_{n,m}$ admit $\Sym_n \times \Sym_m$-symmetric circuits of orbit size polynomial in $n+m$,
			then there is a constant $k \in \bbN$ that bounds the counting width of $p_{n,m}$ on $(n,m)$-vertex edge-weighted graphs, for all $n,m \in \bbN$.
			\item If the $p_{n,m}$ admit $\Sym_n \times \Sym_m$-symmetric circuits of orbit size $2^{o(n)}$, then the counting width of $p_{n,m}$ on $(n,m)$-vertex edge-weighted graphs is at most $o(n)$.
		\end{enumerate}
	\end{theorem}
	\begin{proof} 
		We prove the contrapositive of the first statement. Suppose the counting width of $p_{n,m}$ on $(n,m)$-vertex edge-weighted graphs is not bounded by a constant. Then for every $k \in \bbN$, there exist $n,m \in \bbN$ such that we find $\mathcal{C}^{k}$-equivalent $(n,m)$-vertex edge-weighted graphs $G,H$ with $p_{n,m}(G) \neq p_{n,m}(H)$.
		The graphs being $\mathcal{C}^{k}$-equivalent means that Duplicator has a winning strategy in the bijective $k$-pebble played on $(G,H)$. Then by \cite[Theorem~4.8]{dawar2021lower}, any $\Sym_n \times \Sym_m$-symmetric circuit whose support sizes are bounded by $k/2$ yields the same output on the adjacency matrices of $G$ and $H$ (the game in that theorem is not the standard bijective $k$-pebble game, but an equivalent one). Since $p_{n,m}(G) \neq p_{n,m}(H)$, no such circuit can represent $p_{n,m}$. Since this argument can be made for every $k \in \bbN$, this shows that there do not exist $\Sym_n \times \Sym_m$-symmetric circuits with constant support size for all $p_{n,m}$.
		Then by the first part of \cref{lem:constantSupportOfGates} (and the fact that circuits can always be assumed to be rigid), there do not exist symmetric circuits with polynomial orbit size.
		
		The second statement can be shown similarly: Suppose the counting width of $p_{n,m}$ on $(n,m)$-vertex edge-weighted graphs is $\Omega(n)$. Then for every function $f(n) \in o(n)$, there exist $n,m \in \bbN$ such that we find $\mathcal{C}^{f(n)}$-equivalent $(n,m)$-vertex edge-weighted graphs $G,H$ with $p_{n,m}(G) \neq p_{n,m}(H)$. With the same reasoning as before, using the second part of \cref{lem:constantSupportOfGates} now, we conclude that there do not exist $\Sym_n \times \Sym_m$-symmetric circuits for $p_{n,m}$ with orbit size $2^{o(n)}$.
	\end{proof}

	By \cref{thm:counting-width}, if a family of polynomials $p_{n,m} \in \mathbb{Q}[\mathcal{X}_{n,m}]$ admits $\Sym_n\times \Sym_m$-symmetric circuits of orbit size polynomial in $n+m$,
	then the counting width of $p_{n,m}$ is bounded on edge-weighted bipartite graphs.
	In fact, the counting width corresponds precisely to the treewidth of the graphs in the homomorphism representation:
	%TODO: read this
	\begin{theorem}\label{thm:counting-width-explicit}
		Let $k,n,m \in \mathbb{N}$. If $p \in \mathfrak{T}^k_{n,m}$,
		then the counting width of $p$ on edge-weighted bipartite $(n,m)$-vertex graphs is at most $k$.
	\end{theorem}
	\begin{proof}
		Since $p$ is a linear combination of homomorphism polynomials,
		it suffices to show that, if $G, H \in \mathbb{Q}^{n\times m}$ are $\mathcal{C}^k$-equivalent, then $\hom_{F, n,m}(G) = \hom_{F, n,m}(H)$ for every bipartite multigraph $F$ with $\tw(F) < k$.
		To that end, consider the following claim.

		For relational structures $G$ and $H$ over the same universe $U$,
		we say that $(G,\boldsymbol{v})$  and $(H,\boldsymbol{w})$ are \emph{$\mathcal{C}^k$-equivalent} for $\boldsymbol{v},\boldsymbol{w} \in U^\ell$ and $\ell \leq k$
		if, for every $\mathcal{C}^k$-formula $\phi$ with $\ell$ free variables, $G\models \phi(\boldsymbol{v}) \iff H \models \phi(\boldsymbol{w})$.
		We prove the following:
		
		\begin{claim}
			Let $\ell, r \in \mathbb{N}$ such that $\ell +r \leq k$ and $\boldsymbol{F} \in \mathcal{T}^k(\ell, r)$.
			Let $\boldsymbol{v},\boldsymbol{v}' \in [n]^\ell$ and $\boldsymbol{w},\boldsymbol{w}' \in [m]^r$. 
			If $G, G' \in \mathbb{Q}^{n\times m}$ are such that $(G, \boldsymbol{v}, \boldsymbol{w})$ and $(G', \boldsymbol{v}', \boldsymbol{w}')$ are $\mathcal{C}^k$-equivalent,
			then $\boldsymbol{F}_{n,m}(\boldsymbol{v}, \boldsymbol{w})(G) = \boldsymbol{F}_{n,m}(\boldsymbol{v}', \boldsymbol{w}')(G')$.
		\end{claim}
		\begin{claimproof}
			Write $\boldsymbol{F} = (F, \boldsymbol{a}, \boldsymbol{b})$ and $V(F) = A \uplus B$ for its bipartition.
			By \cref{def:treewidth-labelled} and \cite[Lemma~7]{bodlaender_partial_1998},
			there exists a tree decomposition $(T, \beta)$ of $F$ such that
			\begin{enumerate}
				\item there exists a vertex $q \in V(T)$ such that $\beta(q)= \{a_1, \dots, a_\ell, b_1, \dots, b_r \}$,
				\item for every $s \in V(T)$ which has
                                  two children $t_1, t_2 \in V(T)$ in
                                  the rooted tree $(T, q)$, we have $\beta(t_1) = \beta(s) = \beta(t_2)$,
				\item for every $s \in V(T)$ which has
                                  a single child $t \in V(T)$ in the
                                  rooted tree $(T, q)$, we have $\beta(t) \subseteq \beta(s)$ and $|\beta(s)| = |\beta(t)|+1$ or  $\beta(s) \subseteq \beta(t)$ and $|\beta(t)| = |\beta(s)|+1$,
				\item for every $s \in V(T)$, $|\beta(s)| \leq k$.
			\end{enumerate}
			We prove the claim by induction on the height of the rooted tree $(T, q)$.
			
			For the base case, suppose that $q$ is the only vertex in $T$.
			Since all vertices of $F$ are labelled, 
			$\boldsymbol{F}_{n,m}(\boldsymbol{v}, \boldsymbol{w})(G)$ depends solely on the isomorphism type of the edge-weighted bipartite subgraph of $G$ induced by $\boldsymbol{v}$ and $\boldsymbol{w}$.
			Since $\ell + r \leq k$, this isomorphism type can be defined in $\mathcal{C}^k$.
			This implies the desired statement.
			
			For the inductive step, consider three cases:
			\begin{enumerate}
				\item The vertex $q$ has a single child $t \in V(T)$, $\beta(t) \subseteq \beta(q)$, and $|\beta(q)| = |\beta(t)|+1$.
				
				Define $F'$ as the subgraph of $F$ induced by $\bigcup_{s \in V(T) \setminus \{q\}} \beta(s)$.
				Suppose without loss of generality that $\beta(q) \cap B = \beta(t) \cap B$.
				Let $i \in [\ell]$ be such that $a_i \not\in \beta(t)$.
				Then $\boldsymbol{F}' \coloneqq (F', \boldsymbol{a}[i/], \boldsymbol{b})$ is a graph to which the inductive hypothesis applies.
				
				If $(G, \boldsymbol{v}, \boldsymbol{w})$ and $(G', \boldsymbol{v}', \boldsymbol{w}')$ are $\mathcal{C}^k$-equivalent,
				then so are $(G, \boldsymbol{v}[i/], \boldsymbol{w})$ and $(G', \boldsymbol{v}'[i/], \boldsymbol{w}')$.
				By the inductive hypothesis,
				$\boldsymbol{F}'_{n,m}(\boldsymbol{v}[i/], \boldsymbol{w})(G) = \boldsymbol{F}'_{n,m}(\boldsymbol{v}'[i/], \boldsymbol{w}')(G')$.

				Now $\boldsymbol{F}$ differs from $\boldsymbol{F}'$ only in the vertex $a_i$ and the potential edges connecting $a_i$ to the vertices in $\boldsymbol{b}$.
				More precisely,
				\[
					\boldsymbol{F}_{n,m}(\boldsymbol{v}, \boldsymbol{w})(G)
					= \boldsymbol{F}'_{n,m}(\boldsymbol{v}[i/], \boldsymbol{w})(G) \cdot \prod_{a_ib_j \in E(F)} G(v_i w_j).
				\]
				Since $(G, \boldsymbol{v}, \boldsymbol{w})$ and $(G', \boldsymbol{v}', \boldsymbol{w}')$ are $\mathcal{C}^k$-equivalent, $G(v_i w_j) = G'(v'_i w'_j)$ for all $i \in [\ell]$ and $j \in [r]$.
				This implies the claim.
				
				\item The vertex $q$ has a single child $t \in V(T)$, $\beta(t) \supseteq \beta(q)$, and $|\beta(q)|+1 = |\beta(t)|$.
				
				Without loss of generality $\beta(t) \cap B = \beta(q) \cap B$. Write $a$ for the unique element of $\beta(t) \setminus \beta(q)$. 
				The inductive hypothesis applies to $\boldsymbol{F}' \coloneqq (F, \boldsymbol{a}a, \boldsymbol{b})$.
				Here, we insert the vertex $a$ at the end of the tuple only to ease notation.
				If  $(G, \boldsymbol{v}, \boldsymbol{w})$ and $(G', \boldsymbol{v}', \boldsymbol{w}')$ are $\mathcal{C}^k$-equivalent and $\boldsymbol{v}$ and $\boldsymbol{w}$ are tuples of total length at most $k-1$,
				then there exists a bijection $\pi \colon [n] \to [n]$ such that 
				$(G, \boldsymbol{v}v, \boldsymbol{w})$ and $(G', \boldsymbol{v}'\pi(v), \boldsymbol{w}')$ are $\mathcal{C}^k$-equivalent for all $v \in [n]$.
				By the inductive hypothesis,
				$\boldsymbol{F}'_{n,m}(\boldsymbol{v}v, \boldsymbol{w})(G) = \boldsymbol{F}'_{n,m}(\boldsymbol{v}'\pi(v), \boldsymbol{w}')(G')$ for all $v \in [n]$.
				Hence, by \cref{lem:unlabelling},
				\[
					\boldsymbol{F}_{n,m}(\boldsymbol{v}, \boldsymbol{w})(G)	
					= \sum_{v \in [n]} \boldsymbol{F}'_{n,m}(\boldsymbol{v}v, \boldsymbol{w})(G)	
					= \sum_{v \in [n]} \boldsymbol{F}'_{n,m}(\boldsymbol{v}'v', \boldsymbol{w}')(G')
					= 	\boldsymbol{F}_{n,m}(\boldsymbol{v}', \boldsymbol{w}')(G').	
				\]
				
				\item The vertex $q$ has several children $t_1, \dots, t_\ell$.
				
				For $i \in [\ell]$,
				define $F^i$ as the subgraph of $F$ induced by $\beta(t_i)$, and $\beta(s)$ for all descendents $s$ of $t$.
				Then the inductive hypothesis applies to $\boldsymbol{F}^i \coloneqq (F^i, \boldsymbol{a}, \boldsymbol{b})$.
				By \cref{lem:gluing}, $\boldsymbol{F}_{n,m}(\boldsymbol{v}, \boldsymbol{w}) = \prod_{i=1}^\ell \boldsymbol{F}^i_{n,m}(\boldsymbol{v}, \boldsymbol{w})$.
				The claim follows.\qedhere
			\end{enumerate}
		\end{claimproof}
		The claim implies the theorem with $\ell = 0 = r$.	
	\end{proof}
	
	\cref{thm:counting-width-explicit,thm:main1} yield an alternative proof of the first claim of \cref{thm:counting-width}.
We are going to establish the converse of \cref{thm:counting-width-explicit} in various restricted cases, as summarised in \cref{fig:overview}. To start with, we note that if we only care about the semantics of the polynomials as functions on simple graphs, then counting width and treewidth indeed coincide.

	\subsection{Characterisation of polynomials up to multilinearisation}
	The \emph{multilinearisation} of a polynomial $p$ is obtained from $p$ by setting all non-zero exponents to~$1$. 
	If we consider polynomials only up to multilinearisation, i.e.\ up to their evaluation under $\{0,1\}$-assignments, then
	the following converse of \cref{thm:counting-width-explicit} holds.	
	\multilinear*

	\begin{proof}
		For the forward direction, observe that, by \cref{thm:counting-width-explicit},
		the desired statement holds for $q$.
		The polynomials $p$ and $q$ assume the same values on the Boolean cube $\{0,1\}^{n \times m}$.
		Hence, the counting width of $p$ on $(n,m)$-vertex simple bipartite graphs is at most $k$.
		
		For the converse direction,
		we write the function $p \colon \{0,1\}^{n\times m} \to \mathbb{Q}$ as a linear combination of $\mathcal{C}^k$-type indicator functions.
		This argument requires interpolation and thus relies on the fact that there are only finitely many simple $(n,m)$-vertex bipartite graphs.
		
		Let $C_1 \uplus \dots \uplus C_r = \{0,1\}^{n\times m}$ be an enumeration of the $\mathcal{C}^k$-types of simple $(n,m)$-vertex bipartite graphs.
		Pick representatives $G_1, \dots, G_r$.
		By \cite{dvorak_recognizing_2010},
		for $i \neq j$, there  exists a bipartite graph $F_{ij}$ of treewidth at most $k-1$ such that $\hom(F_{ij}, G_i) \neq \hom(F_{ij}, G_j)$.
		For $i \in [r]$, define the polynomial
		\[
			q_i\coloneqq \prod_{j \neq i} \frac{\hom_{F_{ij}, n,m} - \hom(F_{ij}, G_j)}{\hom(F_{ij}, G_i) - \hom(F_{ij}, G_j)} \in \mathfrak{T}^k_{n,m}.
		\]
		Here, $\hom_{F_{ij}, n,m}$ is a polynomial in $\mathfrak{T}^k_{n,m}$ while $\hom(F_{ij}, G_j)$ is merely an integer.
		For all $G \in \{0,1\}^{n\times m}$, $q_i(G) = 1$ if, and only if, $G \in C_i$, and $q_i(G) =0$ otherwise.
		Since $p$ is constant on $\mathcal{C}^k$-types,
		it is
		\[
			p = \sum_{i=1}^r \alpha_i q_i
		\]
		as functions on $\{0,1\}^{n\times m}$ for some coefficients $\alpha_i \in \mathbb{Q}$.
		Hence, $p$ is equal to the polynomial $\sum_{i=1}^r \alpha_i q_i \in \mathfrak{T}^k_{n,m}$ up to multilinearisation.
	\end{proof}
	We see in \cref{ex:matching} that the statements in \cref{thm:multilinear} do not imply that the counting width of $p$ on $(n,m)$-vertex \emph{edge-weighted} graphs is bounded.
	The argument yielding \cref{thm:multilinear} is bound to fail when multilinearisation is not applied:
	\begin{example}\label{ex:simple-interpolation}
		Let $F$ be a graph and $n,m\in \mathbb{N}$.
		Write $C = \{\hom(F, G) \mid G \in \{0,1\}^{n \times m}\}$ for the finite set of homomorphism counts from $F$ to any simple $(n,m)$-vertex graph.
		Consider the polynomial
		\begin{align*}
		p \coloneqq \hom_{F, n,m} \prod_{c \in C} (\hom_{F, n,m} - c)
		&= \hom_{F, n,m} \sum_{C' \subseteq C} \hom_{|C \setminus C'| F, n,m} \prod_{c \in C'} (-1)^{|C'|} \cdot c\\
		&=   \sum_{C' \subseteq C} \hom_{(|C \setminus C'|+1) F, n,m} \prod_{c \in C'} (-1)^{|C'|} \cdot c.
		\end{align*}
		Here, $kF$ for a graph $F$ and an integer $k \geq 1$ denotes the disjoint union of $k$ copies of $F$.
		Then $p$ is a linear combination of homomorphism polynomials and $p(G) = 0$ for all $G \in \{0,1\}^{n \times m}$.
		Hence, the counting width of $p$ on simple $(n,m)$-vertex graphs is zero.
		However, $p \not\in \mathfrak{T}^{k}_{n,m}$ for $k < \tw(F) - 1$ when $n,m \geq |E(F)|$.
	\end{example}

	\section{Symmetric complexity of homomorphism polynomials}
	
	Towards a converse of \cref{thm:counting-width-explicit},
	we characterise the homomorphism polynomials and the linear combinations of homomorphism polynomials of sublinear patterns that have bounded counting width.

\subsection{Single homomorphism polynomials}

	We first consider families of single homomorphism polynomials and show the following sufficient and necessary criterion for tractability.

\begin{theorem}
	\label{thm:dichotomy-chromatic}
	For every family  $(F_{n,m})_{n,m \in \mathbb{N}}$ of bipartite multigraphs,
	the following are equivalent:
	\begin{enumerate}
		\item the treewidth $\tw(F_{n,m})$ is bounded,\label{it:hom1}
		\item the counting width of $\hom_{F_{n,m},n,m}$ on $(n,m)$-vertex simple graphs is bounded,\label{it:hom2}
		\item the counting width of $\hom_{F_{n,m},n,m}$ on $(n,m)$-vertex edge-weighted graphs is bounded,\label{it:hom3}
		\item the $\hom_{F_{n,m}, n, m}$ admit $\Sym_n \times \Sym_m$-symmetric circuits of orbit size polynomial in $n+m$,\label{it:hom5}
		\item the $\hom_{F_{n,m}, n, m}$ admit $\Sym_n \times \Sym_m$-symmetric circuits of size polynomial in $\norm{F_{n,m}} + n +m$ and orbit size polynomial in $n+m$.\label{it:hom4}
	\end{enumerate}
\end{theorem}

%	Remarkably, \cref{thm:dichotomy-chromatic} shows that for single homomorphism polynomials there is no difference between counting width on simple graphs and counting width on edge-weighted graphs. \todo{isn't that also the case with the sublinear size linear combinations? is it even remarkable then?}
	Given the results of the preceding section,  it remains to show \enquote{\ref{it:hom2} $\Rightarrow$ \ref{it:hom1}}.
	To that end, 
	we prove a counting width lower bound for homomorphism counts of unbounded treewidth.

	\subsubsection{The Non-Uniform Homomorphism Distinguishing Closure}
	\label{sec:non-uniform-closure}
	
	Since it is conceivable that a more general statement will contribute to future work,
	we formulate our counting width lower bound in the language of homomorphism indistinguishability.
	Two graphs\footnote{In \cref{sec:non-uniform-closure,sec:bipartite-respecting}, all graphs are simple. 
		Considering homomorphism counts from simple graphs suffices in the context of 
		 \cref{thm:dichotomy-chromatic,thm:lincomb} since we prove lower bounds on the counting width on simple graphs.} $G$ and $H$ are \emph{homomorphism indistinguishable} over a graph class $\mathcal{F}$, in symbols $G \equiv_{\mathcal{F}} H$,
	if $\hom(F, G) = \hom(F, H)$ for all $F \in \mathcal{F}$.
	By \cite{dvorak_recognizing_2010,dell-grohe-rattan},
	two graphs are $\mathcal{C}^k$-equivalent
	if, and only if,
	they are homomorphism indistinguishable over all graphs of treewidth less than $k$.
	Similar characterisations are known for many graph isomorphism relaxations from finite model theory \cite{grohe_counting_2020,fluck_going_2024}, quantum information theory \cite{mancinska_quantum_2020}, and category theory \cite{dawar_lovasz-type_2021,montacute_pebble-relation_2024} etc.,
	cf.\ the monograph \cite{seppelt_homomorphism_2024}.
	The distinguishing power of such relations is well-understood via the so called \emph{homomorphism distinguishing closure} \cite{roberson_oddomorphisms_2022} of a graph class $\mathcal{F}$ defined as
	\[
		\cl(\mathcal{F}) \coloneqq \{F \mid \forall G, H.\ G \equiv_{\mathcal{F}} H \implies \hom(F, G) = \hom(F, H)\}.
	\]
	For example, it is known \cite{neuen_homomorphism-distinguishing_2024} that the class  $\mathcal{TW}_k$ of all graphs of treewidth $\leq k$ is \emph{homomorphism distinguishing closed}, i.e.\ $\cl(\mathcal{TW}_k) = \mathcal{TW}_k$.
	By \cite{dvorak_recognizing_2010},
	this means that the function $\hom(F, -)$ has counting width 
	at most $k$ if, and only if, $\tw(F) < k$.
	
	In order to treat homomorphism counts $\hom_{F_n, n}$ of patterns $F_n$ depending on the host size $n$,
	we introduce the \emph{non-uniform homomorphism distinguishing closure} of a graph class.
	For $n \in \mathbb{N}$, let
	\[
		\cl_n(\mathcal{F}) \coloneqq \{F \mid \forall G, H \text{ on $n$ vertices}.\ G \equiv_{\mathcal{F}} H \implies \hom(F, G) = \hom(F, H)\}.
	\]

	Our motivation for introducing this notion is the observation that 
	$\hom_{F_n, n}$ has counting width at most $k$
	if, and only if,
	$F_n \in \cl_n(\mathcal{TW}_{k-1})$ for all $n \in \mathbb{N}$.
	Note that $\cl_n(\mathcal{F})$ differs from $\cl(\mathcal{F})$
	since e.g.\ $K_{n+1} \in \cl_n(\mathcal{F})$ for every graph class $\mathcal{F}$,
	which this is not the case for $\cl(\mathcal{F})$.
	See \cref{app:cl} for further properties of $\cl_n(\mathcal{F})$.
	
	The following theorem
	determines the non-uniform homomorphism distinguishing closure
	of $\mathcal{TW}_k$.
	It also applies to the class of all graphs of pathwidth $\leq k$ \cite{seppelt_homomorphism_2024},
	treedepth $\leq k$ \cite{fluck_going_2024},
	or the class of all planar graphs \cite{roberson_oddomorphisms_2022}.
	All these graph classes satisfy the following \cref{proviso}.
	
	\begin{proviso}\label{proviso}
		Let $\mathcal{F}$ be a graph class such that
		\begin{enumerate}
			\item $\mathcal{F}$ is closed under taking minors,
			\item $\mathcal{F}$ is non-empty, i.e.\ contains the one-vertex graph $K_1$,
			\item $\mathcal{F}$ is closed under disjoint unions, i.e.\ if $F_1, F_2 \in \mathcal{F}$, then $F_1 + F_2 \in \mathcal{F}$, and
			\item $\mathcal{F}$ is closed under weak oddomorphisms, i.e.\ if $F \to G$ is a weak oddomorphism, cf.\ \cref{def:oddomorphism}, and $F \in \mathcal{F}$, then $G \in \mathcal{F}$.
		\end{enumerate}
	\end{proviso}

	\begin{theorem}\label{lem:non-uniform-closure}
		For every graph class $\mathcal{F}$ satisfying \cref{proviso},
		there exists a function $f \colon \mathbb{N} \to \mathbb{N}$
		such that, 
		for every graph $F$ with chromatic number $\chi(F)$,
		if $F \in \cl_{f(\chi(F))}(\mathcal{F})$, then $F \in \mathcal{F}$.
	\end{theorem}
	
	\begin{proof}
		By the Robertson--Seymour theorem \cite{robertson_graph_2004},
		$\mathcal{F}$ is characterised by a finite a set of forbidden minors. 
		Let $N$ denote the maximum number of vertices in any such minor.
		Let $f \colon \mathbb{N} \to \mathbb{N}$ be the function $c \mapsto c (N2^{N-1} +1)$.
		
		By contraposition,
		suppose that $F \not\in \mathcal{F}$.
		Write $c \coloneqq \chi(F)$.
		Hence, $F$ contains a minor on at most $N$ vertices which does not belong to $\mathcal{F}$.
		In particular, $F$ contains an induced minor $P$ on at most $N$ vertices which does not belong to $\mathcal{F}$.
		An \emph{induced minor} of a simple graph $F$ is a graph which can be obtained from an induced subgraph of $F$ by contracting edges.
		By \cite[Lemma~4]{bulian_fixed-parameter_2017},
		$P$ can be assumed to be connected.

		Let $G'_0$ and $G'_1$ denote the \textsmaller{CFI} graphs of $P$ (see also Appendix \ref{sec:cfi}).
		Since $\mathcal{F}$ is closed under weak oddomorphisms, by \cref{thm:rob3.13},
		$G'_0$ and $G'_1$ are homomorphism indistinguishable over $\mathcal{F}$.
		Construct $G_0$ and $G_1$ by making the lexicographic products $G'_0 \cdot K_c$ and $G'_1 \cdot K_c$, respectively,
		adjacent with a copy of $K_c$, i.e.\ 
		$G_i \coloneqq (G'_i \cdot K_c) \boxplus K_c$ for $i \in \{0,1\}$.
		The resulting graphs have at most $N2^{N-1} c + c = f(c)$ vertices.
		%		Regardless of whether $G'_0$ and $G'_1$ are connected,
		%		it holds that $G_0$ and $G_1$ are connected, as $c \geq 1$.
		
		By \cite[Theorem~14]{seppelt_logical_2024}, $G'_0 \cdot K_c$ and $G'_1 \cdot K_c$ are
		homomorphism indistinguishable over $\mathcal{F}$.
		By applying the argument yielding \cite[Theorem~14]{seppelt_logical_2024} to \cref{eq:boxplus}, 
		it follows that $G_0$ and $G_1$ are also homomorphism indistinguishable over $\mathcal{F}$.
		
		We proceed by showing that $\hom(F, G_0) \neq  \hom(F, G_1)$.
		For $i \in \{0,1\}$,
		\begin{align}
			\hom(F, G_i)
			&\overset{\eqref{eq:boxplus}}{=} \sum_{U \subseteq V(F)} \hom(F[U], G'_i \cdot K_c) \hom(F - U, K_c) \notag \\
			&\overset{\eqref{eq:lexprod}}{=} \sum_{U \subseteq V(F)} \sum_{\mathcal{R} \in \Gamma(F[U])} \hom(F[U]/\mathcal{R}, G'_i) \hom(\coprod_{R \in \mathcal{R}} F[R], K_c) \hom(F - U, K_c). \label{eq:lem-induced-minor}
		\end{align}
		Recall that $\Gamma(K)$ for a graph $K$ denotes the set of all partitions $\mathcal{R}$ of $V(K)$ such that $K[R]$ is connected for all $R \in \mathcal{R}$.
		
		By \cref{thm:rob3.13},
		$\hom(F[U]/\mathcal{R}, G'_0) \geq \hom(F[U]/\mathcal{R}, G'_1)$
		for all $U$ and $\mathcal{R}$.
		Furthermore,\\ $\hom(\coprod_{R \in \mathcal{R}} F[R], K_c) > 0$ for all $\mathcal{R}$ and $\hom(F - U, K_c)  > 0$ for all $U$ by construction.
		Moreover, $\hom(F[U]/\mathcal{R}, G'_0) > \hom(F[U]/\mathcal{R}, G'_1)$ for $U$ and $\mathcal{R}$ such that $F[U]/\mathcal{R} \cong P$ by \cref{thm:rob3.13}.
		This choice of $U$ and $\mathcal{R}$ exists as $P$ is an induced minor of $F$.
		It follows that $\hom(F, G_0) \neq \hom(F, G_1)$.
	\end{proof}

	\subsubsection{Homomorphisms Respecting Bipartitions}
	\label{sec:bipartite-respecting}
	
	In order to apply \cref{lem:non-uniform-closure} 
	to prove \cref{thm:dichotomy-chromatic},
	we need to handle the technicality that the graph parameters this article is concerned with do not count general homomorphisms but homomorphisms respecting fixed bipartitions.
	
	For a bipartite graph $F$ without distinguished bipartition and a bipartition $V(F) = A \uplus B$,
	we write $F_{A, B}$ for the bipartite graph $F$ with distinguished bipartition $A \uplus B$.
	The number of all graph homomorphisms
	from $F$ to $G$  which do not necessarily respect any fixed bipartition is denoted by $\hom(F, G)$ while $\hom(F_{A, B}, G_{X, Y})$ denotes the number of homomorphisms $h \colon F \to G$ such that $h(A) \subseteq X$ and $h(B) \subseteq Y$.

	For a class $\mathcal{F}$ of graphs with fixed bipartitions, consider
	\[
		\cl_{n,m}^{\textup{bip}}(\mathcal{F})
		\coloneqq \left\{F \mid \forall G, H \text{ on $(n,m)$ vertices}.\ G \equiv_{\mathcal{F}} H \implies \hom(F, G) = \hom(F,H) \right\}.
	\]
	Here, $F$, $G$, and $H$ range over graphs with fixed bipartitions.
	We identify a class $\mathcal{F}$ of bipartite graphs with the class of all graphs $F_{A,B}$ with fixed bipartition for $F \in \mathcal{F}$ and any bipartition $A,B$ of~$F$. 
	We show the following:
	\begin{proposition}\label{prop:bip-cl}
		For every class $\mathcal{F}$ of bipartite graphs
		and $n \in \mathbb{N}$,
		it holds that
		$\cl_{n,n}^{\textup{bip}}(\mathcal{F}) \subseteq \cl_n(\mathcal{F})$.
	\end{proposition}

	Towards \cref{prop:bip-cl}, we make the following observations:
	\begin{fact}\label{lem:bipartite3}
		For every connected bipartite graph $F$,
		$\hom(F, K_2) = 2$.
	\end{fact}
	In other words, every connected bipartite graph admits a bipartition which is unique up to flipping the sides. 
	
	\begin{lemma}\label{lem:bipartite1}
		Let $F$ and $G$ be bipartite graphs
		with bipartitions $V(F) = A \uplus B$ and $V(G) = X \uplus Y$. Suppose that $F$ is connected.
		Then
		$
		\hom(F, G) = \hom(F_{A, B}, G_{X, Y}) + \hom(F_{A, B}, G_{Y, X}).
		$ 
	\end{lemma}
	\begin{proof}
		First observe that $
		\hom(F, G) \geq \hom(F_{A, B}, G_{X, Y}) + \hom(F_{A, B}, G_{Y, X}).
		$ 
		This is because bipartition-respecting homomorphisms are homomorphisms and the sets of homomorphisms counted by $\hom(F_{A, B}, G_{X, Y}) $ and $\hom(F_{A, B}, G_{Y, X})$ are disjoint.
		
		For the converse inequality, observe that for every homomorphism $h \colon F \to G$, it holds $h(A) \cap h(B) = \emptyset$.
		Indeed, if there are $a \in A$ and $b \in B$ such that $h(a) = h(b)$, then the odd-length path connecting $a$ and $b$ in $F$ is mapped to an odd-length cycle in $G$,
		contradicting that $G$ is bipartite.
		
		Let $G'$ denote the subgraph of $G$ containing the edges and vertices in the image of $F$ under $h$.
		As the image of a connected graph,  $G'$ is connected.
		As subgraph of $G$, $G'$ is bipartite.
		Hence, the bipartition of $G'$ is unique up to flipping sides by \cref{lem:bipartite3}.
		
		This bipartition is, on the one hand, given by $(X \cap V(G')) \uplus (Y \cap V(G'))$ and, on the other hand, 
		given by $h(A) \uplus h(B)$.
		It follows that $h(A) \subseteq X$ and $h(B) \subseteq Y$, or $h(B) \subseteq X$ and $h(A) \subseteq Y$, as desired.
	\end{proof}

	The \emph{bipartite double cover} of a graph $G$ is the graph $G \times K_2$ where  $\times$ denotes the categorical graph product.
	Its vertex set is $V(G) \times \{0,1\}$ and $(v, i)$ and $(w,j)$ are adjacent if, and only if, $vw \in E(G)$ and $i \neq j$.

	\begin{lemma}\label{lem:bipartite2}
		Let $F$ and $G$ be bipartite graphs.
		Let $H$ denote the bipartite double cover of $G$.
		Let $V(F) = A \uplus B$ be the bipartition of $F$ and $V(H) = X \uplus Y \coloneqq (V(G) \times \{0\}) \uplus (V(G) \times \{1\})$ be the canonical bipartition of the bipartite double cover of $G$.
		Then
		$
		\hom(F_{A, B}, H_{X, Y})
		= \hom(F_{A, B}, H_{Y, X}).
		$
	\end{lemma}
	\begin{proof}
		Write $V(H) = V(G) \times \{0,1\}$. 
		The automorphism $\phi$ of $H$ given by $(v, i) \mapsto (v, 1-i)$ for all $v \in V(G)$ and $i \in \{0,1\}$ induces a bijection between the set of homomorphisms $F_{A, B}  \to H_{X, Y}$ and the set of homomorphisms $F_{A, B} \to H_{Y, X}$.
	\end{proof}
	
	\begin{corollary}\label{lem:double-cover-combined}
		Let $F$ be a bipartite graph with bipartition $V(F) = A \uplus B$.
		Let $G$ be a graph.
		Write $H$ for the bipartite double cover of $G$  
		and write $X \uplus Y = V(H)$ for its canonical bipartition.
		Then $\hom(F, G) = \hom(F_{A,B}, H_{X,Y})$.
	\end{corollary}
	\begin{proof}
		Write $F = F^1 + \dots+ F^r$ for the connected components of $F$.
		Then
		\[
			\hom(F, G)
			= \prod_{i =1}^r \hom(F_i, G)
			=  \prod_{i =1}^r \frac{\hom(F_i, H)}{2}
			= \prod_{i =1}^r  \hom((F_i)_{A, B}, H_{X, Y})
			= \hom(F_{A,B}, H_{X,Y}).
		\]
		Here, the first identity follows from \cref{eq:coproduct},
		the second from \cref{lem:bipartite3} together with \cref{eq:product},
		the third from \cref{lem:bipartite2,lem:bipartite1},
		and the fourth from \cref{eq:coproduct}.
	\end{proof}
	
	We now prove \cref{prop:bip-cl}:
	\begin{proof}[Proof of \cref{prop:bip-cl}]
		Let $G$ and $H$ be graphs on $n$ vertices
		such that $G \equiv_{\mathcal{F}} H$.
		By \cref{lem:double-cover-combined},
		the bipartite double covers $G'_{W, X}$ and $H'_{Y, Z}$ of $G$ and $H$
		are homomorphism indistinguishable over $\mathcal{F}$
		as bipartite graphs with fixed bipartitions.
		Let $F_{A,B} \in \cl_{n,n}^{\textup{bip}}(\mathcal{F})$.
		By \cref{lem:double-cover-combined}, 
		$\hom(F, G) = \hom(F_{A,B}, G'_{W, X}) = \hom(F_{A,B}, H'_{Y, Z}) = \hom(F, H)$,
		since $G'_{W,X}$ and $H'_{Y,Z}$ have $(n,n)$ vertices.
		Hence, $F \in \cl_n(\mathcal{F})$.
	\end{proof}

	\subsubsection{Proof of \cref{thm:dichotomy-chromatic}}

	\begin{proof}[Proof of \cref{thm:dichotomy-chromatic}]
		\Cref{thm:hom-circuit-main} yields that \ref{it:hom1} implies \ref{it:hom4}.
		That \ref{it:hom4} implies \ref{it:hom5} is immediate.
		That \ref{it:hom5} implies \ref{it:hom3} follows from \cref{thm:counting-width}.
		By \cref{def:counting-width},
		it holds that \ref{it:hom3} implies \ref{it:hom2}.
		Thus, it remains to show that \ref{it:hom2} implies \ref{it:hom1}.
		
		Write $\mathcal{B}$ for the class of all bipartite graphs and fix a family $(F_{n,m})_{n,m \in \bbN}$ of bipartite multigraphs, as in \cref{thm:dichotomy-chromatic}.
		Analogous to arguments in \cite{dvorak_recognizing_2010},
		if $\hom_{F_{n,m}, n,m}$ has counting width at most $k+1$  on $(n,m)$-vertex simple bipartite graphs,
		then $F_{n,m} \in \cl^{\textup{bip}}_{n,m}(\mathcal{TW}_k \cap \mathcal{B})$ for all $n,m \in \bbN$.
		By \cref{prop:bip-cl,lem:non-uniform-closure},
		for $\ell \coloneqq \min\{n,m\}$ sufficiently larger than $k$,
		\[
			F_{n,m} 
			\in \cl^{\textup{bip}}_{n,m}(\mathcal{TW}_k \cap \mathcal{B}) \cap \mathcal{B}
			\subseteq \cl^{\textup{bip}}_{\ell,\ell}(\mathcal{TW}_k \cap \mathcal{B})  \cap \mathcal{B}
			\subseteq \cl_\ell(\mathcal{TW}_k \cap \mathcal{B})  \cap \mathcal{B}
			\subseteq \cl_\ell(\mathcal{TW}_k) \cap \mathcal{B}
			\subseteq \mathcal{TW}_k.
		\]
		Hence, the $F_{n,m}$ have bounded treewidth.
	\end{proof}
	
	\subsection{Linear Combinations of Homomorphism Polynomials for Patterns of Sublinear Size}
	
	In general, families of symmetric polynomials are not single homomorphism polynomials but linear combinations of homomorphism polynomials (see \cref{lem:gnm-sym}).
	In this section, we tackle this general situation, at least for pattern graphs of sublinear size:
	
	\begin{theorem}
		\label{thm:lincomb}
		For $n,m \in \mathbb{N}$, 
		let $F_{n,m,i}$ be bipartite multigraphs and $\alpha_{n,m,i} \in \mathbb{Q} \setminus \{0\}$.
		Let $p_{n,m} \coloneqq \sum_i \alpha_{n,m,i} \hom_{F_{n,m,i}, n,m}$.
		Suppose that, for all $\nu, \mu \in \mathbb{N}$,
		there are $n', m' \in \mathbb{N}$ 
		such that $\max_i |F_{\nu n, \mu m,i}| \leq \min \{n,m\}$ for all $n > n'$ and $m > m'$.		
		Then the following are equivalent:
		\begin{enumerate}
			\item\label{it:lincomb1} $\max_i \tw(F_{n,m,i})$ is bounded,
			\item\label{it:lincomb2} the counting width of $p_{n,m}$ on $(n,m)$-vertex simple graphs is bounded,
			\item the counting width of $p_{n,m}$ on $(n,m)$-vertex edge-weighted graphs is bounded,
			\item the $p_{n,m}$ admit $\Sym_n\times \Sym_m$-symmetric circuits of orbit size polynomial in $n+m$.
		\end{enumerate}
	\end{theorem}
	
	We prove the theorem by establishing a descriptive complexity monotonicity for non-uniform graph motif parameters whose patterns are of sublinear size. 
	This monotonicity resembles a similar statement proven by \textcite{curticapean_homomorphisms_2017}
	for graph motif parameters and fixed-parameter computations. 
	They show that a linear combination $\sum \alpha_F \hom(F, -)$ is as hard for \textsmaller{FPT} computations as the hardest term $\hom(F, -)$ occurring in it.
	We prove a similar monotonicity for the descriptive complexity of homomorphism counts whose patterns may depend on the size of the host, thus
	generalising \cite[Lemma~6]{seppelt_logical_2024}.
	In order to minimise the technical overhead,
	we first consider the case of not necessarily bipartite graphs.

%	The strategy employed for proving \cref{thm:dichotomy-chromatic} falls short of yielding a criterion for arbitrary linear combinations of homomorphism polynomials since it cannot deal with cancellations between homomorphism counts.
%	It is essential that all homomorphism counts appear with non-negative coefficients in \cref{eq:lem-induced-minor}.
%	This is what allows to amplify the contributions from homomorphism counts of high-treewidth induced minors.
%	Such an argument does not carry through in the general case. 
%	Also, as we argued in \cref{ex:simple-interpolation},
%	it does not suffice to consider homomorphism counts into simple bipartite graphs.
	
	\begin{theorem}[Descriptive Complexity Monotonicity] \label{thm:complexity-monotonicty}
		Let $p_{n}(-) = \sum_i \alpha_{n,i} \hom(F_{n,i}, -)$ for graphs $F_{n,i}$ and non-zero rational coefficients $\alpha_{n,i}$.
		Suppose that $v(n) \coloneqq \max_i |F_{n,i}| \in o(n)$
		and that $\max_i \chi(F_{n,i}) \in O(1)$.
		Let $\mathcal{F}$ be as in \cref{proviso} and suppose that
		\[
			G \equiv_{\mathcal{F}} H \implies
			p_n(G) = p_n(H)
		\]
		for all graphs $G$, $H$ on $n$ vertices.
		Then $F_{n,i} \in \mathcal{F}$ for all sufficiently large $n$ and all $i$.
	\end{theorem}
	\begin{proof}
		By \cref{proviso}
		and since $\hom(F +K_1, G) = n \hom(F, G)$ if $G$ has $n$ vertices,
		we may suppose that none of the patterns in $p_n$ contains isolated vertices.
		Under this assumption,
		\begin{equation}\label{eq:shift}
			p_n(G + \ell K_1) = p_n(G)
		\end{equation}
		for all $(n-\ell)$-vertex graphs $G$ and $n, \ell \in \mathbb{N}$.
		The proof is conducted in two steps.
		We first show the following:
		\begin{claim}\label{claim:monotonicity}
			For every $\nu \in \mathbb{N}$,
			it holds that $F_{\nu n, i} \in \cl_\nu(\mathcal{F})$ for all $n \gg \nu$ and all $i$.
		\end{claim}
		\begin{claimproof}
			Let $G$ and $H$ be $\nu$-vertex graphs such that $G \equiv_{\mathcal{F}} H$.
			Since $v(n) \in o(n)$,
			there exists $n' \in \mathbb{N}$
			such that $v(\nu n) \leq n$ for all $n > n'$.
			Let $K$ be an arbitrary graph on $\ell \leq n$ vertices.
			By \cref{eq:product},
			$G \times K \equiv_{\mathcal{F}} H \times K$.
			By \cref{proviso} and \cite[Theorem~7]{seppelt_logical_2024},
			$G \times K + \nu (n - \ell)K_1 \equiv_{\mathcal{F}} H \times K +  \nu (n - \ell)K_1$.
			Hence, by assumption, noting that $G \times K +  \nu (n - \ell) K_1$ has $\nu n$ vertices, 
			\[
			p_{\nu n}(G \times K) \overset{\eqref{eq:shift}}{=} p_{\nu n}(G \times K +  \nu (n - \ell) K_1 ) = p_{\nu n}(H \times K +  \nu (n - \ell) K_1 )  \overset{\eqref{eq:shift}}{=}  p_{\nu n}(H \times K).
			\]
			By \cref{eq:product},
			\[
				p_{\nu n}(G \times K) = \sum_i \alpha_{\nu n,i} \hom(F_{\nu n,i}, G \times K)  = \sum_i \alpha_{\nu n,i} \hom(F_{\nu n,i}, G)  \hom(F_{\nu n,i}, K) 
			\]
			By \cite{lovasz_operations_1967}, cf.\ \cite[Lemma~6]{seppelt_logical_2024},
			the matrix $(\hom(F, K))_{F, K}$ where $F$, $K$ range over graphs on at most $n$ vertices is invertible.
			Hence,
			\[
				\alpha_{\nu n,i} \hom(F_{\nu n,i}, G) = \alpha_{\nu n,i} \hom(F_{\nu n,i}, H)
			\]
			for all $i$.
			In other words, $F_{\nu n,i} \in \cl_{\nu}(\mathcal{F})$ for all $i$. 
		\end{claimproof}
	
		In order to deduce the theorem from the claim,
		fix by \cref{lem:non-uniform-closure} an integer $\nu \in \mathbb{N}$ such that $F \in \cl_\nu(\mathcal{F}) $ implies $F \in \mathcal{F}$
		for all graphs $F$ whose chromatic number does not exceed $\max_n \max_i \chi(F_{n,i})$.
		
		By \cref{claim:monotonicity,lem:non-uniform-closure},
		$F_{\nu n,i} \in \mathcal{F}$ for all sufficiently large $n$ and all $i$.
		It remains to take care of the $F_{n,i}$ whose first index is not a multiple of $\nu$.
		We do so by padding the $p_n$.
		
		To that end, for every $j < \nu$, 
		consider the graph parameter $q^j$
		defined via $q_n^j(G) \coloneqq p_{n+j}(G + jK_1) = p_{n+j}(G)$ for every $n$-vertex graph $G$,
		by \cref{eq:shift}.
		By \cref{proviso} and \cite[Theorem~7]{seppelt_logical_2024},
		for all graphs $G$ and $H$ on $n$ vertices,
		\[
			G \equiv_{\mathcal{F}} H
			\implies G+jK_1 \equiv_{\mathcal{F}} H+ jK_1
			\implies p_{n+j}(G+jK_1) = p_{n+j}(H+jK_1)
			\implies q_n^j(G) = q_n^j(H).
		\]
		Hence, $q^j$ satisfies the assumptions of the theorem
		and it is $F_{\nu n+j, i} \in \cl_{\nu}(\mathcal{F})$ for all $i$ and
		all $n \gg \nu$ by \cref{claim:monotonicity}.
		Since the $\nu n + j$ for $j < \nu$ and $n \gg \nu$ cover all indices sufficiently larger than $\nu$,
		it follows that $F_{n,i} \in \cl_\nu(\mathcal{F})$ for all $n \gg \nu$ and all $i$.
		The theorem follows by \cref{lem:non-uniform-closure}.
	\end{proof}
	
	For graphs with fixed bipartitions, the proof is analogous.
	
	\begin{theorem}[Descriptive Complexity Monotonicity, Bipartite Case] \label{thm:complexity-monotonicty-bipartite}
		~\\
		Let $p_{n,m}(-) = \sum_i \alpha_{n,m,i} \hom(F_{n,m,i}, -)$ for bipartite graphs $F_{n,m,i}$ and non-zero rational coefficients $\alpha_{n,m,i}$.
		Suppose that, for all $\nu, \mu \in \mathbb{N}$, there exist $n', m' \in \mathbb{N}$,
		such that $\max_i |F_{\nu n, \mu  m,i }| \leq \min\{n,m\}$ for all $n > n'$ and $m > m'$.
		Let $\mathcal{F}$ be a class of bipartite graphs as in \cref{proviso} and suppose that
		\[
		G \equiv_{\mathcal{F}} H \implies
		p_{n,m}(G) = p_{n,m}(H)
		\]
		for all bipartite graphs $G$, $H$ on $(n,m)$ vertices with fixed bipartitions.
		Then $F_{n,m,i} \in \mathcal{F}$ for all sufficiently large $n,m$, and all $i$.
	\end{theorem}
	\begin{proof}
		As in the proof of \cref{thm:complexity-monotonicty},
		it may be supposed that the graphs $F_{n,m,i}$ do not contain isolated vertices.
		The proof is conducted in two steps.
		We first show the following:
		\begin{claim}
			For every $\nu, \mu \in \mathbb{N}$,
			it holds that $F_{\nu n, \mu m, i} \in \cl_{\nu, \mu}^{\textup{bip}}(\mathcal{F})$ for all $n \gg \nu$, $m \gg \mu$, and all $i$.
		\end{claim}
		\begin{claimproof}
			Let $G$ and $H$ be $(\nu, \mu)$-vertex bipartite graphs such that $G \equiv_{\mathcal{F}} H$.
			By assumption,
			there exist $n'$ and $m'$ depending on $\nu$ and $\mu$
			such that  $\max_i |F_{\nu n, \mu  m,i }| \leq \min \{n,m\}$
			for all $n > n'$ and $m > m'$.
			Let $K$ be an arbitrary bipartite graph on $(n'',m'')$ vertices for $n'' \leq n$ and $m'' \leq n$.
			By \cref{eq:product},
			$G \times K \equiv_{\mathcal{F}} H \times K$.
			By \cref{proviso} and \cite[Theorem~7]{seppelt_logical_2024},
			\[ G \times K + \nu(n-n'') K_1 + \mu (m-m'') K_1 \equiv_{\mathcal{F}} H \times K  + \nu(n-n'') K_1 + \mu (m-m'') K_1. \]
			By assumption,
			\begin{align*}
				p_{\nu n, \mu m}(G \times K) 
				&=  p_{\nu n, \mu m}(G \times K + \nu(n-n'') K_1 + \mu (m-m'') K_1 ) \\
				&=  p_{\nu n, \mu m}(H \times K+ \nu(n-n'') K_1 + \mu (m-m'') K_1) \\
				&= p_{\nu n, \mu m}(H \times K).
			\end{align*}
			By \cref{eq:product},
			\[
			p_{\nu n, \mu m}(G \times K) 
			= \sum_i \alpha_{\nu n, \mu m,i} \hom(F_{\nu n, \mu m,i}, G \times K)  
			= \sum_i \alpha_{\nu n, \mu m,i} \hom(F_{\nu n, \mu m,i}, G) \hom(F_{\nu n, \mu m,i}, K) 
			\]
			By \cite{lovasz_operations_1967}, cf.\ \cite[Lemma~6]{seppelt_logical_2024},
			the matrix $(\hom(F, K))_{F, K}$ where $F$, $K$ range over $(n'',m'')$-vertex bipartite graphs for $n'' \leq n$ and $m'' \leq m$ is invertible.\footnote{Strictly speaking, this requires the theorem of \textcite{lovasz_operations_1967} for the category of graphs with fixed bipartitions and bipartition-respecting homomorphisms. One may verify that the original proof goes through in this setting or consult a more category-theoretical reference such as \cite{pultr_isomorphism_1973}.}
			Hence,
			\[
			\alpha_{\nu n, \mu m,i} \hom(F_{\nu n, \mu m,i}, G) = \alpha_{\nu n, \mu m,i} \hom(F_{\nu n, \mu m,i}, H)
			\]
			for all $i$.
			In other words, $F_{\nu n, \mu m,i} \in \cl_{\nu, \mu}^{\textup{bip}}(\mathcal{F})$ for all $i$. 
		\end{claimproof}
		
		In order to deduce the theorem from the claim,
		fix by \cref{lem:non-uniform-closure,prop:bip-cl} an integer $\nu \in \mathbb{N}$ such that $F \in \cl_{\nu,\nu}^{\textup{bip}}(\mathcal{F}) $ implies $F \in \mathcal{F}$
		for all bipartite graphs $F$.
		For every $j, \ell < \nu$,
		consider the graph parameter $q^{j, \ell}$
		defined via 
		$q_{n,m}^{j, \ell}(G) \coloneqq p_{n+j, m+\ell}(G +jK_1 + \ell K_1)$ for every $(n,m)$-vertex bipartite graph $G$.
		As we may assume that the patterns in $p$ do not contain isolated vertices, it holds that
		$q^{j, \ell}_{n, m}(G) = p_{n+j, m+\ell}(G + jK_1 + \ell K_1) = \sum \alpha_{n+j, m+\ell, i} \hom(F_{n+j, m+\ell,i}, G)$ 
		and hence $q^{j, \ell}$ is as in the assumption of the theorem.
		
		By \cref{proviso} and \cite[Theorem~7]{seppelt_logical_2024},
		\begin{align*}
			G \equiv_{\mathcal{F}} H
			&\implies G+jK_1 + \ell K_1 \equiv_{\mathcal{F}} H+ jK_1 + \ell K_1 \\
			&\implies p_{n+j, m+\ell}(G+jK_1 + \ell K_1) = p_{n+j, m+\ell}(H+jK_1+\ell K_1)\\
			&\implies q_{n,m}^{j, \ell}(G) = q_{n,m}^{j, \ell}(H).
		\end{align*}
		Hence, $q^{j,\ell}$ satisfies the assumptions of the theorem
		and it is $F_{\nu n +j, \nu m + \ell, i} \in \cl_{\nu, \nu}^{\text{bip}}(\mathcal{F})$ for all $i$ and
		all $n, m \gg \nu$ by the claim.
		It follows that $F_{n, m,i} \in \cl_{\nu, \nu}^{\textup{bip}}(\mathcal{F})$ for all $n,m \gg \nu$ and all $i$.
	\end{proof}

	\begin{proof}[Proof of \cref{thm:lincomb}]
		Given \cref{thm:main1},
		it suffices to prove that \ref{it:lincomb2} implies \ref{it:lincomb1}.
		This follows from \cref{thm:complexity-monotonicty-bipartite}.
	\end{proof}
	
	\section{Symmetric complexity of subgraph polynomials}

	\Cref{thm:main1} characterises the polynomials admitting symmetric circuits of polynomial orbit size in terms of their representation as linear combinations of homomorphism polynomials.
	By \cref{lem:gnm-sym}, all symmetric polynomials also admit an alternative presentation as linear combinations of \emph{subgraph polynomials}.
	The results from this section suggest that in this subgraph representation, $\mathfrak{T}^k_{n,m}$ can be characterised via another natural graph parameter, namely a variant of the \emph{vertex cover number} that is closed under graph complementation. 
	We establish this characterisation at least in the case of subgraph polynomials for single pattern graphs of sublinear size.
		
	As a first step, we give a sufficient criterion for a subgraph polynomial to be in $\mathfrak{T}^k_{n,m}$.
	This criterion involves the \emph{hereditary treewidth} $\hdtw(F)$ of a bipartite multigraph $F$ defined as the maximum of the treewidth $\tw(F')$ of all bipartite multigraphs $F'$ which admit a vertex- and edge-surjective bipartition-respecting homomorphism $F \to F'$.
	As already observed in \cite{curticapean_homomorphisms_2017}, these graphs $F'$ are precisely those that appear when writing the function $\sub(F, -)$ as a linear combination of functions $\hom(F', -)$ via Möbius inversion, cf.\ \cref{thm:sub-hom}.
	
	\begin{theorem}\label{thm:sub-hom-poly}
		Let $F$ be a bipartite multigraph and $n,m \in \mathbb{N}$.
		Then $\sub_{F, n,m} \in \mathfrak{T}^{k+1}_{n,m}$ for $k \coloneqq \hdtw(F)$.
	\end{theorem}
	\begin{proof}
		By \cref{thm:sub-hom}, $\sub_{F, n,m}$ can be written as linear combination of homomorphism polynomials $\sub_{F', n,m}$ for quotients $F'$ of $F$.
		By definition of hereditary treewidth, $\tw(F') \leq \hdtw(F) = k$.
		Hence, $\sub_{F, n,m} \in \mathfrak{T}^{k+1}_{n,m}$ since the $k+1$ in $ \mathfrak{T}^{k+1}_{n,m}$ denotes the size of the largest bag rather than the treewidth.
	\end{proof}

	The expression for $\sub_{F, n,m}$ derived in \cref{thm:sub-hom-poly} does not depend on $n,m$.
	For this reason, \cref{thm:sub-hom-poly} fails to characterise membership of subgraph polynomials in $\mathfrak{T}^k_{n,m}$,
	as illustrated by the following example which is provided by the machinery developed in \cref{sec:ops}.
	\begin{example} \label{ex:knm}
		For all $n,m \in \mathbb{N}$,
		$\sub_{K_{n,m}, n,m} \in \mathfrak{T}^2_{n,m}$.
	\end{example}
	\begin{proof}
		Since $\sub_{K_{n,m}, n,m} = \prod_{v\in [n]}
                \prod_{w \in [m]} x_{vw}$, this polynomial is in $\mathfrak{T}^2_{n,m}$ by \cref{ex:edge,thm:product-lincomb}.
	\end{proof}
	
	\Cref{ex:knm} shows that $\mathfrak{T}^k_{n,m}$ contains subgraph polynomials of dense graphs.
	Hence, a characterisation of $\mathfrak{T}^k_{n,m}$ in terms of subgraph polynomials cannot involve a monotone graph parameter such as (hereditary) treewidth.
	Towards such a characterisation, we make the following observation:
	
	\begin{theorem}\label{thm:knn}
		For a function $f \colon \mathbb{N} \to \mathbb{N}$ such that $f(n) \leq n$ for all $n \in \mathbb{N}$,
		the following are equivalent:
		\begin{enumerate}
			\item there exists a constant $k \in \mathbb{N}$ such that $\min\{f(n), n- f(n)\} \leq k$ for all $n \in \mathbb{N}$,\label{it:knn1}
			\item the counting width of $\sub_{K_{f(n), f(n)}, n,n}$ on $(n,n)$-vertex simple graphs is bounded,\label{it:knn2}
			\item the counting width of $\sub_{K_{f(n), f(n)}, n,n}$ on $(n,n)$-vertex edge-weighted graphs is bounded,\label{it:knn3}
			\item the $\sub_{K_{f(n), f(n)}, n,n}$ admit $\Sym_n \times \Sym_n$-symmetric circuits of orbit size polynomial in $n$,\label{it:knn4}
			\item there exists a constant $k \in \mathbb{N}$ such that $\sub_{K_{f(n), f(n)}, n,n} \in \mathfrak{T}^k_{n,n}$ for all $n \in \mathbb{N}$,\label{it:knn5}
			\item the polynomials $\sub_{K_{f(n), f(n)}, n,n}$ admit $\Sym_n \times \Sym_n$-symmetric circuits of size polynomial in $n$.\label{it:knn6}
		\end{enumerate}
	\end{theorem}
	\begin{proof}
		First consider the implication \ref{it:knn1} $\Rightarrow$ \ref{it:knn6}:
		Observe that
		\begin{align*}
			\sub_{K_{k,k}, n,n} &= \sum_{\substack{A \subseteq [n] \\ |A| = k}}\sum_{\substack{B \subseteq [n] \\ |B| = k}} \prod_{a \in A} \prod_{b \in B} x_{ab}, \\
			\sub_{K_{n-k,n-k}, n,n} &= \sum_{\substack{A \subseteq [n] \\ |A| = k}}\sum_{\substack{B \subseteq [n] \\ |B| = k}} \prod_{a \in [n] \setminus A} \prod_{b \in [n] \setminus B} x_{ab}.
		\end{align*}
		These formulas represent $\Sym_n \times \Sym_n$-symmetric circuits of size $O(n^{2k})$.
		
		The implications \ref{it:knn6} $\Rightarrow$ \ref{it:knn4} $\Leftrightarrow$ \ref{it:knn5} follow from \cref{thm:main1}.
		The implications \ref{it:knn4} $\Rightarrow$ \ref{it:knn3} $\Rightarrow$ \ref{it:knn2} follow from \cref{thm:counting-width,def:counting-width}.
		The remaining implication \ref{it:knn2} $\Rightarrow$ \ref{it:knn1} is proved by contraposition.
		
		Let $k \geq 2$ be arbitrary.
		By assumption, there exists $n \in \mathbb{N}$ large enough such that
		\[
			\min\{ f(n), n - f(n) \} \geq  k2^{k-1} - k.
		\]
		In particular,  $f(n) \geq k$.
		Let $F \coloneqq K_{f(n), f(n)}$ and write $P \subseteq V(F)$ for a set of vertices such that $F[P] \cong K_{k,k}$.
		Blow up the vertices in $P$ to \textsmaller{CFI} gadgets (see Appendix \ref{sec:cfi}) and call the resulting bipartite graphs $G_0^P$ and $G_1^P$.
		Note that both graphs have $k 2^{k-1} + f(n) - k \leq n$ vertices on each side.
		It follows from \cref{thm:neuen-bipartite} that
		$G_0^P$ and $G_1^P$ are $\mathcal{C}^k$-equivalent as bipartite graphs with fixed bipartitions.
		
		Write $\rho \colon G_i^P \to F$ for the projection map, cf.\ \cref{sec:cfi}.
		We count embeddings $\emb(F, G_i^P)$.
		For a homomorphism $h \colon F \to F$, 
		write $\emb_h(F, G_i^P)$ for the number of embeddings $e \colon F \to G_i^P$ such that $\rho \circ e = h$.
		Consider the following claims:
		
		\begin{claim}\label{cl:knn1}
			If $h \colon F\to F$ is not surjective onto $P$,
			then $\emb_h(F, G_0^P) = \emb_h(F, G_1^P)$.
		\end{claim}
		\begin{claimproof}
			For $u \in V(P)$, write $G_u^P$ for the \textsmaller{CFI} graph where the vertex $u$ carries the odd weight.
			By e.g.\ \cite[Lemma~11]{neuen_homomorphism-distinguishing_2024}, for every $u, v \in P$, there exists an isomorphism $\phi \colon G_u^P \to G_v^P$ such that $\rho \circ \phi = \rho$.
			Hence, $\emb_h(F, G_1^P) = \emb_h(F, G_u^P)$ for every vertex $u \in P$.
			The graphs obtained from $G_0^P$ and $G_u^P$ by removing the gadgets corresponding to $u$ are isomorphic.
			Hence,
			$\emb_h(F, G_0^P) = \emb_h(F, G_u^P) = \emb_h(F, G_1^P)$ for every $u \in P$ such that $u \not\in h(V(F))$.
		\end{claimproof}
		
		\begin{claim}\label{cl:knn2}
			If $h \colon F \to F$ is surjective onto $P$
			and non-injective,
			then $\emb_h(F, G^P_0) = 0 = \emb_h(F, G^P_1)$.
		\end{claim}
		\begin{claimproof}
			By assumption, there exist distinct $a, b \in V(F)$ such that $h(a) = h(b) \eqqcolon v$.
			If $v \not\in P$,
			then $\emb_h(F, G_0^P) = 0 = \emb_h(F, G_1^P)$ since any map $e \colon F \to G_i^P$ such that $\rho \circ e = h$ is non-injective and thus not an embedding.
			Hence, $v \in P$.

			Let $e \colon F \to G_i^P$ be an embedding such that $\rho \circ e = h$.
			Since $h$ is surjective onto $P$,
			both $e(a)$ and $e(b)$ have a shared neighbour in every \textsmaller{CFI} gadget on the opposing side of the bipartition.
			This implies that $e(a) = e(b)$ by \cref{def:cfi},
			contradicting that $e$ is injective. 
		\end{claimproof}
		
		By \cref{cl:knn1,cl:knn2},
		\begin{align*}
			\emb(F, G_0^P) -
			\emb(F, G_1^P)
			& = \sum_{h \colon F \to F} \emb_h(F, G_0^P) - \sum_{h \colon F \to F} \emb_h(F, G_1^P) \\
			&= \sum_{\substack{h \colon F \to F \\ \text{surjective onto } P \\ \text{injective}}} \emb_h(F, G_0^P)
			- \sum_{\substack{h \colon F \to F \\ \text{surjective onto } P \\ \text{injective}}} \emb_h(F, G_1^P)
		\end{align*}
		Here, all $h$ are bipartition-respecting. For $i \in \{0,1\}$, 
		\[
			\sum_{\substack{h \colon F \to F \\ \text{surjective onto } P \\ \text{injective}}} \emb_h(F, G_i^P)
			= f(n)^{\underline{k}} \cdot f(n)^{\underline{k}} \cdot \emb(K_{k,k}, (K_{k,k})_i) \cdot
							 \emb(K_{f(n)-k,f(n)-k}, K_{f(n)-k, f(n)-k})
		\]
		where $n^{\underline{k}} \coloneqq n(n-1) \cdots (n-k+1)$ denotes the falling factorial and $(K_{k,k})_i$ the \textsmaller{CFI} graphs of $K_{k,k}$.
		By \cref{cor:sub-cfi},
		$\emb(F, G_0^P) \neq \emb(F, G_1^P)$.
		Thus, the counting width of $\sub_{K_{f(n), f(n)}, n,n}$ on $(n,n)$-vertex simple graphs is at least $k$.
		As $k$ was chosen arbitrarily, the desired implication follows.
	\end{proof}
	
	\Cref{thm:knn} shows that $\sub_{K_{k,k}, n,n}$ admits small symmetric circuits if, and only if, $k$ is small or $n-k$ is small.
	We generalise the backward direction of this argument for all pattern graphs by introducing a new graph parameter which captures non-uniformity and relaxes hereditary treewidth.
	To that end, we first recall the \emph{vertex cover number} $\vc(F)$ of a graph $F$ which is defined as the minimum size of a set $C \subseteq V(F)$ such that every edge in $F$ is incident to a vertex in $C$.
	Vertex cover number, hereditary treewidth, and matching number are functionally equivalent, cf.\ \cref{sec:hdtw} and \cite{curticapean_homomorphisms_2017}.
	\begin{lemma}\label{lem:vc-mn-hdtw}
		For every graph $F$,
		\( \frac12 \hdtw(F) \leq \mn(F) \leq \vc(F) \leq 2\mn(F) \leq (\hdtw(F) +2)^2. \)
	\end{lemma}

	We relax the graph parameters above as follows:
	
\begin{definition}
	\label{def:bcc}
	Let $n,m \in \mathbb{N}$.
	Let $F$ be a simple $(n,m)$-vertex bipartite graph with bipartition $A \uplus B$.
	Define the simple $(n,m)$-vertex bipartite graph $\overline{F}$ via $V(\overline{F}) \coloneqq V(F)$ and $E(\overline{F}) \coloneqq (A \times B) \setminus E(F)$.
	The \emph{logical vertex cover number of $F$} is
	$
	\cc_{n,m}(F) \coloneqq \min\{\vc(F), \vc(\overline{F})\}.
	$
	For $n' \geq n$ and $m' \geq m$, define $\cc_{n',m'}(F) \coloneqq \cc_{n',m'}(F')$ where $F'$ is obtained from $F$ by adding $n'-n$ isolated vertices on the left side and $m'-m$ isolated vertices on the right side.
	If $n' < n$ or $m' < m$, let $\cc_{n',m'}(F) \coloneqq 0$.
\end{definition}

	For an $(n,m)$-vertex graph $F$, we have $\cc_{n,m}(F) = k$ if, and only if, there exists a set $C \subseteq V(F)$ of size at most $k$ such that $F - C$ is an independent set or a biclique.
	Equipped with this definition,
	we first show that a bound on the logical cover number of a pattern yields polynomial size symmetric circuits for their subgraph polynomials.
	
	\begin{theorem}\label{thm:sub-small-circuit}
		For every family  $(F_{n,m})_{n,m \in \mathbb{N}}$ of simple bipartite graphs,
		if $\cc_{n,m}(F_{n,m})$ is bounded,
		then the $\sub_{F_{n,m}, n,m}$ admit  $\Sym_n \times \Sym_m$-symmetric circuits of size polynomial in $n+m$.
	\end{theorem}
	\begin{proof}
		Fix $F \coloneqq F_{n,m}$ for some $n,m \in \mathbb{N}$ and write $A \uplus B$ for the fixed bipartition of $F$.
		We assume first that $F$ has no isolated vertices and that removing the logical vertex cover results in an independent set rather than a biclique (that is, the logical vertex cover $F$ is indeed a vertex cover).
		Let $K = K_A \uplus K_B$ be a vertex cover of $F$, such that $|K| \leq k$ and $K_A \subseteq A$, $K_B \subseteq B$. We perform the following circuit construction for every $\iota\colon K \hookrightarrow [n] \uplus [m]$. 
		Given~$\iota$, let $p_{K,\iota} \coloneqq \prod_{uv \in E(F[K])} x_{\iota(u)\iota(v)}$. 
		This polynomial is represented by a constant-size circuit which is symmetric under $\StabP(\iota(K))$, the pointwise stabiliser of the set $\iota(K)$ in $\Sym_n \times \Sym_m$.
		Now for every type $S \subseteq K_A$ or $S \subseteq K_B$, introduce a fresh variable $t_S$. If $S \subseteq K_A$, we call $S$ an $A$-type, otherwise a $B$-type. Let
		\[
		q_{\iota} \coloneqq p_{K,\iota} \cdot \prod_{j \in [n] \uplus [m] \setminus \iota(K)} \left( \sum_{S \subseteq K \text{ an A-type} } t_S \cdot \prod_{i \in \iota(S)} x_{ij} + \sum_{S \subseteq K \text{ a B-type} }  t_S \cdot \prod_{i \in \iota(S)} x_{ji}\right)
		\]
		The obvious circuit representation of this has size $O((n+m) \cdot 2^k \cdot k)$ (the number of different types is at most $2^k$), and the circuit is also $\StabP(\iota(K))$-symmetric.
		Now we view $q_{\iota}$ as a polynomial in the variables $\{t_S \mid S \text{ a type} \}$ and with coefficients in $\bbQ[\Xx_{n,m}]$.
		For a type $S$, let $\#(S)$ denote the number of vertices in $V(F) \setminus K$ whose $E(F)$-neighbourhood in $K$ is precisely $S$.
		
		The coefficient of the monomial $\prod_{S \text{ a type}} t^{\#(S)}_S$ in $q_{\iota}$ is the part of $\sub_{F,n,m}$ that sums over all injections $V(F) \hookrightarrow [n] \uplus [m]$ which extend~$\iota$. More formally, this coefficient is
		\[
		\sub_{F,n,m;\iota} = \sum_{\substack{\iota' \colon V(F) \hookrightarrow [n] \uplus [m] \\ \iota'|_K = \iota}} \prod_{vw \in E(F)} x_{\iota'(v)\iota'(w)} .
		\]
		Given a symmetric circuit for $q_{\iota}$, we can compute $\sub_{F,n,m;\iota}$ by interpolation. 
		The polynomial $q_\iota$ has at most $2^k$ variables.
		In each of these variables, its maximal degree is $n+m-k$.
		Thus, by \cref{cor:interpolationTrickCircuits}, $\sub_{F,n,m;\iota}$ can be computed with a $\StabP(\iota(K))$-symmetric circuit of size $O((n+m)^{2^k+1} \cdot 2^k \cdot k)$.
		
		Let $C_{\iota}$ denote the circuit that realises this. The polynomial $\sub_{F,n,m}$ is then computed by the circuit $C$ that simply sums up all $C_{\iota}$, for all injections $\iota: K \hookrightarrow [n] \uplus [m]$. The size of $C$ is $O((n+m)^{2^k+k+1} \cdot 2^k \cdot k)$. It remains to argue that $C$ is symmetric. 
		For this, we firstly observe that for any two $\iota_1, \iota_2 \colon K \hookrightarrow [n] \uplus [m]$, the polynomials $q_{\iota_1}$ and $q_{\iota_2}$ are symmetric to each other via any $(\pi, \sigma) \in \Sym_n \times \Sym_m$ that maps $\iota_1$ to $\iota_2$. Since we use the same points $a_1,...,a_{n+m-k} \in \bbQ$ for the interpolation in each subcircuit $C_{\iota}$, and the coefficients in the linear combination given by \cref{lem:multivariate-polynomial-interpolation} are the same, $C_{\iota_1}$ and $C_{\iota_2}$ are also symmetric to each other. This finishes the case where removing the logical vertex cover results in an independent set. 
		
		In the other case, the graph after removing the logical vertex cover is a biclique.
		Then we perform the same construction as above, with the difference that now, $p_{K,\iota} = \prod_{uv \in E(F[K])} x_{\iota(u)\iota(v)} \cdot \prod_{i \in [n] \setminus \iota(K_A), j\in [m] \setminus \iota(K_B)} x_{ij}$.
		Finally, if isolated vertices are present in $F$, then we just multiply our circuit with the appropriate constant factor to obtain the subgraph count polynomial.
	\end{proof}	
	
	We conjecture that the converse of \cref{thm:sub-small-circuit} holds.
	That is, the parameter $\cc_{n,m}$ is the right parameter for measuring the symmetric circuit complexity of subgraph polynomials.
	\begin{conjecture}[restate=conjSub,label=conj:sub,name=]
		For every family $(F_{n,n})_{n,m \in \mathbb{N}}$ of simple bipartite graphs,
		the following are equivalent:
		\begin{enumerate}
			\item $\cc_{n,m}(F_{n,m})$ is bounded,\label{it:subconj1}
			\item the counting width of $\sub_{F_{n,m},n,m}$ on $(n,m)$-vertex edge-weighted graphs is bounded,\label{it:subconj2}
			\item the $\sub_{F_{n,m}, n, m}$ admit $\Sym_n \times \Sym_m$-symmetric circuits of orbit size polynomial in $n+m$,\label{it:subconj3}
			\item the $\sub_{F_{n,m}, n, m}$ admit $\Sym_n \times \Sym_m$-symmetric circuits of size polynomial in $n+m$.\label{it:subconj4}
		\end{enumerate}
	\end{conjecture}
	
	The implications \ref{it:subconj1} $\Rightarrow$ \ref{it:subconj4} $\Rightarrow$ \ref{it:subconj3} $\Rightarrow$ \ref{it:subconj2} follow from \cref{def:counting-width,thm:counting-width,thm:sub-small-circuit}.
	We prove \cref{conj:sub} for families of patterns of sublinear size.
	\begin{theorem}
		\label{thm:sublinear}
		Let $(F_{n,m})_{n,m \in \mathbb{N}}$ be a family of simple bipartite graphs such that, 
		for all $\nu, \mu \in \mathbb{N}$,
		there exist $n',m' \in \mathbb{N}$,
		such that $|F_{\nu n, \mu m}| \leq \min \{n,m\}$ for all $n > n'$ and $m > m'$.
		The following are equivalent:
		\begin{enumerate}
			\item $\cc_{n,m}(F_{n,m})$ is bounded,\label{it:sublinear1}
			\item $\vc(F_{n,m})$ is bounded,\label{it:sublinear1a}
			\item the counting width of $\sub_{F_{n,m},n,m}$ on $(n,m)$-vertex simple graphs is bounded,\label{it:sublinear2}
			\item the counting width of $\sub_{F_{n,m},n,m}$ on $(n,m)$-vertex edge-weighted graphs is bounded,\label{it:sublinear3}
			\item the $\sub_{F_{n,m}, n, m}$ admit $\Sym_n \times \Sym_m$-symmetric circuits of orbit size polynomial in $n+m$,\label{it:sublinear5}
			\item the $\sub_{F_{n,m}, n, m}$ admit $\Sym_n \times \Sym_m$-symmetric circuits of size polynomial in $n+m$.\label{it:sublinear4}
		\end{enumerate}
	\end{theorem}

	\begin{proof}
		The implications \ref{it:sublinear1} $\Rightarrow$ \ref{it:sublinear4} $\Rightarrow$ \ref{it:sublinear5} $\Rightarrow$ \ref{it:sublinear3} $\Rightarrow$ \ref{it:sublinear2} follow from \cref{thm:sub-small-circuit,thm:counting-width,def:counting-width}.
		By \cref{def:bcc}, the vertex cover number of $F$ is generally an upper bound for $\cc_{n,m}(F)$, so \ref{it:sublinear1a} implies \ref{it:sublinear1}. 
		The remaining implication \ref{it:sublinear2} $\Rightarrow$ \ref{it:sublinear1a} follows from \cref{thm:lincomb,thm:mn-hdtw,fact:mn-vc}.
	\end{proof}
	
	In the remainder of this section,
	we give evidence for the remaining implication.
	
%	\subsection{Invariance of Counting Width of Subgraph Polynomials under Complements}
	
	\Cref{conj:sub} predicts that, for every simple $(n,m)$-vertex bipartite graph $F$, $\sub_{F,n,m}$ has bounded counting width if, and only if, $\sub_{\overline{F},n,m}$ has bounded counting width.
	We prove this consequence (independently of the conjecture).

	\begin{theorem}\label{thm:complements}
		Let $n,m \in \mathbb{N}$.
		For every $(n,m)$-vertex simple bipartite graph $F$,
		the polynomials $\sub_{F, n,m}$ and $\sub_{\overline{F}, n,m}$ have the same counting width on $(n,m)$-vertex edge-weighted graphs.
	\end{theorem}
	\begin{proof}
		The proof is by constructing, given $(n,m)$-vertex bipartite edge-weighted graphs $G$ and $H$ such that $G \equiv_{\mathcal{C}^k} H$ and $\sub_{F, n,m}(G) \neq \sub_{F, n,m}(H)$,
		two $(n,m)$-vertex bipartite edge-weighted graphs $G'$ and $H'$ such that $G' \equiv_{\mathcal{C}^k} H'$ and $\sub_{\overline{F}, n,m}(G') \neq \sub_{\overline{F}, n,m}(H')$.
		
		To that end, we turn the edge-weights from $\mathbb{Q}$ in $G$ into formal variables.
		Let $y_1, \dots, y_r$ denote formal variables, one for every value appearing as an edge weight in $G$ and $H$.
		Let $G''$ and $H''$ denote the $(n,m)$-vertex bipartite graphs whose edges are annotated by these formal variables according to the corresponding values in $G$ and $H$.
		Write $C_1 \uplus \dots \uplus C_r = E(G)$ for the partition of the edges according to this correspondence.
		Write $V(F) = A \uplus B$ and $V(G) = [n] \uplus [m]$ for the bipartitions.
		Consider the following polynomials in $\mathbb{Q}[y_1, \dots, y_r]$:
		\begin{align*}
		\sub_{F, n,m}(G'')
		&= \frac{1}{|\Aut(F)|} \sum_{\substack{h \colon A \uplus B \to [n] \uplus [m] \\ \text{bijective}}} \prod_{i = 1}^r y_i^{|h(E(F)) \cap C_i|}, \\
		\sub_{\overline{F}, n,m}(G'') 
		&= \frac{1}{|\Aut(F)|} \sum_{\substack{h \colon A \uplus B \to [n] \uplus [m] \\ \text{bijective}}} \prod_{i = 1}^r y_i^{|h(E(\overline{F})) \cap C_i|}
		= \frac{1}{|\Aut(F)|}  \sum_{\substack{h \colon A \uplus B \to [n] \uplus [m] \\ \text{bijective}}} \prod_{i = 1}^r y_i^{|C_i| - |h(E(F)) \cap C_i|}.
		\end{align*}
		The final equality holds since $h(E(\overline{F})) \cap C_i = ([n] \times [m] \setminus h(E(F))) \cap C_i = C_i \setminus (h(E(F)) \cap C_i)$  as $F$ and $G$ have the same number of vertices.
		
		It follows that $\sub_{F, n,m}(G''), \sub_{\overline{F}, n,m}(G'') \in \mathbb{Q}[y_1, \dots, y_r]$ are the same polynomials up to a permutation of coefficients.
		More precisely,
		the coefficient of $y_1^{n_1} \cdots y_r^{n_r}$ in $\sub_{F, n,m}(G'')$ for $n_i \in \{0, \dots, |C_i|\}$, $i \in [r]$, 
		equals the number of bijective maps $h \colon A \uplus B \to [n] \uplus [m]$ such that $|h(E(F)) \cap C_i| = n_i$ for $i \in [r]$.
		This number equals the coefficient of $y_1^{|C_1| - n_1} \cdots y_r^{|C_r| - n_r}$ in $ \sub_{\overline{F}, n,m}(G'')$.
		Hence, as formal polynomials in $\mathbb{Q}[y_1, \dots, y_r]$, 
		\[
			 \sub_{F, n,m}(G'') =  \sub_{F, n,m}(H'') \iff  \sub_{\overline{F}, n,m}(G'') =  \sub_{\overline{F}, n,m}(H'').
		\]
		Finally, the edge-weighted graphs $G$ and $H$ such that $\sub_{F, n,m}(G) \neq \sub_{F, n,m}(H)$
		give assignments to the variables $y_1, \dots, y_r$ such that the polynomials $ \sub_{F, n,m}(G'')$ and $\sub_{F, n,m}(H'')$ evaluate to different values.
		Hence,  $ \sub_{\overline{F}, n,m}(G'') \neq  \sub_{\overline{F}, n,m}(H'')$
		and there exist assignments of rationals to $y_1, \dots, y_r$ such that these polynomials evaluate to different values.
		These assignments correspond to $(n,m)$-vertex edge-weighted graphs $G'$ and $H'$ such that $\sub_{\overline{F}, n,m}(G') \neq \sub_{\overline{F}, n,m}(H')$.
		The graphs $G'$ and $H'$ can be obtained from $G$ and $H$ by replacing the edge weights according to some function.
		Hence, $G'$ and $H'$ are $\mathcal{C}^k$-equivalent if $G$ and $H$ are.
	\end{proof}
	\Cref{thm:complements} allows us to separate counting width on simple graphs and counting width on edge-weighted graphs.
	This is in stark contrast to \cref{thm:dichotomy-chromatic}, which shows that, for single homomorphism polynomials, the simple and edge-weighted counting width coincide.
	\begin{example}\label{ex:matching}
		Let $n \in \mathbb{N}$
		and write $F \coloneqq \overline{nK_{1,1}}$ for the bipartite complement of the matching on $(n,n)$-vertices.
		The counting width of the multilinear polynomial $\sub_{F, n,n}$ on $(n,n)$-vertex simple graphs is $2$ while its counting width on $(n,n)$-vertex edge-weighted graphs is $\Theta(n)$.
	\end{example}
	\begin{proof}
		For the first claim, let $G$ be an arbitrary simple $(n,n)$-vertex bipartite graph with bipartition $V(G) = X \uplus Y$.
		If $G$ contains a subgraph isomorphic to $F$,
		then all its vertices have degree $n$ or $n-1$.
		This property can be defined in $\mathcal{C}^2$.
		Write $X' \subseteq X$ and $Y' \subseteq Y$ for the sets of vertices of degree $n$ on each side of the bipartition.
		By the handshaking lemma,  $|X'| = |Y'| \eqqcolon \ell$.
		It holds that $\sub(F,G) = \ell!$.
		Clearly, the number $\ell$ is definable in $\mathcal{C}^2$.
		Hence, the counting width of $\sub_{F,n,n}$ is $2$.

		The second claim follows from \cref{thm:complements} and \cite[Theorem~7.2]{dawar_symmetric_2025} which states that the counting width of the permanent, which equals $\sub_{\overline{F}, n,n}$, on simple graphs is $\Theta(n)$.
	\end{proof}

	\section{Symmetric complexity of the immanants}
	\label{sec:immanants}
	The permanent and determinant are the extreme intractable and tractable case of \emph{immanant} polynomials. Immanants are families of $\Sym_n$-symmetric polynomials $(p_n)_{n \in \bbN}$ that are typically defined via so-called \emph{irreducible characters} of the symmetric group. An irreducible character $f\colon \Sym_n \to \mathbb{C}$ is a class function, i.e.\ it only depends on the multiset of cycle lengths of the input permutation. Given such an $f$, the corresponding immanant is defined as
	\[
	\imm_{f} = \sum_{\pi \in \Sym_n} f(\pi) \prod_{i \in [n]} x_{i,\pi(i)}
	\]
	We view these as polynomials over the field $\mathbb{C}$. 
	If $f$ is constantly $1$, then $\imm_f$ is the permanent. The determinant is obtained by letting $f = \sgn$.
	
	Note that while the permanent is $\Sym_n \times
        \Sym_n$-symmetric, this is not true for every immanant
        (indeed, not even of the determinant).
	But for every $f$, $\imm_{f}$ is $\Sym_n$-symmetric because for any $\sigma \in \Sym_n$, the monomial $\prod_{i \in [n]} x_{\sigma(i),\sigma(\pi(i))}$ encodes a permutation from the same conjugacy class as $\pi$, so it has the same value under $f$. Therefore, we study the complexity of $\Sym_n$-symmetric algebraic circuits for $\imm_{f}$. 
	The $\Sym_n \times \Sym_m$-symmetric polynomials that we have mostly been working with naturally expressed properties of weighted undirected bipartite graphs; now, $\Sym_n$-symmetric polynomials express properties of weighted \emph{directed} (not necessarily bipartite) graphs (see also \cref{rem:squareSymmetricPolynomials}).
	
	In \cite{curticapean2021full}, Curticapean shows a complexity dichotomy for the immanants: In the tractable case, $\imm_{f}$ is in $\VP$ (i.e.\ admits polynomial size algebraic circuits) and is computable in polynomial time. In the intractable case, the family of polynomials is not in $\VP$, unless $\VFPT = \VW$, and also, it is not computable in polynomial time, unless $\FPT = \sharpW$. The complexity is controlled by a certain parameter of $f$. 
	
	We show that the same parameter of $f$ produces the analogous dichotomy with regards to the complexity of $\Sym_n$-symmetric algebraic circuits for the immanants. The benefit of considering symmetric circuits is that the lower bound is not conditional on any complexity-theoretic assumptions. 
	We now introduce the relevant parameter of $f$. As explained in \cite{curticapean2021full}, the irreducible characters $f$ correspond naturally to partitions of $[n]$. If $\lambda$ is a partition of $[n]$, we denote by $\chi^{\lambda}\colon \Sym_n \to \mathbb{C}$ its corresponding irreducible character, and we also write $\imm_{\lambda}$ for $\imm_{\chi^\lambda}$. Partitions of $[n]$ are denoted as tuples $(k_1,\dots ,k_s)$, where each entry $k_i$ denotes a part of size $k_i$.
	For example, the permanent is $\imm_{(n)}$, and the determinant is $\imm_{(1,\dots ,1)}$. For other partitions, the rule for computing their corresponding irreducible characters is more involved and not needed here (for details, see \cite{curticapean2021full}).
	 If $\lambda$ is a partition of $[n]$, let $b(\lambda) = n-s$, where $s$ denotes the number of parts. 
	
	Formally, \cite{curticapean2021full} considers a family $\Lambda$ of partitions. It is said that $\Lambda$ \emph{supports growth} $g \colon \bbN \to \bbN$ if for every $n \in \bbN$ there is a partition $\lambda^{(n)}$ in $\Lambda$ with $b(\lambda^{(n)}) \geq g(n)$ and size $\Theta(n)$. According to the dichotomy from \cite{curticapean2021full}, $(\imm_{\lambda})_{\lambda \in \Lambda}$ is tractable if there is a constant that bounds $b(\lambda)$ for all $\lambda \in \Lambda$. Otherwise, if $\Lambda$ supports growth $\omega(1)$, then the immanant family is (conditionally) intractable. If $\Lambda$ even supports growth $\omega(n^k)$, then $(\imm_{\lambda})_{\lambda \in \Lambda}$ is $\VNP$-complete.
	
Our analogous result for symmetric circuits is the following.
\immanantDichotomy*

This theorem subsumes and generalises the results from \cite{dawar_symmetric_2025}.

\subsection{The tractable case}

We show the first part of \cref{thm:immanantDichotomy}.
In \cite{hartmann1985complexity}, Hartmann gives an algorithm for $\imm_{\lambda}$ that runs in polynomial time if $b(\lambda)$ is bounded. This algorithm is essentially a formula that expresses the immanant in terms of the determinant. For the determinant (over fields of characteristic 0), we know that efficient symmetric circuits exist: 
\begin{lemma}[\cite{dawar_symmetric_2025}]
	\label{lem:determinantEfficient}
	The family $(\det_n)_{n \in \bbN}$ admits $\Sym_n$-symmetric algebraic circuits of polynomial size. 
\end{lemma}	

We now present the relevant details of Hartmann's formula and show that using the above result, we can express it with a symmetric circuit. 
Let $m \leq n$ and let $I = (i_1,\dots ,i_m)$ be a tuple of natural numbers. Define a class function $f_I\colon \Sym_n \to \mathbb{C}$ as follows. 
\[
f_I(\pi) \coloneqq \sgn(\pi)\cdot \prod_{\ell \in [m]} \alpha_\ell(\pi)^{i_\ell}, 
\]
where $\alpha_\ell(\pi)$ denotes the multiplicity of the cycle of length $\ell$ in the cycle decomposition of $\pi$.\\ 

Our first subgoal is a circuit for $\imm_{f_I}$. Following \cite{hartmann1985complexity}, we define certain matrices and use their determinants to compute $\imm_{f_I}$.
A cycle in the graph corresponding to the matrix $(x_{ij})_{1 \leq i,j \leq n}$ is a sequence of variables such that each $x_{ij}$ in this sequence is followed by a variable $x_{jk}$, and the right index of the last variable is equal to the left index of the first variable of the sequence. 
Given a fixed tuple $I = (i_1,\dots ,i_m)$ of cycle multiplicities, we consider tuples of cycles in the graph $(x_{ij})_{1 \leq i,j \leq n}$:
Let $(\sigma_1^{(1)}, \sigma_1^{(2)}, \dots , \sigma_1^{(i_1)}, \sigma_2^{(1)}, \dots , \sigma_2^{(i_2)}, \dots , \sigma_m^{(i_m)})$ denote a tuple of (directed) cycles in the graph $(x_{ij})_{1\leq i,j \leq n}$, where the subscript $\ell$ of each $\sigma$ denotes the length of the cycle, and the superscript indexes the cycles of length $\ell$ from $1$ up to $i_\ell$. 

For a given tuple of cycles, we define a matrix $A(\sigma_1^{(1)},\dots ,\sigma_{m}^{(i_m)})$ in variables $(x_{ij})_{1 \leq i,j \leq n}$ and $(t_\ell^{(k)})$, where $\ell$ and $k$ range over the corresponding sub- and superscripts of $\sigma$. For $i,j \in [n]$, let $X_{ij}$ denote the set of index pairs $(k,\ell)$ such that the cycle $\sigma_{\ell}^{(k)}$ contains the variable $x_{ij}$. The matrix $A(\sigma_1^{(1)},\dots ,\sigma_{m}^{(i_m)})$ is defined by letting
\[
a_{ij} \coloneqq x_{ij}\prod_{(k,\ell) \in X_{ij}} t^{(k)}_{\ell}. 
\]
We extend the action of $\Sym_n$ from the variables $x_{ij}$ to the variables $t_\ell^{(k)}$ by just letting $\pi(t_\ell^{(k)}) = t_\ell^{(k)}$.
It can be checked that every $\pi \in \Sym_n$ maps $A \coloneqq A(\sigma_1^{(1)},\dots ,\sigma_{m}^{(i_m)})$ to $A' \coloneqq A(\pi(\sigma_1^{(1)},\dots ,\sigma_{m}^{(i_m)}))$ such that $\pi(a_{ij}) = a'_{\pi(i),\pi(j)}$. 
The following is a $\Sym_n$-invariant polynomial:
\[
S \coloneqq \sum_{(\sigma_1^{(1)},\dots ,\sigma_{m}^{(i_m)}) \in \Cc} \det A(\sigma_1^{(1)},\dots ,\sigma_{m}^{(i_m)}),
\]
where $\Cc$ is the set of all tuples of cycles $(\sigma_1^{(1)},\dots ,\sigma_{m}^{(i_m)})$ (with the respective lengths as given by the subscripts) such that the union of their edges is a partial cycle cover of $(x_{ij})_{1 \leq i,j \leq n}$ (i.e.\ a disjoint union of cycles). 
Note that $\Cc$ is $\Sym_n$-invariant. 
By \cref{lem:determinantEfficient}, the polynomial $S$ is computable by a $\Sym_n$-symmetric circuit $C_S$ of size $|\Cc| \cdot \text{poly}(n)$: Each $\det A(\sigma_1^{(1)},\dots ,\sigma_{m}^{(i_m)})$ in the sum can be represented by a poly-size $\Sym_n$-symmetric circuit $C(\sigma_1^{(1)},\dots ,\sigma_{m}^{(i_m)})$ that is just like the circuit for $\det_n$ where the inputs are modified to be the products $a_{ij}$ instead of single variables $x_{ij}$. If $\pi \in \Sym_n$ does not fix the tuple of cycles $(\sigma_1^{(1)},\dots ,\sigma_{m}^{(i_m)})$, then the circuit automorphism of $C(\sigma_1^{(1)},\dots ,\sigma_{m}^{(i_m)})$ that $\pi$ extends to turns $C(\sigma_1^{(1)},\dots ,\sigma_{m}^{(i_m)})$ into $C(\pi(\sigma_1^{(1)},\dots ,\sigma_{m}^{(i_m)}))$. 
Thus, because $\Cc$ is $\Sym_n$-invariant, $C_S$ is $\Sym_n$-symmetric.\\
\begin{claim}\label{imm:claim1} If we view $S$ as a polynomial in the $t_\ell^{(k)}$-variables, then the coefficient of the monomial $\prod (t_\ell^{(k)})^\ell$ in $S$ (where the product ranges over the entire tuple of variables $(t_1^{(1)},\dots ,t_1^{(i_1)},\dots ,t_m^{(i_m)})$) is exactly $\imm_{f_I}$.
\end{claim}
\begin{claimproof} Let $M$ be a monomial in $\det A(\sigma_1^{(1)},\dots ,\sigma_{m}^{(i_m)})$. Then its submonomial consisting of $x_{ij}$-variables is a monomial in $\det_n$, so it describes a permutation $\pi \in \Sym_n$. Its submonomial in $t_\ell^{(k)}$-variables is equal to $\prod (t_\ell^{(k)})^\ell$ if and only if every cycle in $(\sigma_1^{(1)},\dots ,\sigma_{m}^{(i_m)})$ is a cycle in the cycle decomposition of $\pi$ (note that $(\sigma_1^{(1)},\dots ,\sigma_{m}^{(i_m)})$ may contain the same cycle multiple times, and then this is also true). 
Thus, the coefficient of $\prod (t_\ell^{(k)})^\ell$ in $S$ is of the form
\[
\sum_{\pi \in P}  \sgn(\pi) \cdot \beta_{\pi} \cdot \prod_{i \in [n]} x_{i\pi(i)},  
\]
where $P \subseteq \Sym_n$ is the set of all permutations whose cycle decomposition contains at least one cycle of length $\ell$, for each $\ell \in [m]$. The coefficient $\beta_\pi$ is the number of tuples 
$(\sigma_1^{(1)},\dots ,\sigma_{m}^{(i_m)}) \in \Cc$ such that each cycle $\sigma_\ell^{(k)}$ appears in the cycle decomposition of $\pi$. Since duplicate cycles can appear in tuples in $\Cc$, we have that $\sgn(\pi) \cdot \beta_\pi = \sgn(\pi) \cdot \prod_{\ell \in [m]} \alpha_\ell(\pi)^{i_\ell} = f_I(\pi)$ because for every cycle length $\ell$, we have $i_\ell$ many slots in each tuple in $\Cc$, each of which can be filled with one of the $\alpha_\ell(\pi)$ many cycles of $\pi$. This proves the claim.
\end{claimproof}

This means that it remains to compute the coefficient of the monomial $\prod (t_\ell^{(k)})^\ell$ in $S$, which can be done by the interpolation method.\\

%One could do this via interpolation but this would be too inefficient. Instead, in \cite{hartmann1985complexity}, this is achieved by taking a series of partial derivatives:
%\textbf{Claim 2:} The coefficient of $\prod (t_\ell^{(k)})^\ell$ in $S$ is equal to
%\[
%\prod \frac{1}\frac{(\ell!)^{i_\ell}}\cdot \frac{\delta}{\delta t_1^{(1)}}\dots \Big(\frac{\delta}{\delta t_m^{(i_m)}}\Big^m S
%\]
%\textit{Proof of claim.} The monomial $\prod (t_\ell^{(k)})^\ell$ in $S$ is maximal in the following sense: There is no monomial in $S$ containing a variable $(t_\ell^{(k)})^{\ell'}$, for an $\ell' > \ell$. This would only be possible if there were a monomial in $\det_n$ that contains $\ell'$ many variables which are in the cycle $\sigma_\ell^{(k)}$ that $t_\ell^{(k)}$ stands for. But the length of $\sigma_\ell^{(k)}$ is $\ell$, so at most $\ell$ variables in any monomial in $\det_n$ can be on this cycle. 
%For this reason, every monomial in $S$ that does not contain $\prod (t_\ell^{(k)})^\ell$ as a submonomial will vanish after taking all these derivatives. The only monomials that remain in the end are the ones in $S$ that do contain $\prod (t_\ell^{(k)})^\ell$, with $\prod (t_\ell^{(k)})^\ell$ being removed. The factor $\prod \frac{1}\frac{(\ell!)^{i_\ell}}$ just corrects for the fact that taking the derivatives introduces the factorials of the exponents of the $t$-variables as coefficients.\\
%\end{comment}
\begin{claim}\label{imm:claim2} Let $d$ be the maximum degree of a $t_\ell^{(k)}$-variable in $S$, and let $t$ denote the number of $t_\ell^{(k)}$-variables. Then $d^t \leq m^{\sum_{\ell \in [m]} i_\ell}$.
\end{claim}
%When $S$ is viewed as a polynomial in the $t_\ell^{(k)}$-variables, it has at most 
%$\prod_{\ell \in [m]} (\ell+1)^{i_\ell}$ many distinct monomials.\\
\begin{claimproof}
We know $t = \sum_{\ell \in [m]} i_\ell$. It remains to show that $d \leq m$.
Let $M$ be a monomial in $S$. Then it is a monomial that appears in some $\det A(\sigma_1^{(1)},\dots ,\sigma_{m}^{(i_m)})$.
This means that the exponent of a variable $t_\ell^{(k)}$ in $M$ can be at most $\ell \leq m$ because $t_\ell^{(k)}$ corresponds to the cycle $\sigma_\ell^{(k)}$ of length $\ell$ and can only appear in a monomial in $\det A(\sigma_1^{(1)},\dots ,\sigma_{m}^{(i_m)})$ together with a variable on the cycle $\sigma_\ell^{(k)}$.
\end{claimproof}
\begin{claim}\label{imm:claim3}
\(
|\Cc| \leq \prod_{\ell \in [m]} \left(\frac{n!}{(n-\ell)!} \right)^{i_\ell}.
\)
\end{claim}
\begin{claimproof} The quantity $\frac{n!}{(n-\ell)!}$ is just the number of different ordered length-$\ell$ cycles in the $n$-vertex graph $(x_{ij})_{1 \leq i,j \leq n}$.
\end{claimproof}

Using \cref{imm:claim1,cor:interpolationTrickCircuits}, we can obtain a symmetric circuit $C_I$ for $\imm_{f_I}$ from the circuit $C_S$. With \cref{imm:claim2,imm:claim3} it follows that 
\[
\norm{C_I} \leq m^{\sum_{\ell \in [m]} i_\ell} \cdot \norm{C_S} \leq  m^{\sum_{\ell \in [m]} i_\ell} \cdot \prod_{\ell \in [m]} \left(\frac{n!}{(n-\ell)!} \right)^{i_\ell} \cdot \operatorname{poly}(n).
\]

This finishes the first subgoal of the proof. Now we use $C_I$ to obtain a circuit for $\imm_{\lambda}$.
The proof of Theorem 1 in \cite{hartmann1985complexity} continues as follows. Let $r$ denote the size of the largest part of $\lambda$, and let $m \coloneqq n-s$, where $s$ is the number of parts of $\lambda$.
Let $\Ii$ be the set of all tuples $(k_1,\dots ,k_m)$ such that $m-r+1 \leq \sum_{\ell \in [m]} \ell k_\ell \leq m$.

Let $\eta$ be the function that takes as input a tuple $(k_1,\dots ,k_m) \in \Ii$ and is defined as
\begin{align*}
\eta(k_1,\dots ,k_m) \coloneqq \Big\{ (c,i_1,\dots ,i_m) \mid &c \in \mathbb{C}, i_1,\dots ,i_m \in \bbN,\\
 &\text{ the monomial } c\cdot \prod_{\ell \in [m]} \alpha_\ell^{i_\ell}  \text{ appears in }  \binom{\alpha_1}{k_1} \dots   \binom{\alpha_m}{k_m} \Big\}.
\end{align*}
In the above definition,  $\binom{\alpha_1}{k_1} \dots  \binom{\alpha_m}{k_m}$ is viewed as a polynomial in formal variables $\alpha_1,\dots ,\alpha_m$.
The following identity is from \cite{hartmann1985complexity}:
\begin{align}
\imm_{\lambda} = \sum_{(k_1,\dots ,k_m) \in \Ii} \chi(k_1,\dots ,k_m) \cdot \Big( \sum_{(c,i_1,\dots ,i_m) \in \eta(k_1,\dots ,k_m)}  c \cdot \imm_{f_{(i_1,\dots ,i_m)}} \Big),\label{eq:star1}
\end{align}
where $\chi(k_1,\dots ,k_m)$ is some $\mathbb{C}$-valued function that we do not need to elaborate on further. In a circuit, we can simply hardwire its respective value.
We have 
\[
| \Ii | \leq \sum_{m-r+1\leq i \leq m} p(i),
\]
where $p(i)$ denotes the number of partitions of $[i]$.
As argued in \cite{hartmann1985complexity}, the polynomial $\binom{\alpha_1}{k_1}\cdot \dots  \cdot \binom{\alpha_m}{k_m}$ has at most $\prod_{\ell \in [m]} (k_\ell+1) \leq \binom{2m}{m}$ many monomials and the degree of each $\alpha_\ell$ is at most $k_\ell$.
Therefore, $|\eta(k_1,\dots ,k_m)| \leq \binom{2m}{m}$, and for every tuple $(c,i_1,\dots ,i_m) \in \eta(k_1,\dots ,k_m)$, it holds $i_\ell \leq k_\ell$ for every $\ell \in [m]$. Using this, we can upper-bound one of the products in the size bound of $C_I$ as follows:
\[
m^{\sum_{\ell \in [m]} i_\ell} \leq m^{\sum_{\ell \in [m]} k_\ell} \leq m^m.
\]	

In total, we can use \cref{eq:star1} and plug in the symmetric circuits $C_I$ for the respective $\imm_{f_{(i_1,\dots ,i_m)}}$. Thus we obtain a symmetric circuit $C_\lambda$ for $\imm_\lambda$ with
\[
\norm{C_\lambda} \in O\Big( \Big(\sum_{m-r+1\leq i \leq m} p(i)\Big) \cdot \binom{2m}{m} \cdot  m^m \cdot \prod_{\ell \in [m]} \Big(\frac{n!}{(n-\ell)!} \Big)^{i_\ell} \cdot \text{poly}(n) \Big).
\]
From \cite{hartmann1985complexity}, we get the following estimation:
\begin{align*}
 \Big(\sum_{m-r+1\leq i \leq m} p(i)\Big) \cdot \binom{2m}{m} \cdot m^m \cdot \prod_{\ell \in [m]} \Big(\frac{n!}{(n-\ell)!} \Big)^{i_\ell} &\leq (m+1)e^{3m^{1/2}}e^{2m} \cdot m^m \cdot n^m\\
  &\leq n^{6(n-s)+4} \cdot (n-s)^{(n-s)}.
\end{align*}
The last step holds for $n > 3$ and uses the fact that $m+1 \leq n$ and $m=n-s$. Since we are in the first case of \cref{thm:immanantDichotomy}, we assume that $m = n-s$ is bounded by a constant $k$. In particular, the factor $(n-s)^{(n-s)}$ and the exponent of $n$ become constant.
So $\norm{C_\lambda}$ is in $O( n^{6k+4})$, multiplied with the symmetric complexity of the determinant, which is also polynomial in $n$.

\subsection{The hard case}
Let us first focus on the second part of \cref{thm:immanantDichotomy}.
Fix a family $\Lambda$ of partitions such that $b(\Lambda)$ is unbounded, i.e.\ it supports growth $g \in \omega(1)$.
We exploit the detailed analysis in \cite{curticapean2021complexitydichotomyimmanantfamilies} to show that the counting width of $(\imm_{\lambda})_{\lambda \in \Lambda}$ on edge-weighted directed graphs is unbounded. The hardness proof in \cite{curticapean2021full} reduces the problem of counting $k$-matchings, for growing $k$, to the immanant. 
Consider the following family of graph parameters, for a given function $h\colon \bbN \to \bbN$. Let $n \in \bbN$ and let $G$ be an undirected graph with $n$ vertices. Then $\Match_{n,h}(G)$ denotes the number of matchings in $G$ consisting of exactly $h(n)$ many edges.
\begin{lemma}
	\label{lem:kMatchingUnboundedCountingWidth}
	For any super-constant function $h(n)$ with $h(n) \leq n/2$, the graph parameter $\Match_{n,h}$ has unbounded counting width on undirected $(n,n)$-vertex bipartite graphs with fixed bipartition.
	Moreover, the number of \emph{perfect} matchings is a graph parameter with counting width $\Omega(n)$ in undirected $(n,n)$-vertex bipartite graphs with fixed bipartition.
\end{lemma}	
\begin{proof}
	Let $k \in \bbN$ be arbitrary.
	For a large enough $n \in \bbN$, we can do the following:
	Let $H'_0, H'_1$ be large enough \textsmaller{CFI} graphs from \cite[Theorem~7.2]{dawar_symmetric_2025}, possibly each augmented with a disjoint copy of a size-$t$ matching, for an appropriate $t \in \bbN$, so that the size of a perfect matching in each of them is precisely $h(n)$; the additional $t$-matching may be needed because $h(n) \in \bbN$ is not necessarily a number that can be the size of a perfect matching in any \textsmaller{CFI} graph.
	The graphs $H'_0$ and $H'_1$ are $\mathcal{C}^k$-equivalent because the \textsmaller{CFI} graphs are by \cite[Theorem~7.2]{dawar_symmetric_2025}, and taking the disjoint union with a $t$-matching in each graph does not change this. The graphs are bipartite (and we can assume the bipartition to be marked with unary relation symbols without breaking $\mathcal{C}^k$-equivalence, cf.\ \cref{thm:neuen-bipartite}).
	Let $H_0, H_1$ be defined by padding $H_0', H_1'$, respectively, with isolated vertices so that they both have exactly $n$ vertices. 
	Then still, $H_0 \equiv_{\Cc^k} H_1$. Moreover, the number of size-$h(n)$ matchings in $H_0$ and $H_1$, respectively, is equal to the number of perfect matchings in $H_0'$ and $H_1'$, respectively, and these numbers are different by \cite{dawar_symmetric_2025} (having added an additional $t$-matching to the graphs does not change the number of perfect matchings).
	Therefore, $\Match_{n,h}(H_0) \neq \Match_{n,h}(H_1)$. Since $k$ was arbitrary, this shows that the counting width of $\Match_{n,h}$ is unbounded.
	For the second part of the lemma, we do not even have to add isolated vertices to $H'_0$ and $H'_1$. They are $n$-vertex bipartite graphs with different numbers of perfect matchings. It also follows from \cite{dawar_symmetric_2025} that $H'_0 \equiv_{\mathcal{C}^{o(n)}} H'_1$.
\end{proof}	
Now one has to go through the proof of \cite[Theorem 2]{curticapean2021complexitydichotomyimmanantfamilies} in some detail: Let $h(n) = \sqrt{g(n)}/24$.
Fix an arbitrary $k \in \bbN$.
By \cref{lem:kMatchingUnboundedCountingWidth}, there exists an $n \in \bbN$ and $\mathcal{C}^k$-equivalent undirected bipartite $(n,n)$-vertex graphs $H_0, H_1$ such that 
$\Match_{n,h}(H_0) \neq \Match_{n,h}(H_1)$. Due to the growth condition on $\Lambda$, there is a $\lambda \in \Lambda$ with $b(\lambda) \geq 24h(n) = \sqrt{g(n)}$ and such that $\lambda$ satisfies certain further properties that we do not need to introduce here. The properties of $\lambda$ fall into two cases, and in each of them, it is shown in \cite{curticapean2021complexitydichotomyimmanantfamilies} that $\Match_{n,h}$ is reducible to $\imm_{\lambda}$. 
In the following, the term \emph{digon} refers to a directed $2$-cycle.

\paragraph*{First case}
In the first case, the second part of Lemma 22 from \cite{curticapean2021full} is used. It states that for any given undirected bipartite $(n,n)$-vertex graph $H$, there exists a directed edge-weighted graph $G(H)$ with edge weights $\{0,1,x\}$ ($x$ being a formal variable), such that for the $\lambda$ at hand we have:
\begin{equation}
\Match_{n,h}(H) = c(\lambda, h(n)) \cdot [x^{2h(n)}] \imm_{\lambda}(G(H)) \label{eq:star2}
\end{equation}
Here, $[x^{2h(n)}] \imm_{\lambda}(G(H))$ denotes the coefficient of the monomial $x^{2h(n)}$ in $\imm_{\lambda}(G(H))$, and $c(\lambda, h(n)) \in \bbQ$ is a constant depending only on $\lambda$ and $h(n)$ but not on the graph $H$ (see \cite[Lemma 26]{curticapean2021complexitydichotomyimmanantfamilies}).

The precise construction of $G(H)$ can be found in  \cite[Lemma 20]{curticapean2021complexitydichotomyimmanantfamilies}:
Let $L \uplus R = V(H)$ be the bipartition of $H$. To obtain $G(H)$, all edges in $H$ are directed from $L$ to $R$, and all edges $ R \times L$ are added to $G$. A set $T$ of fresh vertices is added (whose size depends only on $\lambda$), together with all edges $R \times T$ and $T \times L$. For every $v \in V(H)$, a so-called switch vertex $s_v$ is added, with a self-loop of weight $x$, and a digon between $v$ and $s_v$. Finally, a number (that depends only on $\lambda$) of padding digons and padding vertices with self-loops is added. 
For $q \in \bbQ$, denote by $G_q$ the graph constructed as just described, but with edge weight $q$ instead of $x$.
\begin{lemma}
	\label{lem:raduReductionPreservesCkEquivalence}
	For every $q \in \bbQ$ and bipartite graphs $H_0$, $H_1$ with fixed bipartitions,
	if $H_0 \equiv_{\mathcal{C}^k} H_1$, then $G_q(H_0) \equiv_{\mathcal{C}^{k}} G_q(H_1)$.
\end{lemma}	
\begin{proof}
	It can be checked that after every construction step of $G_q$, the Duplicator still has a winning strategy in the bijective $k$-pebble game played on the two graphs.  
	The first step is directing the edges from $L$ to $R$, and adding all edges $ R \times L$. This new edge relation is definable by a quantifier-free FO-formula because $H_0$ and $H_1$ have their bipartitions marked with unary relations.
	It is clear that $\mathcal{C}^k$-equivalence is preserved through such quantifier-free FO-definitions because any formula speaking about the new graphs can be translated into a formula speaking about the original graphs by replacing the symbol for the edge relation with its defining formula. This does not increase the number of variables, so no $k$-variable formula can distinguish the new pair of graphs.
	 
	The second step is to add a set $T$ of fresh isolated vertices. We can assume them to be coloured with a new unary relation symbol $T$. It is clear that if the two graphs up to this point are $\mathcal{C}^k$-equivalent, then adding the same number of $T$-coloured isolated vertices to each graph does not change this fact because Duplicator still has a winning strategy. In the next step, the edges $R \times T$ and $T \times L$ are inserted. These edges are again definable with a quantifier-free formula, using the colours, so again, $\mathcal{C}^k$-equivalence is preserved.
	The next step is to add a switch vertex $s_v$ for each $v \in L \cup R$, to connect it with $v$ via a digon, and to add a self-loop with weight $q$ (represented as a new binary relation) to every $s_v$.
	Here, one can easily define a winning strategy for Duplicator in the bijective $k$-pebble game on the new pair of graphs, using the Duplicator strategy on the pair before inserting the switch vertices: Any position in the game on the modified structures corresponds to a position in the game on the original structures: If a switch vertex $s_v$ is pebbled, this corresponds to $v$ being pebbled in the original graph. Then if $f$ is Duplicator's chosen bijection in the game on the original graphs, this can be lifted to a bijection on the modified graphs by just mapping $s_{v}$ to $s_{f(v)}$. It is easy to see that Duplicator wins with this strategy. 
	Finally, the same number of isolated digons and vertices with self-loops is added to both graphs, which clearly does not change the fact that Duplicator has a winning strategy.
\end{proof}	

Putting \cref{eq:star2}  together with the equation $\Match_{n,h}(H_0) \neq \Match_{n,h}(H_1)$, we have:
\[
 [x^{2h(n)}] \imm_{\lambda}(G(H_0)) \neq [x^{2h(n)}] \imm_{\lambda}(G(H_1)).
\]
It follows that there must exist a $q \in \bbQ$ such that $\imm_{\lambda}(G_q(H_0)) \neq \imm_{\lambda}(G_q(H_1))$, for if we had $\imm_{\lambda}(G_q(H_0)) = \imm_{\lambda}(G_q(H_1))$ for every $q \in \bbQ$, then $\imm_{\lambda}(G(H_0))$ and $\imm_{\lambda}(G(H_1))$ would be the same polynomial in $\bbQ[x]$, contradicting the statement above. 
By \cref{lem:raduReductionPreservesCkEquivalence}, we have $G_q(H_0) \equiv_{\mathcal{C}^k} G_q(H_1)$, which proves that $\imm_{\lambda}$ is not $\mathcal{C}^k$-invariant on $\bbQ$-edge-weighted graphs. 

\paragraph*{Second case}
In the second case, the properties of $\lambda$ are different, and the second part of \cite[Lemma 28]{curticapean2021full} has to be applied. 
This describes another graph construction, let us call it $G'(H)$, such that
\begin{equation}
\Match_{n,h}(H) = c'(n, |E(H)|, h(n), \lambda) \cdot \imm_{\lambda}(G'(H)) \label{eq:2stars}
\end{equation}
Here, the constant $c'$ depends on $n=|V(H)|, |E(H)|, h(n)$ and $\lambda$ (see \cite[Lemma 28]{curticapean2021complexitydichotomyimmanantfamilies}), but not on $H$ itself.
The graph $G'(H)$ is directed and has edge weights $\{0,1,-1\}$ if $H$ is unweighted.
The construction of $G'(H)$ is the following \cite[Lemma 28]{curticapean2021complexitydichotomyimmanantfamilies}: Replace
every edge $uv$ in $H$ with an edge gadget consisting of $u, v$, two
fresh vertices, and five digons, as depicted in \cref{fig:gadget} (see also page~6 of \cite{curticapean2021complexitydichotomyimmanantfamilies}).
\begin{figure}
\centering
\begin{tikzpicture}[->, thick, shorten <= 2pt, shorten >= 2pt, every node/.style={circle, draw, fill=black, minimum size=0.3cm}]
	% Nodes
	\node[label=$u$] (A) {};
	\node[above right = 1cm and 3cm of A] (B) {};
	\node[below right = 1cm and 3cm of A] (C) {};
	\node[label=$v$, right = 6cm of A] (D) {};
	% Edges
	\draw (A) edge[bend left = 20] (B); 
	\draw (B) edge[bend left = 20] (A); 
	
	\draw (A) edge[bend left = 20] (C); 
	\draw (C) edge[bend left = 20] (A);
	
	\draw (B) edge[bend left = 20] node[draw=none, fill=none, midway, right]{$-1$} (C); 
	\draw (C) edge[bend left = 20] (B); 
	
	\draw (B) edge[bend left = 20] (D); 
	\draw (D) edge[bend left = 20] (B); 
	
	\draw (C) edge[bend left = 20] (D); 
	\draw (D) edge[bend left = 20] (C); 
\end{tikzpicture}
\caption{The gadget that replaces an edge $uv$.\label{fig:gadget}}
\end{figure}

This edge gadget is symmetric between $u$ and $v$. Next, for each vertex $v \in V(H)$, add a fresh \enquote{switch vertex} $s_v$ connected to $v$ with a digon. Add another $2h(n)$ many fresh vertices, called \enquote{receptor vertices}, and add an edge between each pair of receptor and switch vertex. Replace each of these new edges with the aforementioned edge gadget. 
Finally, a number of isolated digons and a number of isolated vertices is added. These numbers depend on $\lambda, |V(H)|, |E(H)|$, and $h(n)$.
\begin{lemma}
	\label{lem:raduReductionPreservesCkEquivalence2}
	Let $k \geq 4$.
	If $H_0, H_1$ are undirected simple bipartite graphs with fixed bipartition, and $H_0 \equiv_{\mathcal{C}^k} H_1$, then $G'(H_0) \equiv_{\mathcal{C}^{k/2}} G'(H_1)$.
\end{lemma}	
\begin{proof}
Again, we verify that every construction step of $G'$ preserves $\mathcal{C}^k$-equivalence. 
Most of it works just like in \cref{lem:raduReductionPreservesCkEquivalence}. The only additional construction step here is the replacement of every edge with the gadget depicted above. 
Since these edge gadgets contain ``internal'' vertices, they give Spoiler the possibility to fix an edge using just one pebble. This is why we need to halve the number of pebbles in this construction step to ensure that Duplicator can still win. 
The argument is the following. We now start with $H_0$ and $H_1$ and
consider only the construction step where each edge is replaced with
the gadget. Consider a position in the bijective $k$-pebble game on
$H_0$ and $H_1$. The Duplicator winning strategy gives us a bijection
$f\colon V(H_0) \to V(H_1)$ that respects the $\mathcal{C}^k$-types of
the vertices, with the $\leq k-1$ pebbles on the board as
parameters. This is the bijection that Duplicator plays. If at most
$k-2$ pebbles are currently on the board, then in addition to $f$,
there also exists a bijection $f_2\colon V(H_0)^2 \to V(H_1)^2$ on the
pairs of vertices which respects the parametrised
$\mathcal{C}^k$-types of the pairs. If such a bijection did not exist,
then the tuples of pebbled vertices would have different
$\mathcal{C}^k$-types and Duplicator would not be able to win from
this position. (Note that $f$ and $f_2$ in general do not agree on the vertices, unless $H_0$ and $H_1$ are isomorphic.)   
With this in mind, we can define a Duplicator winning strategy in the $k/2$-pebble game on the modified graphs: A position on the modified graphs translates back into a pebble position on the original graphs as follows: Pebbles on original vertices are kept, and each pebble on an internal vertex of the gadget of an edge $uv$ is replaced by two pebbles, one on $u$ and one on $v$, in the original graph. Before Duplicator gets to specify a bijection in the game on the modified structures, there are at most $k/2-1$ many pebbles on each graph, so the corresponding pebble position in the original graphs has at most $k-2$ many pebbles on each graph. Because Duplicator has a winning strategy in the $k$-pebble game on the original graphs, there exist a $\mathcal{C}^k$-type preserving bijection $f$ on vertices, and a $\mathcal{C}^k$-type preserving bijection $f_2$ on the vertex pairs, as argued before. Because the existence of an edge is part of the $\mathcal{C}^k$-type of a pair, $f_2$ preserves edges. We use $f$ and $f_2$ to define Duplicator's bijection in the game on the modified graphs: The internal vertices of the edge gadget of an edge $e$ are mapped to the corresponding internal vertices of the gadget of the edge $f_2(e)$. The non-gadget vertices are mapped according to $f$. Because $f$ and $f_2$ are type-preserving, after Spoiler's move in the game on the modified graphs, when the position is translated back to the other game, the pebbled tuples again have the same $\mathcal{C}^k$-type, so Duplicator can indeed maintain this invariant and win the game.

The other construction steps can be dealt with as in \cref{lem:raduReductionPreservesCkEquivalence}, noting that the introduction of edges between all pairs of receptor and switch vertices is FO-definable if we assume the set of receptor and switch vertices to be coloured with a distinguished colour, respectively. 
In the final construction step, we just have to note that the number of isolated vertices and digons that are added depend on $\lambda, |V(H)|, |E(H)|$, and $h(n)$, and we have $|V(H_0)| = |V(H_1)|$ and $|E(H_0)| = |E(H_1)|$ because $H_0$ and $H_1$ are $\mathcal{C}^k$-equivalent, and $k \geq 2$.
\end{proof}	
The constant $c'(n, |E(H)|, h(n), \lambda) \in \bbQ$ in \cref{eq:2stars} is in particular equal for $\mathcal{C}^k$-equivalent graphs $H_0, H_1$, as they have the same number of edges and vertices.
Thus, as before, \cref{eq:2stars} implies that $\imm_{\lambda}(G'(H_0)) \neq \imm_{\lambda}(G'(H_1))$. The two graphs are $\mathcal{C}^{k/2}$-equivalent by \cref{lem:raduReductionPreservesCkEquivalence2}, so $\imm_{\lambda}$ is not $\mathcal{C}^{k/2}$-invariant on edge-weighted directed graphs. Since $k$ was arbitrary, we find such an example of graphs for every $k$, so the counting width of $\imm_{\lambda}$ is unbounded.

The two cases together show the second part of \cref{thm:immanantDichotomy} because they prove that the family $\imm_{\Lambda}$ is not $\mathcal{C}^k$-invariant for any fixed $k \in \bbN$. 
	
The first part of \cref{thm:counting-width} then implies that $\imm_{\Lambda}$ does not admit polynomial size symmetric circuits.

The third part of \cref{thm:immanantDichotomy} is shown analogously. The proof in \cite{curticapean2021complexitydichotomyimmanantfamilies} is then again divided into two cases, in each of which we get \cref{eq:star2,eq:2stars}, respectively, but where $\Match_{n,h}(H)$ is now replaced with the number of \emph{perfect} matchings. \cref{lem:kMatchingUnboundedCountingWidth} then gives us graphs that differ with respect to the number of perfect matchings and are $\mathcal{C}^{o(n)}$-equivalent. It is easy to check that \cref{lem:raduReductionPreservesCkEquivalence} and \cref{lem:raduReductionPreservesCkEquivalence2} are also true for the graph constructions from \cite{curticapean2021complexitydichotomyimmanantfamilies} that are used here (they use exactly the same gadgets as before).
So in total, if $g \in \Omega(n^k)$, then it follows that for any sublinear function $f(n)$, and every large enough $n$, we can find graphs $G(H_0) \equiv_{\mathcal{C}^{f(n)}} G(H_1)$ with $\imm_\lambda(G(H_0)) \neq \imm_\lambda(G(H_1))$. Hence, the counting width of the immanant is $\Omega(n)$.

The $2^{\Omega(n)}$ size lower bound on symmetric circuits follows from the second part of \cref{thm:counting-width}.

\section{Conclusion}

In recent years, a rich theory of symmetric computation has been emerging which has established a tight and surprising relationship between logical definability, circuit complexity and also computation in other models such as linear programming (see~\cite{dawar_csl2020,dawar_icalp2024} for pointers).  A central plank of this is the use of variations of the Cai-F\"urer-Immerman construction to establish unconditional lower bounds in the context of models of computation with natural symmetries.  The work in~\cite{dawar_symmetric_2020,dawar_symmetric_2025} is a significant example of this, showing an unconditional exponential separation between the complexity of the determinant and the permanent polynomials in the context of symmetric algebraic circuits.

The present work gives a sweeping generalisation of these results in two different directions.  First, considering matrix polynomials symmetric under arbitrary permutations of the rows and columns (i.e.\ the $\Sym_n \times \Sym_m$ action), we give a complete characterisation of the tractable cases in terms of homomorphism polynomials of bounded treewidth, linking this study to other complexity classifications such as~\cite{curticapean_homomorphisms_2017}.  For polynomials on \emph{square} matrices symmetric under the simultaneous application of a permutation to the rows and columns (i.e.\ the $\Sym_n$ action), we generalise the separation of the determinant and permanent by giving a complete classification of the symmetric complexity of all immanants.

The work raises a number of directions for future research.  It suggests that there is a close connection between the counting width of polynomials, which we define, and their symmetric algebraic complexity. We have not established this in full generality and thus the first important direction is to prove or disprove \cref{conj:sub}. A step in this direction would be to understand the complexity of linear combinations of homomorphism polynomials where neither \cref{thm:lincomb} nor \cref{thm:dichotomy-chromatic} applies. 
That is to say, there is more than one homomorphism polynomial involved and the size of the pattern graphs is not sublinearly bounded.

Our most general results are for matrix polynomials symmetric under $\Sym_n \times \Sym_m$, where we have a complete characterisation of tractable families. The polynomials on square matrices symmetric under the $\Sym_n$ action include many that are not $\Sym_n \times \Sym_n$-symmetric (for instance, the determinant and many other immanants).  While we give a complete classification of the immanant polynomials, there are many other such polynomial families. It seems plausible that a complete classification in terms of homomorphism polynomials from directed graphs might be achievable in this case, and we leave it for future work.

It would be interesting to relate this to other methods for
establishing lower bounds in algebraic complexity and to natural
algorithmic methods in that field.  Indeed, one of the striking
features of the study of symmetric computation models is that a
surprising range of algorithms that are used (for instance in
combinatorial optimisation) are indeed symmetric and thus the
unconditional lower bounds show the limitations of standard
techniques.

Finally, this work may provide the means to prove lower bounds for
symmetric circuits for interesting polynomial families beyond the
determinant and the permanent.  One potential application would be in
proof complexity.  There has been considerable recent interest in the
\emph{Ideal Proof System}~\cite{grochow_pitassi}, in which refutations of unsatisfiable Boolean formulas or systems of polynomial equations take the
form of algebraic circuits.  One could conceivably use the methods
developed here to show that certain families of formulas with natural
symmetries do not admit succinct symmetric refutations.

	\newpage
	\printbibliography
	\newpage
	
	\appendix 
	
	\section{Möbius Inversion over Polynomial Rings}
	
	Let $L$ be a finite poset and $R$ be a ring.
	Define the \emph{Möbius function} $\mu \colon L \times L \to R$ of $L$ recursively via
	\begin{align*}
		\mu(s,s) &= 1 && \text{for all } s \in L, \\
		\mu(s,u) &= - \sum_{s \leq t < u} \mu(s, t) && \text{for all } s < u \text{ in } L.
	\end{align*}
	
	\begin{lemma}[{\cite[Proposition~3.7.2]{stanley_enumerative_1997}}]\label{lem:moebius}
		Let $L$ be a finite poset, $R$ a ring, and $M$ an $R$-module.
		Let $f, g \colon L \to M$. Then
		\[
		g(s) = \sum_{t \geq s} f(t) \text{ for all } s\in L
		\iff
		f(s) = \sum_{t \geq s} g(t) \mu(s,t) \text{ for all } s \in L.
		\]
	\end{lemma}
	\begin{proof}
		Assume the first identity holds. Then
		\begin{align*}
			\sum_{t \geq s} g(t) \mu(s,t) &= \sum_{t \geq s} \mu(s,t) \sum_{u \geq t} f(u) \\
			&= \sum_{u \geq s} f(u) \sum_{u \geq t \geq s} \mu(s, t) \\ 
			&= \sum_{u \geq s} f(u) \delta_{u = s} \\
			&= f(s).
		\end{align*}
		The converse follows analogously.
	\end{proof}

	\Cref{lem:moebius} is applied to the lattice $\Pi(A)$ of partitions of some finite set $A$.
	The Möbius function of the partition lattice is given by 
	the Frucht--Rota--Schützenberger formula \cite[(A.2)]{lovasz_large_2012}.
	For a partition $\pi \in \Pi(A)$,
	\begin{equation}\label{eq:frs}
		\mu_{\pi} \coloneqq (-1)^{|A| - |A/\pi|} \prod_{ P \in A/\pi} (|P|-1)!.
	\end{equation}

	\begin{lemma} \label{lem:moebius-polynomial}
		Let $A$ and $I$ be finite sets. 
		For every map $h \colon A \to I$, let $p_h \in M$.
		Then, the following hold:
		\begin{align*}
			\sum_{h \colon A \to I} p_h &= \sum_{\pi \in \Pi(A)} \sum_{h \colon A/\pi \hookrightarrow I} p_{h \circ \pi}, \\
			\sum_{h \colon A \hookrightarrow I} p_h &= \sum_{\pi \in \Pi(A)} \mu_\pi \sum_{h \colon A/\pi \to I} p_{h \circ \pi}.
		\end{align*}
	\end{lemma}
	\begin{proof}
		We first verify the first identity.
		To that end, note that the set of maps $h \colon A \to I$ can be partitioned according to the partition $\pi \in \Pi(A)$ the map $h$ induces.
		In other words,
		\begin{align*}
			\{(\pi, h) \mid \pi \in \Pi(A), h \colon A/\pi \hookrightarrow I \} &\to I^A \\
			(\pi, h) & \mapsto h \circ \pi
		\end{align*}
		is a bijection.
		Towards the second identity, we verify the assumptions of \cref{lem:moebius}.
		For a partition $\pi \in \Pi(A)$,
		let
		\[
			f(\pi) \coloneqq \sum_{h \colon A/\pi \hookrightarrow I} p_{h \circ \pi}, \quad \text{and} \quad
			g(\pi) \coloneqq \sum_{h \colon A/\pi \to I} p_{h \circ \pi}.
		\]
		Now, the first identity reads as $g(\bot) = \sum_{\pi \in \Pi(A)} f(\pi)$
		while the second identity reads as $f(\bot) = \sum_{\pi \in \Pi(A)} \mu_\pi g(\pi)$.
		Towards the assumptions of \cref{lem:moebius}, let $\pi \in \Pi(A)$ be arbitrary.
		Then
		\begin{align*}
			g(\pi) &= \sum_{h \colon A/\pi \to I} p_{h \circ \pi} \\
				   &= \sum_{\sigma \in \Pi(A/\pi)} \sum_{h \colon (A/\pi)/\sigma \hookrightarrow I} p_{(h \circ \sigma) \circ \pi}  \\
				   &= \sum_{\substack{\sigma \in \Pi(A) \\ \sigma \geq \pi}} \sum_{h \colon A/\sigma \hookrightarrow I} p_{h \circ \sigma} \\
				   &= \sum_{\substack{\sigma \in \Pi(A) \\ \sigma \geq \pi}} f(\sigma),
		\end{align*}
		as desired.
		\Cref{lem:moebius} yields the second identity.
	\end{proof}
	
	Consider the following corollary.
	Let $F$ be a bipartite graph with bipartition $A \uplus B = V(F)$.
	For $n, m\in \mathbb{N}$, consider the \emph{embedding polynomial}
	\[
	\emb_{F, n,m} \coloneqq \sum_{h \colon A \uplus B \hookrightarrow [n] \uplus [m]} \prod_{vw \in E(F)} x_{h(v)h(w)}.
	\]
	Clearly, $\emb_{F, n,m} = |\Aut(F)| \cdot \sub_{F, n,m}$.

	\begin{corollary} \label{thm:sub-hom}
		Let $F$ be a bipartite multigraph with bipartition $A \uplus B = V(F)$.
		Let $n,m \in \mathbb{N}$.
		Then
		\[
		\sub_{F, n,m} = \frac{1}{|\Aut(F)|} \sum_{\pi \in \Pi(A)} \mu_{\pi} \sum_{\sigma \in \Pi(B)} \mu_{\sigma} \hom_{F/(\pi, \sigma), n, m}.
		\]
	\end{corollary}
	
	\begin{proof}
		Applying \cref{lem:moebius-polynomial},
		we treat the two parts of the bipartition separately.
	
		For maps $h \colon A \to [n]$ and $h' \colon B \to [m]$, define the polynomial
		\[
			p_{h,h'} \coloneqq \prod_{vw \in E(F)} x_{h(v)h'(w)}.
		\]
		Applying \cref{lem:moebius-polynomial} twice,
		\begin{align*}
			\emb_{F, n, m} 
			&= \sum_{h \colon A \hookrightarrow [n]} \sum_{h' \colon B \hookrightarrow [m]} p_{h,h'} \\
			&= \sum_{\pi \in \Pi(A)} \mu_\pi \sum_{h \colon A/\pi \to [n]} \sum_{h' \colon B \hookrightarrow [m]} p_{h \circ \pi,h'} \\
			&= \sum_{\pi \in \Pi(A)} \mu_\pi \sum_{h \colon A/\pi \to [n]} \sum_{\sigma \in \Pi(B)} \mu_\sigma \sum_{h' \colon B/\sigma \to [m]} p_{h \circ \pi, h' \circ \sigma} \\
			&= \sum_{\pi \in \Pi(A)} \sum_{\sigma \in \Pi(B)} \mu_\pi \mu_\sigma \hom_{F/(\pi, \sigma), n,m}. \qedhere
		\end{align*}
	\end{proof}
	
	\begin{example}
		Consider the $F = P_3$, i.e.\ the path on three vertices where $A$ is a singleton and $B$ is of size two. Then
		\[
		2\sub_{F, n, m} = \sum_{i \in [n]} \sum_{\substack{k, j \in [m] \\ k \neq j}} x_{ij}x_{ik}.
		\]
		There is only one partition in $\Pi(A)$, namely $\bot = \top$ with $\mu_\bot = 1$.
		In $\Pi(B)$, there are two partitions, namely $\bot$ and $\top$ with $\mu_\bot = 1$ and $\mu_\top = -1$.
		It follows that
		\[
		2\sub_{F, n, m} =  \hom_{F/(\bot, \bot), n,m} - \hom_{F/(\bot, \top), n,m}
		= \hom_{P_3, n, m} - \hom_{K_2^2, n, m}
		=  \sum_{i \in [n]} \sum_{k, j \in [m]} x_{ij}x_{ik} - \sum_{i \in [n]} \sum_{k \in [m]} x_{ik}^2
		\]
		where $K_2^2$ denotes the loopless multigraph with two vertices and two edges.
		In particular, for $m = 1$, the polynomial $\sub_{F, n, m}$ is the zero polynomial.
	\end{example}

	\section{Interpolation}
	
	\begin{lemma}[Multivariate Polynomial Interpolation] \label{lem:multivariate-polynomial-interpolation}
		Let $\mathbb{K}$ be a field and $A$ be a $\mathbb{K}$-algebra.
		Let $p \in A[x_1, \dots, x_t]$ be a multivariate polynomial such that the degree in each of the variables is less than~$n$.
		Let $a_1, \dots, a_n \in \mathbb{K}$ be distinct.
		Then the coefficients of $p$ are $\mathbb{K}$-linear combinations of the $p(a_{j_1}, \dots, a_{j_t})$ for $j_1, \dots, j_t \in [n]$. The coefficients of the respective linear combination only depend on $a_1, \dots, a_n$ and on the coefficient of $p$ to be expressed by the linear combination.
	\end{lemma}
	
	Note that the polynomial is evaluated at $n^t$ points in $\mathbb{K}^t$.
	We think of applying the lemma to $\mathbb{K} = \mathbb{Q}$ and $A = \mathbb{Q}[y_1,\dots, y_t]$.
	
	\begin{proof}
		Write $p = \sum_{0 \leq i_1,\dots, i_t < n} \alpha_{i_1\dots i_t} x^{i_1}_1 \cdots x^{i_t}_t$ 
		for $\alpha_{i_1,\dots, i_t} \in A$.
		Then for all $j_1, \dots, j_t \in [n]$
		\[
		p(a_{j_1}, \dots, a_{j_t}) = \sum_{0 \leq i_1,\dots, i_t < n}  a^{i_1}_{j_1} \cdots a^{i_t}_{j_t} \alpha_{i_1\dots i_t}
		\]
		Define the $\{1, \dots, n\} \times \{0, \dots, n-1\}$-matrix $V$ such that $V_{ji} \coloneqq a_{j}^i$.
		By assumption, $V$ is invertible with entries in $\mathbb{K}$.
		Hence, the Kronecker product $W \coloneqq V \otimes \dots \otimes V$ of $t$ copies of $V$  is invertible. 
		The above equation now reads as
		\[
		p(a_{j_1}, \dots, a_{j_t}) = \sum_{0 \leq i_1,\dots, i_t < n}  W_{j_1\dots j_t, i_1\dots i_t} \alpha_{i_1\dots i_t}
		\]
		Hence, for all $0 \leq i_1,\dots, i_t < n$,
		\[
		\alpha_{i_1\dots i_t} = \sum_{j_1, \dots, j_t \in [n]} W^{-1}_{i_1\dots i_t, j_1\dots j_t} p(a_{j_1}, \dots, a_{j_t}).
		\]
		Thus, the coefficients $\alpha_{i_1\dots i_t}$ are linear combinations of the $p(a_{j_1}, \dots, a_{j_t})$ with coefficients in $\mathbb{K}$.
	\end{proof}

	\begin{corollary}
		\label{cor:interpolationTrickCircuits}
		Let $\bbF$ be a field. 	
		Let $p \in \bbF[x_1, \dots, x_t, y_1,\dots y_r]$. 
		Let $n$ be an upper bound on the degree of any of the variables $x_i$ in $p$.
		Let $q \in \bbF[y_1,\dots y_r]$ be the coefficient of a monomial in the $x_i$-variables in $p$, when $p$ is viewed as a polynomial in $(\bbF[y_1,\dots y_r])[x_1,\dots, x_t]$.
		If $p$ admits a circuit representation of size $s$, then $q$ admits a circuit representation of size at most $n^t \cdot s$.
		If the circuit for $p$ is $\Gamma$-symmetric for a group $\Gamma$ that acts on the $y_i$-variables and fixes the $x_i$-variables pointwise, then so is the the circuit for $q$.
	\end{corollary}
	\begin{proof}
		\cref{lem:multivariate-polynomial-interpolation} tells us that we can express the desired $q$ as a linear combination of $n^t$ many different evaluations of $p(x_1, \dots, x_t)$. Therefore, we just need $n^t$ many copies of the circuit for $p$, where different constants are hardwired to the inputs in each copy.
		If the circuit for $p$ is symmetric under a group $\Gamma$ that acts on the $y$-variables and fixes the $x$-variables pointwise, then substituting the $x$-variables with constants does not change these symmetries, so the resulting circuit are a sum of $\Gamma$-symmetric circuits.
	\end{proof}

	\section{\textsmaller{CFI} Graphs}
	\label{sec:cfi}
	In this section, we recall the \textsmaller{CFI} graphs studied in \cite{roberson_oddomorphisms_2022} and some of their properties.
	In the literature, this construction is known as \textsmaller{CFI} graphs without internal vertices.
	Throughout, we highlight properties of bipartite \textsmaller{CFI} graphs.
	
	For a connected simple graph $G$ and $v \in V(G)$, write $E(v) \subseteq E(G)$ for the set of edges incident to~$v$.
	
	\begin{definition} \label{def:cfi}
		For a connected simple graph $G$ and a vector $U \in \mathbb{Z}_2^{V(G)}$, define the \textsmaller{CFI} graph $G_U$ via
		\begin{align*}
			V(G_U) &\coloneqq \left\{(v, T) \ \middle| \ v \in V(G), T \in \mathbb{Z}_2^{E(v)}, \sum_{e \in E(v)} T(e) = U(v) \right\}, \\
			E(G_U) &\coloneqq \left\{(v, T)(u, S) \ \middle|\ vu \in E(G), T(uv) + S(uv) = 0 \right\}.
		\end{align*}
	\end{definition}
	The set of vertices of the form $(v,T) \in V(G_U)$ associated with $v \in V(G)$ is also called the \emph{CFI gadget} for $v$.
	We write $\gamma(G) \coloneqq \sum_{v\in V(G)}2^{\deg_G(v) -1}$ for the size of $G_U$.
	Even though we have defined one graph $G_U$ for every vector $U \in \mathbb{Z}_2^{V(G)}$, there are in fact only two such graphs up to isomorphism.
	\begin{lemma}[{\cite[Corollary~3.7]{roberson_oddomorphisms_2022}}] \label{lem:cfi-hom-base-graph}
		For a connected simple graph $G$ and a vector $U \in \mathbb{Z}_2^{V(G)}$, the following are equivalent:
		\begin{enumerate}
			\item $\sum_{v \in V(G)} U(v) = 0$,
			\item $G_0 \cong G_U$,
			\item $\hom(G, G_0) = \hom(G, G_U)$.
		\end{enumerate}
	\end{lemma}
	
	By virtue of \cref{lem:cfi-hom-base-graph}, we define $G_1$ as any of the graph $G_U$ for $\sum_{v\in V(G)} U(v) = 1$.
	The graphs $G_0$ and $G_1$ are the \emph{even} and the \emph{odd \textsmaller{CFI} graphs} of $G$.
	
	In \cite{roberson_oddomorphisms_2022}, a characterisation of the graphs $F$ was given that are such that $\hom(F, G_0) = \hom(F, G_1)$ for some connected graph $G$.
	We recall this characterisation for the purpose of proving \cref{lem:non-uniform-closure}.
	
	\begin{definition}[{\cite[Definition~3.9]{roberson_oddomorphisms_2022}}]\label{def:oddomorphism}
		Let $F$ and $G$ be simple graphs and $\phi \colon F \to G$ be a homomorphism.
		A vertex $a \in V(F)$ is \emph{$\phi$-even} / \emph{$\phi$-odd} if $|N_F(a) \cap \phi^{-1}(v)|$ is even / odd for all $v \in N_G(\phi(a))$.
		The map $\phi$ is an \emph{oddomorphism} if 
		\begin{enumerate}
			\item every vertex $a \in V(F)$ is $\phi$-odd or $\phi$-even and
			\item every fibre $\phi^{-1}(v) \subseteq V(F)$ for $v \in V(G)$ contains an odd number of $\phi$-odd vertices.
		\end{enumerate}
		The map $\phi$ is a \emph{weak oddomorphism} if there exists a subgraph $F' \subseteq F$ such that $\phi|_{F'} \colon F' \to G$ is an oddomorphism.
	\end{definition}

	For example, the identity map $G \to G$ is an oddomorphism.
	
	\begin{theorem}[{\cite[Theorem~3.13]{roberson_oddomorphisms_2022}}]\label{thm:rob3.13}
		Let $G$ be a connected simple graph.
		For every simple graph $F$,
		\[
			\hom(F, G_0) \geq \hom(F, G_1)
		\]
		with strict inequality if, and only if, there exists a weak oddomorphism $F \to G$.
	\end{theorem}

	The details of \cref{def:oddomorphism} are irrelevant for this article.
	In the following observation, we state some straightforward properties of weak oddomorphisms.
	\begin{observation} \label{obs:oddo-surjective}
		Let $F$ and $G$ be simple graphs. Every weak oddomorphism $\phi \colon F\to G$ is surjective on vertices and edges.
	\end{observation}
	\begin{proof}
		For every $v \in V(G)$, $\phi^{-1}(v)$ contains an odd number of $\phi$-odd vertices. Hence, $\phi^{-1}(v)$ is non-empty and $\phi$ is surjective on vertices.
		Let $uv \in E(G)$ be an edge and $a \in \phi^{-1}(v)$ be some $\phi$-odd vertex. 
		Then $a$ has an odd number of neighbours in $\phi^{-1}(u)$.
		In particular, there is a neighbour $b \in
                \phi^{-1}(u)$ of $a$ and $\phi(ab) = uv$.
		Hence, $\phi$ is surjective on edges.
	\end{proof}
	
	\begin{theorem}[{\cite[Lemma~12, Corollary~13]{neuen_homomorphism-distinguishing_2024}}]\label{thm:neuen}
		\begin{enumerate}
			\item Let $G$ be a connected simple graph and $k \in \mathbb{N}$.
			If $\tw(G) \geq k$, then $G_0$ and $G_1$ are $\mathcal{C}^k$-equivalent.
			\item For all simple graphs $F$ and $G$ for
                          which there exists a weak oddomorphism $F
                          \to G$, we have $\tw(F) \geq \tw(G)$.
		\end{enumerate}
	\end{theorem}

	\subsection{Bipartite \textsmaller{CFI} Graphs}
	
	In this section, the results stated above are adjusted for bipartite graphs with fixed bipartition.
	Recall \cref{def:cfi}.
	The map $\rho \colon G_U \to G$ given by $(v, T) \mapsto v$ is a homomorphism.
	Hence, if $G$ is bipartite,
	then so is $G_U$.
	If the fixed bipartition of $G$ is $A \uplus B$,
	then we fix the bipartition $\rho^{-1}(A) \uplus \rho^{-1}(B)$ of~$G_U$.
	We write $\gamma_A(G) \coloneqq \sum_{v\in A}2^{\deg_G(v) -1}$ and $\gamma_B(G) \coloneqq \sum_{v\in B}2^{\deg_G(v) -1}$ for the size of left and the right side of the bipartition of $G_U$.
	
	\begin{corollary}\label{thm:rob3.13-bipartite}
		Let $G$ be a connected simple bipartite graph.
		For every simple bipartite graph $F$,
		\[
		\hom(F, G_0) \geq \hom(F, G_1)
		\]
		with strict inequality if, and only if, there exists a bipartition-respecting weak oddomorphism $F \to G$.
	\end{corollary}
	Note that the identity map $G \to G$ is a bipartition-respecting weak oddomorphism. Hence, $\hom(G, G_0) > \hom(G, G_1)$.
	\begin{proof}
		For a homomorphism $\psi \colon F\to G$, write $\hom_\psi(F, G_i)$ for the number of homomorphisms $h \colon F \to G_i$ such that $\rho h = \psi$.
		Note that
		\[
			\hom(F, G_i) = \sum_{\substack{\psi \colon F \to G \\ \psi \text{ respects bipartitions}}} \hom_\psi(F, G_i). 
		\]
		By \cite[Theorem~3.6]{roberson_oddomorphisms_2022},
		$\hom_\psi(F, G_0) \geq \hom_\psi(F, G_1)$ for every homomorphism $\psi \colon F \to G$.
		This implies that $\hom(F, G_0) \geq \hom(F, G_1)$ where $F$, $G_0$, and $G_i$ are regarded as graphs with fixed bipartition.
		
		It holds that $\hom(F, G_0) > \hom(F, G_1)$ if, and only if, there exists a bipartition-respecting homomorphism $\psi \colon F\to G$ such that $\hom_\psi(F, G_0) > \hom_\psi(F, G_1)$.
		By the proof of \cite[Theorem~3.13]{roberson_oddomorphisms_2022}, $\hom_\psi(F, G_0) > \hom_\psi(F, G_1)$ if, and only if, $\psi$ is a weak oddomorphism.
	\end{proof}

	\begin{corollary}\label{cor:sub-cfi}
		For every connected simple bipartite graph $G$,
		$\emb(G, G_0) > \emb(G, G_1)$.
	\end{corollary}
	\begin{proof}
		Write $A \uplus B$ for the fixed bipartition of $G$.
		By \cref{thm:sub-hom}, for $i \in \{0,1\}$,
		\[
		\emb(G, G_i) = \sum_{\pi \in \Pi(A)} \sum_{\sigma \in \Pi(B)} \mu_\pi \mu_\sigma \hom(G/(\pi, \sigma), G_i).
		\]
		Here, $G/(\pi, \sigma)$ can be assumed to be simple by discarding possible multiedges.
		The number of vertices in $G/(\pi, \sigma)$ is $|A/\pi|$ on the left side and $|B/\sigma|$ on the right side.
		By \cref{thm:rob3.13-bipartite,obs:oddo-surjective},
		\begin{align*}
			\emb(G, G_0) - \emb(G, G_1) 
			&= \mu_\bot \sigma_\bot \left( \hom(G/(\bot, \bot), G_0) - \hom(G/(\bot, \bot), G_1) \right) \\
			&= \hom(G, G_0) - \hom(G, G_1)
		\end{align*}
		where $\bot$ denotes the discrete partitions of $A$ and $B$.
	\end{proof}

	\begin{corollary}\label{thm:neuen-bipartite}
		Let $G$ be a connected simple bipartite graph and $k \in \mathbb{N}$.
		If $\tw(G) \geq k$, then $G_0$ and $G_1$ are $\mathcal{C}^k$-equivalent as bipartite graphs with fixed bipartition.
	\end{corollary}
	\begin{proof}
		Note that the Duplicator strategy constructed in \cite[Lemmas~11 and~12]{neuen_homomorphism-distinguishing_2024} respects the gadgets of the \textsmaller{CFI} graphs.
		Thus, this strategy allows Duplicator to win even when the bipartitions are fixed.
	\end{proof}

	\section{Hereditary Treewidth, Vertex Cover Number, and Matching Number}
	\label{sec:hdtw}
	In this section, we show that the three graph parameters hereditary treewidth, vertex cover number, and matching number are functionally equivalent.
	That is, we prove \cref{lem:vc-mn-hdtw}.

	A \emph{matching} of $F$ is a set $M \subseteq E(F)$ such that for all distinct $e_1, e_2 \in M$ it holds that $e_1 \cap e_2 = \emptyset$.
	Define the \emph{matching number} $\mn(F)$ of $F$ as the size of the largest matching in $F$.
	\begin{theorem} \label{thm:mn-hdtw}
		For every graph $F$,
		$\frac12 \hdtw(F) \leq \mn(F) \leq \binom{\hdtw(F) + 2}{2}$.
	\end{theorem}
	
	We first consider the upper bound:
	
	\begin{lemma}
		For every graph $F$,
		$\mn(F) \leq \binom{\hdtw(F)+2}{2}$.
	\end{lemma}
	\begin{proof}
		Let $F$ be a graph.
		By \cite[Fact~3.4]{curticapean_homomorphisms_2017},
		if $F'$ is a graph with at most $\mn( F)$ edges and no isolated vertices then $\tw(F') \leq \hdtw(F)$.
		Let $k \in \mathbb{N}$ be such that $\binom{k+2}{2} \geq \mn(F) \geq \binom{k+1}{2}$.
		Then $F'$ can be taken to be the $(k+1)$-clique, which has treewidth $k$.
		Hence, $\hdtw(F) \geq k$.
		Thus, $\binom{\hdtw(F) +2}{2} \geq \binom{k+2}{2} \geq \mn(F)$.
	\end{proof}
	
	By \cite[220]{curticapean_homomorphisms_2017}, one gets an asymptotic linear upper bound on $\mn(F)$ in terms of $\hdtw(F)$.
	
	\begin{lemma}
		For every graph $F$,
		$\hdtw(F) \leq 2\mn(F)$.
	\end{lemma}
	\begin{proof}
		We first show that $\tw(F) \leq 2\mn(F)$ by constructing a tree decomposition of width $\mn(F)$ from a maximum matching.
		
		Let $M$ be a maximum matching in $F$, i.e.\ a matching of cardinality $\mn(F)$.
		Write $V \coloneqq \bigcup_{e \in M} e$ for the set of vertices incident to matched edges.
		Every edge $e \in V(F)$ is incident to at least one vertex in $V$.
		Let $W \coloneqq V(F) \setminus V$ denote the set of vertices not incident to any matched edge.
		
		Define a tree decomposition $\beta$ over the star graph with centre $x$ and tips indexed by $w \in W$ as follows.
		Let $\beta(x) \coloneqq V$ and $\beta(w) \coloneqq \{w\} \cup V$.
		Since no two vertices in $W$ are adjacent, this is a valid tree decomposition of width $2 \mn(F)$.
		
		Now consider a surjective homomorphism $h \colon F \to F'$.
		Let $V' \coloneqq h(V) \subseteq V(F')$ and write $W' \coloneqq h(W) \setminus V \subseteq V(F')$.
		Define a tree decomposition $\beta'$ for $F'$ over the star graph with centre $x'$ and tips indexed by $w' \in W'$.
		Let $\beta'(x') \coloneqq V'$ and $\beta(w') \coloneqq \{w'\} \cup V'$.
		Every bag is of size at most $2\mn(F) + 1$ and the bags cover the entire graph $F'$ as $h$ is an edge- and vertex-surjective homomorphism.
		It is a valid decomposition because all $w'_1 \neq w'_2$ in $W'$ are non-adjacent since their preimages under $h$ are non-adjacent.
	\end{proof}
	
	%	\Cref{lem:tw-mn} can be strengthened as follows:
	%	
	%	\begin{lemma}
		%		For every graph $F$,
		%		$\hdtw(F) \leq 2\mn(F)$.
		%	\end{lemma}
	%	\begin{proof}
		%		Let 
		%		We show that the width-$2\mn(F)$ tree decomposition $(T, \beta)$ of $F$ constructed in \cref{lem:tw-mn} induces a tree decomposition of the same width of $F'$.
		%		To that end, define $\gamma \colon V(T) \to 2^{V(F')}$ via $\gamma(t) \coloneqq \{h(v) \mid v \in \beta(t)\}$.
		%		Since $h$ is surjective on vertices on edges, the decomposition $(T, \gamma)$ covers all vertices on edges.
		%		However, it may fall short of being a valid tree decomposition since vertices $w \neq w'$  in $W$, as defined in the proof of \cref{lem:tw-mn},
		%		may be mapped to the same vertex under $h$.
		%		In this case, the two corresponding bags have the same image and can be identified. \todo{Why do they have the same image? Couldn't the neighbours of $w$ and $w'$ in $F$ still be mapped to different vertices by $h$?}
		%		Hence, we may obtain a tree decomposition of same width.
		%	\end{proof}

	A \emph{vertex cover} is a set $C \subseteq V(F)$ such that every edge $e \in E(F)$ is incident to $C$, i.e.\ $e \cap C \neq \emptyset$.
	Write $\vc(F)$ for the size of the smallest vertex cover in $F$.
	
	\begin{fact} \label{fact:mn-vc}
		For every graph $F$, $\mn(F) \leq \vc(F) \leq 2\mn(F)$.
	\end{fact}
	
	\Cref{thm:mn-hdtw,fact:mn-vc} yield \cref{lem:vc-mn-hdtw}.
	
	\section{The Non-Uniform Homomorphism Distinguishing Closure}
	\label{app:cl}
	
	In this appendix, we prove some further properties of the non-uniform homomorphism distinguishing closure and conduct a comparison to the homomorphism distinguishing closure introduced by \textcite{roberson_oddomorphisms_2022}.
	
	Firstly note that it does not matter if the definition of $\cl_n(\mathcal{F})$ is replaced with the following:
	\[
	\cl_{\leq n}(\mathcal{F}) \coloneqq \{F \mid \forall G, H \text{ on $\leq n$ vertices}.\ G \equiv_{\mathcal{F}} H \implies \hom(F, G) = \hom(F, H)\}.
	\]
	
	\begin{lemma}\label{lem:cl-leq}
		For $\mathcal{F}$ as in \cref{proviso} and all $n \in \mathbb{N}$,
		$\cl_n(\mathcal{F}) = \cl_{\leq n}(\mathcal{F})$.
	\end{lemma}
	\begin{proof}
		Clearly,
		$\cl_{\leq n}(\mathcal{F}) \subseteq \cl_n(\mathcal{F})$.
		For the converse direction,
		let $G$ and $H$ be graphs on at most $n$ vertices.
		Suppose that $G \equiv_{\mathcal{F}} H$.
		By \cref{proviso}, $K_1 \in \mathcal{F}$ and hence $G$ and $H$ both have $m \leq n$ vertices.
		By \cite[Theorem~7]{seppelt_logical_2024},
		$G + (m-n)K_1 \equiv_{\mathcal{F}} H + (m-n) K_1$.
		Let $F \in \cl_n(\mathcal{F})$
		and write $F = F' + \ell K_1$ where $F'$ does not contain isolated vertices and $\ell \geq 0$.
		Then
		\begin{align*}
			\hom(F, G) 
			&= \hom(F' + \ell K_1, G)\\
			&= \hom(F', G) \cdot m^\ell \\
			&= \hom(F', G + (m-n)K_1) \cdot m^\ell \\
			&= \hom(F, G + (m-n)K_1) \cdot \frac{m^\ell}{n^\ell}
 		\end{align*}
 		Hence, $\hom(F, G) = \hom(F, H)$, i.e.\
 		$F \in \cl_{\leq n}(\mathcal{F})$.
	\end{proof}
	
	\begin{corollary}
		For $\mathcal{F}$ as in \cref{proviso} and all $n \in \mathbb{N}$,
		$\cl_{n+1}(\mathcal{F}) \subseteq \cl_n(\mathcal{F})$.
	\end{corollary}
	\begin{proof}
		By \cref{lem:cl-leq},
		$\cl_{n+1}(\mathcal{F}) = \cl_{\leq n+1}(\mathcal{F}) \subseteq \cl_{\leq n}(\mathcal{F}) = \cl_n(\mathcal{F})$
		where the middle inclusion is immediate from the definition.
	\end{proof}

	\begin{lemma}
		For $\mathcal{F}$ as in \cref{proviso},
		$\cl(\mathcal{F}) = \bigcap_{n \in \mathbb{N}} \cl_n(\mathcal{F})$. 
	\end{lemma}
	\begin{proof}
		Clearly, $\cl(\mathcal{F}) \subseteq \bigcap_{n \in \mathbb{N}} \cl_n(\mathcal{F})$.
		For the converse inclusion,
		let $F \in \bigcap_{n \in \mathbb{N}} \cl_n(\mathcal{F})$ and $G$ and $H$ such that $G \equiv_{\mathcal{F}} H$.
		Since $K_1 \in \mathcal{F}$,
		the graphs $G$ and $H$ have the same number of vertices $n$, say.
		Since $F \in  \cl_n(\mathcal{F})$,
		$G$ and $H$ have the same number of homomorphisms from $F$.
		Hence, $F \in \cl(\mathcal{F})$. 
	\end{proof}

\end{document}